\documentclass[a4paper,UKenglish,cleveref]{lipics-v2021}

\usepackage{enumitem} 
\setlist{nolistsep}
\usepackage{stackengine}
\usepackage{xspace} 
\usepackage{xcolor}
\usepackage{stmaryrd}
\usepackage[all]{xy}     
\usepackage[vskip=3pt]{quoting}

\DeclareFontFamily{OT1}{pzc}{}
\DeclareFontShape{OT1}{pzc}{m}{it}{<->s*[1.30]pzcmi7t}{}
\DeclareMathAlphabet{\mathpzc}{OT1}{pzc}{m}{it}
\def\ct#1{\emath{\mathpzc{#1}}}

\title{Monadic type-and-effect soundness} %TODO Please add

\author{Francesco Dagnino}{DIBRIS, Universit\`a di Genova, Italy}{francesco.dagnino@dibris.unige.it}{https://orcid.org/0000-0003-3599-3535}{}

\author{Paola Giannini}{DiSSTE, Universit\`a del Piemonte Orientale, Italy}{paola.giannini@uniupo.it}{https://orcid.org/0000-0003-2239-9529}{}

\author{Elena Zucca}{DIBRIS, Universit\`a di Genova, Italy}{elena.zucca@unige.it}{https://orcid.org/0000-0002-6833-6470}{}
\authorrunning{F. Dagnino, P. Giannini, and E. Zucca} %TODO mandatory. First: Use abbreviated first/middle names. Second (only in severe cases): Use first author plus 'et al.'

\Copyright{Francesco Dagnino, Paola Giannini, and Elena Zucca} %TODO mandatory, please use full first names. LIPIcs license is "CC-BY";  http://creativecommons.org/licenses/by/3.0/


\nolinenumbers %uncomment to disable line numbering

\EventEditors{John Q. Open and Joan R. Access}
\EventNoEds{2}
\EventLongTitle{42nd Conference on Very Important Topics (CVIT 2016)}
\EventShortTitle{CVIT 2016}
\EventAcronym{CVIT}
\EventYear{2016}
\EventDate{December 24--27, 2016}
\EventLocation{Little Whinging, United Kingdom}
\EventLogo{}
\SeriesVolume{42}
\ArticleNo{23}
\newcommand{\Space}{\hskip 0,7em}
\newcommand{\BigSpace}{\hskip 1,5em}
\newcommand{\HugeSpace}{\hskip 3,5em}
\newcommand{\refToRule}[1]{\textsc{\small (#1)}}
\newcommand{\refItem}[2]{\cref{#1}(\ref{#1:#2})} 
\newcommand{\proofcase}[1]{\textit{#1}.}
\newenvironment{proofOf}[1]{\begin{proof}[Proof of \cref{#1}]}{\end{proof}} 
\newcommand{\emath}[1]{\ensuremath{#1}\xspace}
\newcommand{\ple}[1]{\emath{{\langle #1 \rangle}}}
\newcommand{\Pair}[2]{\ple{#1,#2}}

\newcommand{\fun}[3]{\emath{#1{:}\,#2 \rightarrow #3}} 
\newcommand{\funtype}[2]{\emath{#1\rightarrow #2}} 
\newcommand{\pfun}[3]{\emath{#1{:} #2 \rightharpoonup   #3}}

\newcommand{\N}{\mathbb{N}} 
\newcommand{\Subst}[3]   {#1[#2/#3]}

\newcommand{\NamedRule}[4]{\scriptstyle{\textsc{(#1)}}\
\displaystyle                  %  #1 = nome regola
\frac{#2}{#3}         %  #2 = premesse (modo math) %  #3 = conseguenza (modo math)
\begin{array}{l}
#4     %  #4 = side conditions (modo math)
\end{array}
}
\newenvironment{grammatica}{\begin{array}{lcll}}{\end{array}}
\newcommand{\produzione}[3]{#1&::=&#2&\mbox{#3}}

\newcommand{\kw}[1]{\texttt{#1}}
\newcommand{\aux}[1]{\mathsf{#1}}%stile per funzioni ausiliarie
\newcommand{\mvar}[1]{\mathit{#1}}
\newcommand{\produzioneinline}[2]{#1::=#2}
\newcommand{\seq}[1]{\overline{#1}} 

\newcommand{\nt}[3]{\emath{#1 : #2 \Rightarrow #3}} 
\def\op{^{\mbox{\normalfont\scriptsize op}}}
\newcommand{\Set}{\ct{Set}} 
\newcommand{\Pos}{\ct{Pos}} 
\newcommand{\Id}{\mathsf{Id}} 
\newcommand{\id}{\mathsf{id}} 
\newcommand{\mnd}{\Mnd\mfun}
\newcommand{\mfun}{M}
\newcommand{\mun}{\eta}
\newcommand{\mmul}{\mu}
\newcommand{\mkl}[1]{{#1}^\dagger}
\newcommand{\mklind}[2]{{#1}^\dagger_{#2}}
\newcommand{\mbind}{\mathbin{\gg=}} 
\newcommand{\Mmap}[3][]{\aux{map}\ifblank{#1}{}{_{#1}}\, #2\, #3}
\newcommand{\Mnd}[1]{\mathbb{#1}}
\newcommand{\Mun}[1]{\mun^{\Mnd{#1}}}
\newcommand{\Mmul}[1]{\mmul^{\Mnd{#1}}}

\newcommand{\ExSet}{\aux{Exc}}
\newcommand{\ExceptFun}[1][]{E\ifblank{#1}{}{_{#1}}}
\newcommand{\exc}{\aux{e}}
\newcommand{\excset}{\aux{E}}
\newcommand{\noEx}{\aux{none}} 

\newcommand{\PowerFun}{P}
\newcommand{\ListFun}{L} 
\newcommand{\elist}{\epsilon}
\newcommand{\cons}{\colon} 
\newcommand{\AllLift}[1][]{\mathbf{\forall}\ifblank{#1}{}{^{#1}}}
\newcommand{\ExLift}[1][]{\mathbf{\exists}\ifblank{#1}{}{^{#1}}} 

\newcommand{\DistFun}{D}
\newcommand{\Supp}{\aux{supp}} 

\newcommand{\StateFun}[1][]{S\ifblank{#1}{}{_{#1}}} 
\newcommand{\StateSet}{\aux{S}} 

\newcommand{\OutSet}{\aux{Out}}
\newcommand{\OutCons}{\cdot}
\newcommand{\OutEmpty}{\varepsilon}
\newcommand{\OutFun}{O}

\newcommand{\TreeFun}{T}
\newcommand{\tree}[2]{#1\triangleright #2} 
\newcommand{\etree}{\bot} 

\newcommand{\lang}{\mathcal{L}}
\newcommand{\Exp}{\aux{Exp}}
\newcommand{\expr}{e}
\newcommand{\Val}{\aux{Val}} 
\newcommand{\val}{v} 
\newcommand{\Res}{\aux{Res}}
\newcommand{\res}{r}
\newcommand{\Conf}{\aux{Conf}}
\newcommand{\conf}{c} 
\newcommand{\valtoexp}{\aux{ret}} 
\newcommand{\Wrng}{\aux{Wr}}
\newcommand{\wrng}{\aux{wrong}} 
\newcommand{\mexpr}{\textsc{e}} 
\newcommand{\mval}{\textsc{v}} 
\newcommand{\mres}{\textsc{r}} 
\newcommand{\mconf}{\textsc{c}} 
\newcommand{\metc}[1]{\hat{#1}}
\newcommand{\mrtc}[1]{\hat{#1}}
\newcommand{\mrtcind}[2]{\hat{#1}_#2}
\newcommand{\mvtr}[1]{\hat{#1}} 
\newcommand{\finsem}[1]{\infsem[\star]{#1}}  
\newcommand{\infsem}[2][\infty]{\llbracket #2 \rrbracket_{#1}} 
\newcommand{\ehole}{[\ ]} 
\newcommand{\red}{\to}
\newcommand{\purered}{\red_p} 
\newcommand{\mrun}[1]{\emath{\aux{run}_{#1}}} 
\newcommand{\redstar}{\red^\star} 
\newcommand{\tredfun}{\aux{step}}
\newcommand{\tredfunind}[1]{\aux{step}^{#1}}
\stackMath
\newcommand\xxrightarrow[2][]{\mathrel{%
  \setbox2=\hbox{\stackon{\scriptstyle#1}{\scriptstyle#2}}%
  \stackunder[0pt]{%
    \xrightarrow{\makebox[\dimexpr\wd2\relax]{$\scriptstyle#2$}}%
  }{%
   \scriptstyle#1\,%
  }%
}}
\newcommand{\tred}{\xxrightarrow[\tredfun]{}}
\newcommand{\tredstar}{\tred^\star}
\newcommand{\Red}{\Rightarrow}
\newcommand{\Redstar}{\Red^\star} 

\newcommand{\lambdaEff}{\Lambda_{\Sigma}}
\newcommand{\f}{\mvar{f}}
\newcommand{\x}{\mvar{x}}
\newcommand{\y}{\mvar{y}}

\newcommand{\ve}{\mvar{v}} 
\newcommand{\e}{\mvar{e}}
\newcommand{\opg}{\mvar{op}} 

\newcommand{\Fun}[2]{\lambda #1.#2} 
\newcommand{\RecFun}[3]{\kw{rec}\,#1.\Fun{#2}{#3}}
\newcommand{\App}[2]{#1\,#2} 
\newcommand{\Ret}[1]{\kw{return}\ #1}
\newcommand{\Rt}[1]{\kw{ret}\ #1}
\newcommand{\Do}[3]{\kw{do}\ #1 = #2\kw{;}\ #3} 

\newcommand{\T}{\mvar{T}}
\newcommand{\TVar}[2]{#2:#1}
\newcommand{\TVars}[2]{\seq#2{:}\seq#1}
\newcommand{\TEff}[2]{#1{!}#2}
\newcommand{\FunType}[3]{#1{\rightarrow}_{#2}#3}
\newcommand{\eff}{\mvar{E}} 

\newcommand{\EffSet}{\aux{Eff}}
\newcommand{\eZero}{\epsilon}
\newcommand{\EComp}[2]{#1{\cdot}#2}
\newcommand{\SetOf}[1]{\{#1\}}
\newcommand{\Cons}[2]{#1{:}#2}
\newcommand{\IsWFExp}[3]{#1\vdash#2:#3}
\newcommand{\IsWFGroundExp}[2]{\vdash#1:#2}
\newcommand{\IsEffExp}[4]{#1\vdash#2:\TEff{#3}{#4}}
\newcommand{\IsEffGroundExp}[3]{\vdash#1:\TEff{#2}{#3}}
\newcommand{\subt}{\leq}
\newcommand{\sube}{\subseteq}
\newcommand{\SubT}[2]{#1\subt#2}
\newcommand{\SubTE}[2]{#1\subt#2}
\newcommand{\SubE}[2]{#1\sube#2}
\newcommand{\SubH}[2]{#1\ll#2}

\newcommand{\TS}{\Theta} 
\newcommand{\Types}{\aux{Ty}}
\newcommand{\ty}{\tau}
\newcommand{\MEff}{\mathcal{E}}
\newcommand{\Eff}{\aux{Eff}}
\newcommand{\ef}{\varepsilon}
\newcommand{\eord}{\preceq}
\newcommand{\emul}{\cdot}
\newcommand{\eun}{1} 

\newcommand{\VWT}[1][]{\aux{WT}^{\aux{V}}\ifblank{#1}{}{_{#1}}}
\newcommand{\EWT}[1][]{\aux{WT}^{\aux{E}}\ifblank{#1}{}{_{#1}}}
\newcommand{\CWT}[1][]{\aux{WT}^{\aux{C}}\ifblank{#1}{}{_{#1}}}
\newcommand{\PW}{\mathcal{P}}
\newcommand{\img}[1]{\PW^{#1}} 
\newcommand{\mlift}[1][]{\lambda\ifblank{#1}{}{^{#1}}} 
\newcommand{\mord}{\sqsubseteq}
\newcommand{\msup}{\bigsqcup} 
\newcommand{\mbot}{\bot}
\newcommand{\mctr}{\aux{res}} 
\newcommand{\X}{\mvar{X}}

\newcommand{\Nat}{\texttt{Nat}}
\newcommand{\Bool}{\texttt{Bool}}
\newcommand{\Zero}{0}
\newcommand{\Succ}[1]{\texttt{succ}\, #1}
\newcommand{\Su}[1]{\texttt{s}\, #1}
\newcommand{\SuccSeq}[2]{\texttt{succ}^{#1}\, #2}
\newcommand{\nat}[1]{\hat{#1}}
\newcommand{\Pred}[1]{\texttt{pred}\, #1}
\newcommand{\IsZero}[1]{\texttt{iszero}\ #1}
\newcommand{\If}[3]{\kw{if}\ #1\ \kw{then}\ #2\ \kw{else}\ #3}
\newcommand{\True}{\texttt{true}}
\newcommand{\False}{\texttt{false}}
\newcommand{\Tru}{\texttt{t}}
\newcommand{\Fls}{\texttt{f}}
\newcommand{\Seq}[2]{#1\texttt{;}#2}
\newcommand{\Bot}{\texttt{Bot}}
\newcommand{\Unit}{\texttt{Unit}}
\newcommand{\unit}{\texttt{unit}}

\newcommand{\withindex}[2]{#1\langle#2\rangle}
 \newcommand{\raiseop}[1]{\withindex{\tt raise}{#1}}
\newcommand{\predfun}{\aux{predfun}}
\newcommand{\PredZeroExc}{\aux{PredZero}}

\newcommand{\chooseop}{{\tt choose}}
\newcommand{\List}[1]{[\,#1\,]}

\newcommand{\chooseincr}{\aux{chfun^{\uparrow}}}
\newcommand{\choosedecr}{\aux{chfun^{\downarrow}}}

\newcommand{\writeop}[1]{\withindex{\tt write}{#1}}
\newcommand{\loc}{\ell}
\newcommand{\LocSet}{\aux{Loc}}
\newcommand{\writedecr}{\aux{wfun^{\downarrow}}}
\newcommand{\writeincr}{\aux{wfun^{\uparrow}}}
\newcommand{\extract}[1]{\aux{extract}(#1)}
\newcommand{\Even}[1]{\aux{even}(#1)}

\newcommand{\effsem}[1]{\llbracket #1 \rrbracket}

\newcommand{\hc}{\mvar{c}}
\newcommand{\handler}{\mvar{h}}
\newcommand{\Handler}[3]{#1, #2\mapsto#3} 
\newcommand{\With}[2]{\kw{handle}\ #2\ \kw{with}\  #1}
\newcommand{\WithLong}[4]{\kw{handle}\ #4\ \kw{with}\  \Handler{#1}{#2}{#3}}

\newcommand{\HC}[4]{#1(#2)\mapsto_{#4}#3} 
\newcommand{\mode}{\mu}
\newcommand{\Continue}{\aux{c}}
\newcommand{\Stop}{\aux{s}}
\newcommand{\cfilter}{\mvar{C}}
\newcommand{\CFilter}[3]{#1\mapsto_{#2}#3}
\newcommand{\hfilter}{\mvar{H}}
\newcommand{\HFilter}[2]{#1, #2} 
\newcommand{\FilterF}[1]{\mathcal{F}_{#1}}
\newcommand{\FilterFun}[2]{\FilterF{#2}(#1)}
\newcommand{\FilterFEx}[1]{\widehat{\mathcal{F}}_{#1}}
\newcommand{\FilterFunEx}[2]{\FilterFEx{#2}(#1)}
\newcommand{\InfConc}[1]{\bullet^\infty#1}
\newcommand{\IsWFHandler}[5]{#1;#2\vdash#3:\TEff{#4}{#5}}
\newcommand{\IsWFClause}[4]{#1;#2\vdash#3:#4}
\newcommand{\infEff}{\gamma}
\newcommand{\str}{\mathit{s}}
\begin{document}

\maketitle

\begin{abstract}
We introduce the abstract notions of \emph{monadic operational semantics}, a small-step semantics where computational effects are modularly modeled by a monad, and \emph{type-and-effect system}, including \emph{effect types} whose interpretation lifts well-typedness to its monadic version. In this meta-theory, as usually done in the non-monadic case, we can express progress and subject reduction properties and provide a proof, given once and for all, that they imply soundness. 

The approach is illustrated on a lambda calculus with generic effects.  We equip the calculus with an expressive type-and-effect system, and provide proofs of progress and subject reduction which are parametric on the interpretation of effect types.   In this way, we obtain as instances many significant examples, such as  checking exceptions, preventing/limiting non-determinism, constraining order/fairness of outputs on different locations.    We also provide an extension with constructs to raise and handle computational effects, which can be instantiated  to model different policies.
\end{abstract}

\maketitle

% !TEX root =main.tex

\section{Introduction}
It would be hard to overstate the impact on foundations of programming languages of, on one hand, the idea that computational effects can be modeled by monads \cite{Moggi89,Moggi91},  and, on the other hand, the technique based on progress and subject reduction to prove the soundness of a type system with respect to a small-step  operational  semantics \cite{WrightF94}.

 Moggi's seminal work \cite{Moggi89,Moggi91}  recognized monads as the suitable structure to modularly describe the denotational semantics of effectful languages. 
The key idea was the distinction between \emph{pure} (effect-free) and \emph{monadic} (effectful)  expressions, also called \emph{computations}, the latter getting semantics in a monad. 
Haskell has firstly\footnote{Many other languages have then supported a monad pattern, e.g., Scheme, Python, Racket, Scala, F\#.} shown that such  an  approach can be fruitfully  adopted in  a mainstream language,  through a monad type constructor allowing to encapsulate effectful code. 
However, the structure of a monad\footnote{In Haskell, methods of the \lstinline{Monad} typeclass.} does not include operations for \emph{raising} effects, which need to be defined ad-hoc in instances.  
\emph{Algebraic and generic effects} \cite{PlotkinP01,PlotkinP02,PlotkinP03}, instead, explicitly consider operations to raise effects, 
 interpreted by additional structure on the monad.  
Such an approach, combined with handlers \cite{PlotkinPretnar09,PlotkinPretnar13,BauerP15,Pretnar15}, has been exploited in  fully-fledged programming languages, e.g., in Scala and \mbox{OCaml 5.} 

 To provide guarantees on, besides the result, the computational effects possibly raised by a computation, 
 type systems are generalized to \emph{type-and-effect systems}.   A great many of  these  have been designed for specific calculi, modelling effects by relying on 
auxiliary structures, e.g., memory in imperative calculi, and providing ad-hoc soundness proofs; 
Katsumata \cite{Katsumata14} has provided a unified view of such systems, however based on denotational semantics.

In this paper, instead, we provide an \emph{operational meta-theory of monadic type-and-effect soundness}, 
analogous to the  one mentioned above for usual   type soundness based on small-step semantics, progress and subject reduction \cite{WrightF94}.  
To this end,  we provide abstract notions of small-step monadic semantics, type-and-effect system, and soundness, \mbox{as detailed below.}\footnote{The term ``effect'' is used in literature both as synonym of computational effect, and in the context of type-and-effect systems, as a static approximation of the former.   We will use ``effect'' when there is no ambiguity, otherwise  ``computational effect'' and ``effect type'', respectively.}

\smallskip\noindent
\textbf{Operational semantics}
We design a language semantics which is \emph{monadic}, since effects are, as customary, expressed by a monad, and simultaneously \emph{small-step}, since we define sequences of reduction steps.
To this end, we start from a reduction from language expressions to monadic expressions (in a given monad)  required to be deterministic, and extend such a relation to a \emph{total} function, so to be able to combine steps by Kleisli composition, similarly to the approach in \cite{GavazzoTV24}. 
In this way, reduction sequences are always infinite, so termination is conventionally represented by monadic elements called \emph{results},
 which always reduce to themselves without raising any effect.  
On top of the reduction, we define the \emph{finitary} semantics of an expression, which is either the monadic result reached in many steps, if any, or divergence.  This semantics does not describe the computational effects raised by infinite computations. Hence, we define an \emph{infinitary} semantics, obtained, as customary, as the supremum of a chain of approximants, provided that the monad has the necessary structure.  

\smallskip\noindent
\textbf{Type-and-effect system}
As done in \cite{DagninoBZD20,Dagnino22} for standard type systems, 
we abstractly model a \emph{type-and-effect system} as a family of predicates over expressions, indexed by types and \emph{effect types}, statically approximating the computational effects that  may  be  raised during evaluation. 
Effect types are required to form an ordered monoid, as typically assumed in effect systems \cite{NielsonN99,MarinoM09}  and proposed as algebraic structure by \cite{Katsumata14}. 
The relation between an effect type and the allowed computational effects is  specified by a family of  \emph{predicate liftings} \cite{Jacobs16}. In this way the transformation from a predicate to a monadic one associated to a given effect type is independent from the predicate and its universe. In other words, the transformation can be seen as the semantics of the effect type.

\smallskip\noindent
\textbf{Soundness}
We provide abstract definitions of monadic progress and monadic subject reduction, and a proof, given once and for all, that they imply soundness. 
The latter means that, if a monadic element is the result of a well-typed expression,  then it should be well-typed,  that is, satisfy the lifting through the effect type of \mbox{well-typedness of values.}
 
We illustrate the approach on  $\lambdaEff$,   a lambda calculus with generic effects, equipped with an expressive type-and-effect system. We provide proofs of progress and subject reduction parametrically on the interpretation of effect types.  In this way, we obtain as instances many significant examples, such as  checking exceptions, preventing/limiting non-determinism, constraining order/fairness of outputs on different locations.   We also provide an extension with constructs to handle effects, which can be instantiated as well to model different policies. 

\smallskip\noindent
\textbf{Outline} 
\cref{sect:monads} reports the background on monads. \cref{sect:monadic-sem} introduces monadic operational semantics, exemplified through  $\lambdaEff$ in \cref{sect:lang}, where we also design  a type-and-effect system, discussing its soundness.
The approach is formalized by the abstract framework in \cref{sect:monadic-ty}; the proof technique introduced there is applied in \cref{sect:results} to $\lambdaEff$. Finally, in \cref{sect:handlers} we enhance the example by handlers, and in \cref{sect:related} we discuss related and future work, and summarize the contributions. Proofs omitted from \cref{sect:results,sect:handlers} can be found in \cref{app:lambda-results,app:handlers-results}.

% !TEX root =main.tex

\section{Preliminaries on monads}\label{sect:monads}
Monads \cite{EilenbergM65,Street72} are a fundamental notion in category theory, enabling an abstract and unified study of algebraic structures. 
Since Moggi's seminal papers \cite{Moggi89,Moggi91}, they have also  become a major tool in computer science, especially for describing the semantics of computational effects, and integrating them in programming languages in a structured and principled way. 
In this section, we recall basic notions about monads, and provide some examples. 
We will focus on monads on the category of sets and functions, denoted by \Set, 
referring the reader to standard textbooks \cite{Riehl17} for a detailed introduction in full generality. 

A \emph{monad} $\mnd = \ple{\mfun,\mun,\mmul}$ (on \Set) consists of 
a functor \fun{\mfun}{\Set}{\Set} and two natural transformations \nt{\mun}{\Id}{\mfun} and \nt{\mmul}{\mfun^2}{\mfun} 
such that, for every set $X$, the following diagrams commute: 
\begin{quoting}
\begin{math}\vcenter{\xymatrix{
  \mfun X \ar[r]^-{\mun_{\mfun X}} \ar[rd]_-{\id_{\mfun X}} 
& \mfun^2 X \ar[d]_-{\mmul_X} 
& \mfun X \ar[l]_-{\mfun\mun_X} \ar[ld]^-{\id_{\mfun X}} 
\\ 
& \mfun X 
}} \qquad 
\vcenter{\xymatrix{
  \mfun^3 X \ar[r]^-{\mfun\mmul_X} \ar[d]_-{\mmul_{\mfun X}} 
& \mfun^2 X \ar[d]^-{\mmul_X} 
\\
  \mfun^2 X \ar[r]^-{\mmul_X}
& \mfun X 
}}
\end{math}
\end{quoting}
The functor $\mfun$ specifies, for every set $X$, a set $\mfun X$ of monadic elements built over $X$, in a way  that is  compatible with functions. 
The map $\mun_X$, named \emph{unit}, embeds elements of $X$ into monadic elements in $\mfun X$, and 
the map $\mmul_X$, named \emph{multiplication}, flattens monadic elements built on top of other monadic elements into plain monadic elements.

From these data, one can derive an operation on functions of type $X\to\mfun Y$, dubbed \emph{Kleisli extension}, which is crucial for modelling computational 
effects using monads. 
For all sets $X,Y$, we have a function 
\fun{\mkl{(-)}}{(X\to\mfun Y)}{(\mfun X \to \mfun Y)}, defined by 
$\mkl{f} = \mmul_Y \circ \mfun f$, that is, 
first we lift $f$ through $\mfun$ to apply it to monadic elements and then we flatten the result using $\mmul_Y$. 
It is easy to see that the operation $\mkl{(-)}$ satisfies the following equations for all \fun{f}{X}{\mfun Y} and \fun{g}{Y}{\mfun Z}:
\begin{quoting}
\begin{math}
\mklind{\mun}{X} = \id_{\mfun X} \qquad 
\mkl{f}\circ \mun_X = f \qquad 
\mkl{g}\circ \mkl{f} = \mkl{ (\mkl{g}\circ {f}) } 
\end{math}
\end{quoting}
Actually, a monad can be equivalently specified in the form of a \emph{Kleisli triple} \ple{\mfun,\mun,\mkl{(-)}} \cite{Manes76}, where 
$\mfun$ is a mapping on sets, 
$\mun$ is a family of functions \fun{\mun_X}{X}{\mfun X}, for every set $X$, 
and $\mkl{(-)}$ is a family of functions \fun{\mkl{(-)}}{(X\to\mfun Y)}{(\mfun X \to \mfun Y)}, for all sets $X,Y$, satisfying the three equations above. 
In particular we have $\mmul_X = \mklind{\id}{\mfun X}$. 

Functions of type $X \to \mfun Y$ are called \emph{Kleisli functions} and play a special role: 
they can be regarded as ``effectful functions'' from $X$ to $Y$, raising effects described by the monad $\mnd$. 
Indeed, from the Kleisli extension, we can define a composition on Kleisli functions, known as Kleisli composition: 
given \fun{f}{X}{\mfun Y} and \fun{g}{Y}{\mfun Z} we set 
\begin{quoting}
$g \ast f = \mkl{g}\circ f = \mmul_Z \circ \mfun g \circ f$
\end{quoting} 
Intuitively, $g\ast f$ applies $f$ followed by $g$, sequentially composing the effects they may raise. 
It is immediate to see that Kleisli composition is associative and  $\mun_X$ is the identity Kleisli function on the set $X$, that is, $\mun_X$ is the function raising no effects.

We introduce some useful notation, corresponding to standard operations of monadic types in languages, where such types are assigned to expressions with effects.
Given $\alpha\in \mfun X$, \fun{f}{X}{\mfun Y} and \fun{g}{X}{Y}, we set 

\begin{quoting}
\begin{math}
\begin{array}{ll}
\fun{-\mbind-}{\mfun X}{(X\to\mfun Y)\to\mfun Y}&\alpha \mbind f = \mkl{f}(\alpha)\\
\fun{\aux{map}}{(X\to Y)}{\mfun X \to \mfun Y}&\Mmap{g}{\alpha} = \mfun g(\alpha)  
\end{array}
\end{math}
\end{quoting}
The operator $\mbind$ is also called $\aux{bind}$. As its definition shows, it can be seen as an alternative description of the Kleisli extension, where the parameters are taken in inverse order.  
This view corresponds, intuitively, to the sequential composition of two expressions with effects, where the latter  depends on a parameter \emph{bound} to the result of the former. 
%More precisely, $\alpha \mbind f$ gives $\alpha$ followed by  (the Kleisli extension of) $f$, sequentially composing the effects they may raise. 
The operator $\aux{map}$ describes the effect of the functor $\mfun$ on functions. That is, the lifting of function $g$ through $\mfun$ is applied to a monadic value $\alpha$. 
Note that $\aux{bind}$ and $\aux{map}$ are interdefinable: 
\begin{quoting}
${\alpha \mbind f = \mmul_Y(\Mmap{f}{\alpha})}$\HugeSpace
${\Mmap{g}{\alpha} = \alpha \mbind (\mun_Y \circ g)}$
\end{quoting}
Furthermore, we can express Kleisli composition using $\aux{bind}$: 
$(g\ast f)(x) = f(x) \mbind g$. 
%Finally, we mention that the multiplication $\mmul$ is named $\aux{flatten}$ or $\aux{join}$ in programming languages. 
%

In the following examples we characterize the monads by defining $\aux{bind}$ rather than multiplication $\mmul$ since this is often more insightful, and customary in programming languages.   

\begin{example}[Exceptions] \label{ex:exception-mnd} 
Let us fix a set $\ExSet$. 
The monad $\Mnd{\ExceptFun[\ExSet]} = \ple{\ExceptFun[\ExSet],\Mun{\ExceptFun[\ExSet]},\Mmul{\ExceptFun[\ExSet]}}$ is given by 
$\ExceptFun[\ExSet] X = \ExSet + X$, and 
\begin{quoting}
\begin{math}
\Mun{\ExceptFun[\ExSet]}(x) = \iota_2(x) 
\HugeSpace
\alpha \mbind f = \begin{cases}
f(x)\ \mbox{if}\ \alpha = \iota_1(x) \\
\alpha\ \text{otherwise ($\alpha = \iota_2(\exc)$ for some $\exc\in \ExSet$)} 
\end{cases}  
\end{math} 
\end{quoting}
where $+$ denotes disjoint union (coproduct) and $\iota_1,\iota_2$ the left and right injections, respectively. 
We will omit the reference to the set $\ExSet$ when it is clear from the context. 
% This monad induces a monad transformer: 
% for every monad $\mnd = \ple{\mfun,\mun,\mmul}$, the mapping on sets 
% $X \mapsto \mfun(X + \ExSet)$ is part of a monad $\Mnd{\mfun\ExceptFun}$. 
\end{example}

\begin{example}[Classical Non-Determinism]\label{ex:pow-mnd}
The monad $\Mnd\PowerFun = \ple{\PowerFun,\Mun\PowerFun,\Mmul\PowerFun}$ is given by 
$\PowerFun X = \wp(X)$, that is, $\PowerFun X$ is the set of all subsets of $X$, and 
\begin{quoting}
\begin{math}
\Mun\PowerFun(x) =\{x\} \HugeSpace
\alpha \mbind f  = \bigcup_{x\in\alpha} f(x) 
\end{math} 
\end{quoting}
A variant of this monad is 
the list monad $\Mnd\ListFun = \ple{\ListFun,\Mun\ListFun, \Mmul\ListFun}$, where 
the set $\ListFun X$ of (possibly infinite) lists over $X$ is coinductively defined by the following rules: 
$\elist\in\ListFun(X)$ and, 
if $x\in X$ and $l\in\ListFun(X)$, then $x\cons l \in\ListFun(X)$. 
We use the notation $[x_1,\ldots,x_n]$ to denote the finite list $x_1\cons \ldots\cons x_n\cons\elist$. 
Then, the unit is given by 
$\Mun\ListFun_X(x) = [x]$ and the monadic bind is corecursively defined by the following clauses: 
$\elist\mbind f = \elist$ and 
$(x\cons l)\mbind f = f(x)(l\mbind f)$, 
where juxtaposition denotes the concatenation of possibly infinite lists. 
\end{example} 

\begin{example}[Probabilistic Non-Determinism]\label{ex:prob-mnd}
Denote by 
$\DistFun X$ the set of probability subdistributions $\alpha$ over $X$ with countable support, i.e.,
\fun{\alpha}{X}{[0..1]} with $\sum_{x\in X}\alpha(x) \leq 1$ and $\Supp(\alpha) = \{ x \in X \mid  \alpha(x)\ne 0 \}$ countable set. 
We write $r\cdot\alpha$ for the pointwise multiplication of a subdistribution $\alpha$ with a number $r\in[0,1]$. 
The monad $\Mnd\DistFun = \ple{\DistFun,\Mun\DistFun,\Mmul\DistFun}$ is given by 
\begin{quoting}
\begin{math}
\Mun\DistFun(x) = y \mapsto \begin{cases}
  1 & y = x \\
  0 & \text{otherwise}
\end{cases} \HugeSpace
\alpha \mbind f = \sum_{x\in X} \alpha(x)\cdot f(x) 
\end{math}
\end{quoting}
\end{example}

\begin{example}[Output/Writer]
\label{ex:output-mnd} 
Let \ple{\OutSet,\OutCons,\OutEmpty} be a monoid, e.g., the monoid of strings over a fixed alphabet. 
The monad $\Mnd\OutFun = \ple{\OutFun,\Mun\OutFun,\Mmul\OutFun}$ is given by 
$\OutFun X = \OutSet \times X$ and 
\begin{quoting}
\begin{math}
\Mun\OutFun(x) = \ple{\OutEmpty,x} \HugeSpace
\ple{o,x} \mbind f = \ple{o\OutCons\pi_1(f(x)), \pi_2(f(x))} 
\end{math}
\end{quoting}
% This monad induces a monad transformer: given a monad $\mnd = \ple{\mfun,\mun,\mmul}$, the mapping on sets 
% $X \mapsto \mfun(\OutSet \times X)$ is part of a monad $\Mnd{\mfun\OutFun}$. 
Combining this monad with the exception monad of \cref{ex:exception-mnd}, we obtain  the pointed output monad, whose underlying functor is given by  $\OutFun' X = \OutSet\times(X+\{\bot\})$. 
\end{example}

\begin{example}[Global State]\label{ex:state-mnd}
Let $\StateSet$ be a set of states. 
The monad $\Mnd\StateFun= \ple{\StateFun,\Mun\StateFun,\Mmul\StateFun}$ is given by 
$\StateFun X = \StateSet \to \StateSet\times X$ and 
\begin{quoting}
\begin{math}
\Mun\StateFun(x) = s \mapsto \ple{s,x} \HugeSpace
\alpha\mbind f = s \mapsto  f(\pi_2(\alpha(s))) (\pi_1(\alpha(s)))
\end{math}
\end{quoting}
% This monad induces a monad transformer: given a monad $\mnd = \ple{\mfun,\mun,\mmul}$ the mapping on sets 
% $X \mapsto \StateSet \to \mfun(\StateSet\times X)$ is part of a monad $\Mnd{\StateFun\mfun}$. 
We can combine this monad with the exception monad of \cref{ex:exception-mnd} obtaining 
$\Mnd{\StateFun[\ExSet]} = \ple{\StateFun[\ExSet], \Mun{\StateFun[\ExSet]}, \Mmul{\StateFun[\ExSet]}}$ where 
$\StateFun[\ExSet] = \StateSet \to (\StateSet\times X) + \ExSet$ and 
$\Mun{\StateFun[\ExSet]}_X(x) = s \mapsto \Mun{\ExceptFun[\ExSet]}_{\StateSet\times X}(\ple{s,x})$ and 
$\alpha\mbind f = s \mapsto (\alpha(s) \mbind_{\Mnd{\ExceptFun[\ExSet]}} (x \mapsto f(x)(s))$. 
This combination yields a monad thanks to the fact that $\Mnd\StateFun$ determines a monad transformer \cite{LiangHJ95,JaskelioffM10}. 
\end{example}

\newcommand{\xtl}{\mvar{xl}}
\newcommand{\ytl}{\mvar{yl}}

\begin{example}[Ordered Trees]
\label{ex:tree-mnd} 
The set $\TreeFun X$ of (possibly infinite) trees over a set $X$ is coinductively defined by the following rules: 
$\etree\in\TreeFun(X)$ and, 
if $\x\in X$  and $\xtl\in\ListFun(\TreeFun X)$, then $\tree{\x}{\xtl}\in\TreeFun(X)$. 
In other words, we have $\TreeFun X \cong \nu Y.1+X\times\ListFun Y$  where $\nu$ is the greatest fixed point operator. 
These sets are part of the tree monad $\Mnd\TreeFun = \ple{\TreeFun,\Mun\TreeFun,\Mmul\TreeFun}$, where 
the unit is given by $\Mun\TreeFun_X(\x) = \tree{\x}{\elist}$ and 
the Kleisli extension of a function \fun{f}{X}{\TreeFun Y} is corecursively defined  as follows
\begin{align*}
\mkl{f}(\etree)      &= \etree  \\ 
\mkl{f}(\tree{\x}{\xtl}) &= \tree{\y}{(\ytl\cdot \ytl')} \qquad \text{if $\Mmap[\ListFun]{\mkl{f}}{\xtl}=\ytl$ and $f(\x) = \tree{\y}{\ytl'}$} 
\end{align*} 
That is, given a tree over $X$ with root $\x$ and children $\xtl$, recursively mapping (through the $\aux{map}$ of lists) $\mkl{f}$ to $\xtl$ gives a list $\ytl$ of trees over $Y$, and applying $f$ to $\x$ gives a tree over $Y$ with root $\y$ and children $\ytl'$;   the final result is the tree with root $\y$ and children obtained appending $\ytl'$ to $\ytl$.
\end{example}

% !TEX root =main.tex

\section{Monadic operational semantics}
\label{sect:monadic-sem}

In this section we abstractly describe a framework for (deterministic) monadic operational semantics, adapting from \cite{GavazzoF21,GavazzoTV24}. 

\begin{definition}\label{def:mnd-sem} 
Let $\lang$ be a triple $\ple{\Exp,\Val,\valtoexp}$, called a \emph{language}, with $\Exp$ the set of \emph{expressions}, $\Val$ the set of \emph{values}, and $\fun{\valtoexp}{\Val}{\Exp}$ an injective function. 
A \emph{monadic operational semantics} for $\lang$ consists of:
\begin{itemize}
\item a monad $\mnd = \ple{\mfun,\mun,\mmul}$
\item  a  relation $\red \subseteq \Exp\times\mfun\Exp$, called \emph{monadic (one-step) reduction},  such that 
\begin{itemize}
\item $\red$ is a partial function and 
\item for all $\val\in\Val$, $\valtoexp(\val)\not\red$. 
\end{itemize}
\end{itemize}
\end{definition} 
The set $\Exp$  contains  expressions that can be executed, while $\Val$ contains values produced by the computation. 
The inclusion $\valtoexp$ identifies the expressions representing successful termination with a given value. 
The elements of $\mfun\Exp$, called \emph{monadic expressions}, are the counterpart of expressions in the monad $\mnd$. 
The relation $\red$ models single computation steps, which transform expressions into monadic ones, thus possibly raising computational effects. 
Finally, the first requirement on $\red$ ensures that it is deterministic, while the latter one that expressions representing values \mbox{cannot be reduced. }

 Assume a monadic operational semantics $\Pair{\mnd}{\red}$ for a language $\ple{\Exp,\Val,\valtoexp}$. In standard (small-step) operational semantics, starting from the one-step reduction we can model computations as (either finite of infinite) sequences of reduction steps. In particular, finite computations are obtained by the reflexive and transitive closure $\redstar$ of the one-step reduction. Starting from the \emph{monadic} one-step reduction, which is a relation from a set to a different one, there is no transitive closure in the usual sense.

In the solution proposed in \cite{GavazzoF21},  the monadic reduction can be an arbitrary relation; however, this requires a relational extension of the monad \cite{Barr70}. 
On the other hand, given a relation $\red \subseteq \Exp\times\mfun\Exp$  which is a  \emph{total} function, we can define, by  iterating Kleisli composition, a relation  $\redstar \subseteq \Exp\times\mfun\Exp$ which plays the role of transitive closure, \mbox{as in \cite{GavazzoTV24}. }

Our aim here is to define $\redstar$ taking the second approach, which does not require relational extensions. 
 Unfortunately, the monadic reduction, exactly as the standard one,  is by its own nature  a partial function, where some expressions, representing terminated computations, cannot be reduced. Notably,  those representing successful termination with a value,  and others, intuitively corresponding to stuck computations.  
 To obtain a total function, we extend the monadic reduction to \emph{configurations} (expressions, values, or  a special result $\wrng$). In particular, expressions representing  terminating  computations reduce to (the monadic embedding of) a value, and $\wrng$, respectively. In this way, we can define the transitive closure by Kleisli composition, \mbox{as formally detailed below. }

Set $\Res = \Val + \Wrng$, where $\Wrng = \{\wrng\}$, 
that is, a result $\res$ is either a value, modelling successful termination, or $\wrng$, modelling  a stuck computation. 
Then, we consider the set $\Conf = \Exp + \Res$ of \emph{configurations}, ranged over by $\conf$. 
We have the following commutative diagram of  coproduct injections:
\begin{quoting}
\begin{math}
\xymatrix{
  \Val 
  \ar[r]^-{\iota^\Res_\Val}
  \ar[rd]_-{\iota^\Conf_\Val}
& \Res 
  \ar[d]_-{\iota^\Conf_\Res}
& \Wrng 
  \ar[l]_-{\iota^\Res_\Wrng} 
  \ar[ld]^-{\iota^\Conf_\Wrng}
\\
& \Conf 
}
\end{math}
\end{quoting}
As customary, with a slight abuse of notation, we identify  elements with their images along such injections.
We use $\mexpr$, $\mval$, $\mres$ and $\mconf$ to range over monadic expressions, monadic values, monadic results and monadic configurations, respectively, 
that is, elements of $\mfun\Exp$, $\mfun\Val$, $\mfun\Res$ and $\mfun\Conf$. 
Since monads on \Set preserve injections \cite{AdamekBLM14}, 
by applying $\mfun$ to the diagram above we get  another diagram of injections:
\begin{quoting}
\begin{math}
\xymatrix{
  \mfun\Val 
  \ar[r]^-{\mfun\iota^\Res_\Val}
  \ar[rd]_-{\mfun\iota^\Conf_\Val}
& \mfun\Res 
  \ar[d]_-{\mfun\iota^\Conf_\Res}
& \mfun\Wrng 
  \ar[l]_-{\mfun\iota^\Res_\Wrng} 
  \ar[ld]^-{\mfun\iota^\Conf_\Wrng}
\\
& \mfun\Conf 
}
\end{math}
\end{quoting}
In the following, we use some shortcuts for the application of such injections: notably, 
we write 
$\metc\mexpr$ for $\mfun\iota^\Conf_\Exp(\mexpr)$, 
$\mrtc\mres$ for $\mfun\iota^\Conf_\Res(\mres)$ and 
$\mvtr\mval$ for $\mfun\iota^\Res_\Val(\mval)$. 

We can now extend the  monadic  reduction $\red$ to configurations, getting the relation 
${\tred \subseteq \Conf \times \mfun\Conf}$ shown in \cref{fig:mred-conf}.
\begin{figure}[th]
\begin{small}
\begin{math}
\begin{array}{c}
\NamedRule{exp}{
  \expr\red\mexpr 
}{ \expr \tred \metc\mexpr }
{ } 
\BigSpace
\NamedRule{ret}{
}{\valtoexp(\val)\tred \mun_\Conf(\val) }
{ } 
\\[3ex]
\NamedRule{wrong}{
}{\expr\tred \mun_\Conf(\wrng) }
{ \expr\not\red\\
\expr\ne\valtoexp(\val)\ \mbox{for all}\ \val\in\Val} 
\BigSpace
\NamedRule{res}{
}{\res \tred \mun_\Conf(\res) }
{ } 
\end{array}
\end{math}
\end{small}
\caption{Monadic (one-step) reduction on configurations}\label{fig:mred-conf} 
\end{figure}
As said above, reduction is extended to  expressions which represent terminated computations,  which reduce to the monadic embedding of the corresponding value or $\wrng$, respectively; moreover, it is extended to results (either values or $\wrng$) as well, which conventionally reduce \mbox{to their monadic embedding.}

It is immediate to see that $\tred$ is (the graph of) a total function from $\Conf$ to $\mfun\Conf$, which we simply write $\tredfun$. 
Clearly $\tredfun$ is a Kleisli function for $\mnd$, 
hence we can define the ``monadic reflexive and transitive closure'' $\tredstar\subseteq\Conf\times\mfun\Conf$ of $\tred$   as follows:
\begin{quoting}
\begin{math}
\NamedRule{refl}{ }{ \conf\tredstar \mun_\Conf(\conf)}{} 
\qquad
\NamedRule{step}{
  \conf\tredstar \mconf 
}{ \conf \tredstar \mconf \mbind  \tredfun }{} 
\end{math}
\end{quoting}
 These rules are analogous to those defining the reflexive and transitive closure of a standard one-step relation. In  \refToRule{refl}  a configuration reduces, rather than to itself, to its monadic counterpart. In rule \refToRule{step}, $\tredstar$ is combined with $\tred$, rather than by standard composition, through the $\mbind$ operator.  That is, a computation is extended by one step through a monadic binding of the previously computed monadic configuration $\mconf$ to the $\tredfun$ function. 

Equivalently, we can define the Kleisli $n$-th iteration $\tredfunind{n}$ of the $\tredfun$ function by setting 
$\tredfunind{0} = \mun_\Conf$ and  $\tredfunind{n+1} = \tredfun \ast \tredfunind{n}$. 
Then, the following holds:

\begin{proposition}\label{prop:mnd-n-step}
$\conf\tredstar \mconf$ if and only if $\tredfunind{n}(\conf)=\mconf$ for some $n\in\N$. 
\end{proposition}
\begin{proof}
The left-to-right implication follows by a straightforward induction on rules defining $\tredstar$, while the right-to-left one by  another \mbox{straightforward induction on $n$. }
\end{proof}

In a similar way, we can define a small-step reduction  on monadic configurations. Recall that the Kleisli extension of $\tredfun$ gives the function $\fun{\mkl\tredfun}{\mfun\Conf}{\mfun\Conf}$. 
\begin{definition}
The \emph{small-step reduction induced by $\tred$} is the relation $\Red$ on $\mfun\Conf$ defined by:
$\mconf\Red\mconf'$ iff $\mkl\tredfun(\mconf) = \mconf'$. 
\end{definition}
Then, since $\Red$ is  a relation on $\mfun\Conf$, we can consider its (standard) reflexive and transitive closure $\Redstar \subseteq\mfun\Conf\times\mfun\Conf$, which describes computations on monadic configurations. 

\begin{proposition}\label{prop:mnd-n-step2}
$\mconf\Redstar\mconf'$ if and only if $\mconf \mbind \tredfunind{n} = \mconf'$ for some $n \in \N$. 
\end{proposition}
\begin{proof}
It is immediate observing that 
$\mconf \mbind \tredfunind{n} = \mkl{(\tredfunind{n})}(\mconf) = (\mkl{\tredfun})^{n}(\mconf)$ and that 
$\mconf \Redstar \mconf'$ if and only if 
$\mconf = \mconf_1 \Red \ldots \Red \mconf_n = \mconf'$ if and only if 
$(\mkl{\tredfun})^{n}(\mconf) = \mconf'$. 
\end{proof}

Combining \cref{prop:mnd-n-step,prop:mnd-n-step2}, we get the following corollary which relates \mbox{$\tredstar$ and $\Redstar$. }

\begin{corollary} \label{cor:mnd-reds}
The following are equivalent:
\begin{enumerate}
\item $\conf\tredstar\mconf$
\item $\mun_\Conf(\conf)\Redstar\mconf$
\item $\conf\tredstar\mconf'\Redstar\mconf$, for some $\mconf'\in\mfun\Exp$
\end{enumerate}
\end{corollary}

 To use the above machinery for describing the  semantics of expressions, we essentially follow  the approach in \cite{GavazzoTV24}, with minor adjustments to fit our context. 

First of all note that, being defined on top of  a  total function, $\Redstar$ has no normal forms. However, monadic results should intuitively correspond to termination.
Formally, this is a consequence of the proposition below,  stating that a monadic configuration which is a result only reduces to itself; hence, when a monadic result is reached, its reduction continues with an infinite sequence of trivial reduction steps, which can be seen as a representation of termination. Hence, the outcome of a (terminating) computation is a monadic result.

\begin{proposition}\label{prop:mres}
$\mrtc\mres\Red\mconf$ if and only if $\mconf = \mrtc\mres$. 
\end{proposition}
\begin{proof}
By definition of $\tredfun$, we have 
$\tredfun\circ\iota^\Conf_\Res = \mun_\Conf \circ \iota^\Conf_\Res$. 
Applying the functor $\mfun$ and using the monad laws, we obtain the following commutative diagram\mbox{, which proves the thesis.}
\[\xymatrix{
  \mfun\Res 
  \ar[d]_-{\mfun\iota^\Conf_\Res}
  \ar[r]^-{\mfun\iota^\Conf_\Res}
& \mfun\Conf 
  \ar[d]_-{\mfun\mun_\Conf}
  \ar[rd]^-{\id_{\mfun\Conf}}
\\
  \mfun\Conf
  \ar[r]^-{\mfun\tredfun}
  \ar@/_10pt/[rr]_-{\mkl{\tredfun}}
& \mfun^2\Conf 
  \ar[r]^-{\mmul_\Conf}
& \mfun\Conf 
}\]
\end{proof}

Thanks to  the above proposition,  we can prove the following, stating that the  monadic result  of a computation, if any,  is unique. 
\begin{proposition}\label{prop:mnd-res}
If $\conf\tredstar\mrtcind{\mres}{1}$ and $\conf\tredstar\mrtcind{\mres}{2}$, then $\mres_1 = \mres_2$. 
\end{proposition}
%\cref{prop:mnd-res} shows that every configuration $\conf$ can reduce to at most one monadic result. 
\begin{proof}
By \cref{cor:mnd-reds}, we have $\mun_\Conf(\conf) \Redstar\mrtcind{\mres}{1}$ and $\mun_\Conf(\conf)\Redstar\mrtcind{\mres}{2}$. 
Since $\Redstar$ is the reflexive and transitive closure of a functional relation, it has the diamond property, 
hence there is $\mconf'$ such that 
$\mrtcind{\mres}{1}\Redstar\mconf'$ and $\mrtcind{\mres}{2}\Redstar\mconf'$. 
Then, by \cref{prop:mres}, we conclude $\mrtcind{\mres}{1} = \mconf' = \mrtcind{\mres}{2}$ and the thesis follows from \mbox{the injectivity of $\iota^\Conf_\Res$. }
\end{proof}

Hence, we can define a function \fun{\finsem{-}}{\Exp}{\mfun\Res + \{\infty\}} describing the semantics of expressions as follows: 
\begin{quoting}
\begin{math}
\finsem\expr = \begin{cases}
\mres & \text{if $\expr\tredstar\mrtc\mres$} \\ 
\infty & \text{otherwise} 
\end{cases}
\end{math}
\end{quoting}
This is called \emph{finitary semantics}, as it describes only monadic results that can be reached in finitely many steps. 
In other words, all diverging computations are identified and no information on computational effects they may produce is available. 
 Even worse, when the monad supports some form of non-determinism, we may have computations that terminate in some cases and diverge in  others, 
but the finitary semantics considers them as diverging, as they never reach a result after finitely many steps. 
 %An example of this situation is provided in \cref{sect:examples} (\cref{ex:pow}). 

To overcome this limitation, again following \cite{GavazzoTV24}, we introduce an \emph{infinitary semantics}, which is able to provide more information on diverging computations. 
To achieve this, we need to assume more structure on the monad $\mnd$. 
Recall that, given a partially ordered set \ple{P,\mord}, an $\omega$-chain is an increasing sequence  $(x_n)_{n\in\N}$  of points in $P$. 
We say that \ple{P,\mord} is an $\omega$-CPO if it has a least element $\mbot$ and every $\omega$-chain in \ple{P,\mord} has a supremum $\msup_{n\in\N} x_n$. 
A function \fun{f}{\ple{P,\mord}}{\ple{P',\mord'}} between $\omega$-CPO is said to be $\omega$-continuous if it preserves the least element, and suprema of $\omega$-chains. 
Note that an $\omega$-continuous function is necessarily monotone. 
Then, we have the following definition:

\begin{definition}\label{def:mnd-ord}
An $\omega$-CPO-ordered monad $\mnd = \ple{\mfun,\mord,\mun,\mmul}$ is a monad \ple{\mfun,\mun,\mmul} together with  a partial order $\mord_X$ on $\mfun X$, for every set $X$, such that 
\begin{enumerate}
\item\label{def:mnd-ord:cpo} for every set $X$, the poset \ple{\mfun X,\mord_X} is an $\omega$-CPO and 
\item\label{def:mnd-ord:cont} for all sets $X,Y$, the Kleisli extension \fun{\mkl{(-)}}{(X\to\mfun Y)}{(\mfun X\to\mfun Y)} is \mbox{$\omega$-continuous} with respect to the pointwise extension of $\mord_Y$ to function spaces $X{\to}\mfun Y$ and ${\mfun X{\to}\mfun Y}$. 
\end{enumerate}
\end{definition}

\begin{example}
\label{ex:mnd-cpo}
The powerset and list monads of \cref{ex:pow-mnd} are $\omega$-CPO-ordered with the subset and prefix ordering, respectively.
The subdistribution monad of \cref{ex:prob-mnd} is $\omega$-CPO-ordered with the pointwise ordering on subdistributions.
%The tree monad of \cref{ex:tree-mnd} is $\omega$-CPO-ordered with the left-prefix ordering.
The other monads of \cref{sect:monads} can also be turned into $\omega$-CPO-ordered monads, but require adjustements, typically a combination with the exception monad. 
For instance, in \cref{ex:output-mnd}, the output monad is not $\omega$-CPO-ordered in general, 
but its pointed version is $\omega$-CPO-ordered when the underlying monoid is an $\omega$-CPO  and the multiplication is $\omega$-continuous \mbox{in the second argument. }
\end{example}

From now on, we assume the monad $\mnd$ to have an $\omega$-CPO-ordered structure. 
Our goal is  to define a function \fun{\infsem{-}}{\Exp}{\mfun\Res} modelling the infinitary semantics of expressions. 
To this end, we first define a function \fun{\mctr}{\mfun\Conf}{\mfun\Res} extracting monadic results from monadic configurations. \label{res}
Let \fun{\mctr_0}{\Conf}{\mfun\Res} be the function given by 
\begin{quoting}
\begin{math}
\mctr_0(\conf) = \begin{cases}
\mbot_\Res & \conf = \expr \\ 
\mun_\Res(\res) & \conf = \res 
\end{cases}
\end{math}
\end{quoting}
and set $\mctr = \mklind{\mctr}{0}$. 
The key point is that the total relation $\Red$ on monadic configurations  is compatible with the order $\mord_\Res$ under the application of $\mctr$, as the following \mbox{proposition shows. }

\begin{proposition}\label{prop:mnd-Red-ord}
If $\mconf\Red\mconf'$ then $\mctr(\mconf)\mord_\Res\mctr(\mconf')$. 
\end{proposition}
\begin{proof}
Since $\Red$ is the graph of $\mkl\tredfun$, we have to show that 
$\mctr(\mconf)\mord_\Res\mctr(\mkl\tredfun(\mconf))$. 
Since $\mctr =  \mklind{\mctr}{0}$, we have 
$\mctr\circ\mkl\tredfun = \mkl{(\mctr\circ\tredfun)}$. 
Hence, it suffices to prove that ${\mctr_0(\conf) \mord_\Res \mctr(\tredfun(\conf))}$, for every $\conf\in\Conf$, because the Kleisli extension, being $\omega$-continuous, is monotone. 
We reason by cases on $\conf$. 
If $\conf = \expr$, then $\mctr_0(\conf) = \mbot_\Res$ and so the thesis is trivial. 
If $\conf = \res$,  then $\mctr_0(\conf) = \mun_\Conf(\res) = \mun_\Res(\res)$ and $\tredfun(\res) = \mun_\Conf(\res)$ and the following diagram commutes: 
\[\xymatrix{
  \Res 
  \ar@(l,l)[dd]_-{\iota^\Conf_\Res} 
  \ar[d]_-{\iota^\Conf_\Res}
  \ar[r]^-{\mun_\Res} 
& \mfun\Res 
  \ar[d]_-{\mfun\iota^\Conf_\Res}
  \ar[rd]^-{\mfun\mun_\Res}
  \ar@(r,u)[rrd]^-{\id_{\mfun\Res}} 
\\
  \Conf 
  \ar[r]_-{\mun_\Conf}
& \mfun\Conf
  \ar[r]^-{\mfun\mctr_0}
  \ar@(rd,ld)[rr]_-{\mctr}
& \mfun^2\Res 
  \ar[r]^-{\mmul_\Res}
& \mfun\Res
\\
  \Conf
  \ar@(r,d)[ru]_-{\tredfun} 
}\]
showing that $\mctr(\tredfun(\conf)) = \mun_\Res(\res) = \mctr_0(\conf)$, which proves the thesis. 
\end{proof}

For every $\expr\in\Exp$ and $n\in\N$, we define 
$\infsem[n]{\expr} = \mctr(\tredfunind{n}(\expr))$. 
From \cref{prop:mnd-n-step,prop:mnd-n-step2}, we easily derive 
$\tredfunind{n}(\expr)\Red\tredfunind{n+1}(\expr)$, and, by \cref{prop:mnd-Red-ord}, 
$\infsem[n]{\expr} \mord_\Res \infsem[n+1]{\expr}$. 
Hence, the sequence $(\infsem[n]{\expr})_{n\in\N}$ is an $\omega$-chain in \ple{\mfun\Res,\mord_\Res} and so we define the \mbox{infinitary semantics as}
\begin{quoting}
\begin{math}
\infsem{\expr} = \msup_{n\in\N} \infsem[n]{\expr} 
\end{math}
\end{quoting}
Intuitively, $\infsem[n]{\expr}$ is the portion of the  result that is reached after $n$ reduction steps. 
Hence, the actual result is obtained as the supremum of all such approximations and it may be never reached, thus describing also the observable behaviour of  possibly diverging computations. 

We conclude this section by stating that infinitary and finitary semantics agree on terminating computations.

\begin{proposition}\label{prop:mnd-fin-inf-sem}
If $\finsem{\expr} = \mres$, then $\infsem{\expr} = \mres$. 
\end{proposition}
\begin{proof}
We know that $\expr\tredstar\mrtc\mres$, hence, by \cref{prop:mnd-n-step}, we have $\tredfunind{k}(\expr) = \mrtc\mres$  for some $k$. 
For all $n\ge k$, we have 
$\tredfunind{k}(\expr)\Redstar\tredfunind{n}(\expr)$, hence, by \cref{prop:mres}, we deduce
$\tredfunind{n}(\expr) = \mrtc\mres$. 
Therefore, for all $n\ge k$, we have 
$\infsem[n]{\expr} = \mctr(\mrtc\mres)$ and so 
$\infsem{\expr} = \mctr(\mrtc\mres)$. 
Finally, since $\mctr_0\circ\iota^\Conf_\Res = \mun_\Res$, we deduce 
$\mctr\circ\mfun\iota^\Conf_\Res = \mkl{(\mctr_0\circ\iota^\Conf_\Res)} = \mkl{\mun_\Res} = \id_{\mfun\Res}$, 
thus proving 
$\mctr(\mrtc\mres) = \mres$, as needed. 
\end{proof}

% !TEX root =main.tex

\section{Example: a lambda calculus with generic effects}\label{sect:lang}
The aim of this section is  twofold:
\begin{itemize}
\item to ilustrate the monadic operational semantics in \cref{sect:monadic-sem} through a simple example
\item to equip such example with a type-and-effect system, and to discuss how to express and prove type soundness with respect to finitary/infinitary semantics
\end{itemize}
To this end, we  introduce  $\lambdaEff$,  a call-by-value $\lambda$-calculus with \emph{generic effects}. Here $\Sigma$ is a family of sets $\{\Sigma_k\}_{k\in\N}$ of $k$-ary operations  raising effects. 
We choose generic rather than algebraic effects, thus avoiding explicit continuations, to have a style more convenient for a programmer \cite{Pretnar15}, and a more significant monadic reduction.\footnote{In the case of algebraic effects there would be no monadic reduction inside a context, as in rule \refToRule{do}.} 

The syntax is shown in \cref{fig:syntax}.  We use $\seq{\ve}$ as metavariable for sequences $\ve_1,\ldots, \ve_n$, and analogously for other sequences.
\begin{figure}[th]
\begin{math}
\begin{grammatica}
\produzione{\ve}{\x \mid \RecFun{\f}{\x}{\e}\mid\ldots}{value} \\ 
\produzione{\e}{ \App{\ve}{\ve'}\mid \opg(\seq\ve)\mid\Ret\ve \mid \Do\x{\e_1}{\e_2}
\mid\ldots}{expression} \\ 
\end{grammatica}
\end{math}
\caption{$\lambdaEff$: fine-grain syntax}\label{fig:syntax} 
\end{figure}
We assume variables $\x$, $\y$, $\f$, \ldots, using the last for variables denoting functions. 
We adopt, as customary, the fine-grain approach \cite{LevyPT03}, where \emph{values} are effect-free, whereas \emph{expressions}, also called \emph{computations}, may raise effects. Dots stand for additional, unspecified, constructs, such as operators of primitive types, \mbox{conditional, etc.}

To illustrate type soundness with respect to the infinitary semantics as well, the calculus includes recursive functions; notably, $\RecFun{\f}{\x}{\e}$ is a function with parameter $\x$ and body $\e$ which can recursively call itself through the variable $\f$.
Standard lambda expressions can be recovered as those where $\f$ does not occur free in $\e$, that is, when the function is non-recursive, and we will use the abbreviation $\Fun{\x}{\e}$ for such expressions. 

 In this section, $\Exp$ and $\Val$ denote the sets of closed expressions and values of $\lambdaEff$, respectively. 
 In the following, we define a monadic (one-step) reduction  for the language, parametric on a monad  $\mnd = \ple{\mfun,\mun,\mmul}$,  being a relation $\red$  on  $\Exp \times \mfun\Exp$.  As in \cref{sect:monadic-sem}, we  use $\mval$ and $\mexpr$ to range over $\mfun\Val$ and $\mfun\Exp$, respectively. 

This relation is modularly defined on top of a ``pure''  reduction $\purered$ on $\Exp\times\Exp$.   In our example, such relation only reduces function calls into  the corresponding bodies, as shown in \cref{fig:pure-red}; other rules should be added to deal with additional language constructs, as we will do for handlers in \cref{sect:handlers}. \lstinline{Do} expressions are, then, normal forms for the pure reduction, and will be handled by  the rules \mbox{of  the monadic reduction. }
 
\begin{figure}[th]
\begin{math}
\begin{array}{c}
\NamedRule{app}{ }
{ \App{\ve}{\ve'} \purered \Subst{\Subst{\e}{\ve}{\f}}{\ve'}{\x}}
{\ve =\RecFun{\f}{\x}{\e}}
\end{array}
\end{math}
\caption{Pure reduction}\label{fig:pure-red} 
\end{figure}

Rules defining the monadic reduction are given in \cref{fig:monadic-red}.   As mentioned, they are parametric on the underlying monad; more in detail, they depend on the following ingredients:
\begin{itemize}
\item The function $\fun{\mun_\Exp}{\Exp}{\mfun\Exp}$  embedding  language expressions  into their counterpart in the monad, written simply $\mun$ in this section.
\item The function $\fun{\aux{map}}{(\funtype{\Exp}{\Exp})}{\funtype{\mfun\Exp}{\mfun\Exp}}$  lifting functions from expressions to expressions  to their counterpart in the monad. 
\end{itemize}
Moreover we assume, for each operation $\opg$ with arity $k$, a partial function ${\pfun{\mrun{\opg}}{\Val^k}{\mfun\Val}}$,  returning a monadic value expressing the effects raised by a call of the operation. The function could be undefined, for instance when arguments do not have the expected types. 

\begin{figure}[th]
\begin{math}
\begin{array}{c}
\NamedRule{pure}{
  \e \purered \e' 
}{ \e \red \mun(\e') }
{ } 
\qquad
\NamedRule{effect}{ }
{ \opg(\seq\ve) \red \Mmap{(\Ret{\ehole})}{\mrun{\opg}(\seq\ve)}  }
{ 
}
\\[3ex]
\NamedRule{ret}{ }
{ \Do{\x}{\Ret\ve}{\e} \red \mun(\Subst\e\ve\x) }
{ } 
\qquad 
\NamedRule{do}{
  \e_1 \red \mexpr 
}{ \Do\x{\e_1}{\e_2} \red \Mmap{(\Do\x{\ehole}{\e_2})}{\mexpr} } 
{ }
\end{array}
\end{math}
\caption{Monadic (one-step) reduction}\label{fig:monadic-red} 
\end{figure}

Rule \refToRule{pure} propagates a pure step, embedding its result in the monad.
In rule \refToRule{effect}, the effect is actually raised. 
To this end, we apply the function of type $\funtype{\mfun\Val}{\mfun\Exp}$ obtained by lifting, through $\aux{map}$, the context $\Ret{\ehole}$ to the monadic value obtained from the call.
 Here we identify the context $\Ret{\ehole}$,  which is an expression with a hole, with the function $\ve \mapsto \Ret[\ve]$ of type $\Val\to\Exp$. 
  In rule \refToRule{ret}, when the first subterm of a \lstinline{do} expression returns a value, the expression is reduced to the monadic  embedding  of the second subterm, after replacing the variable with the  value.
Rule \refToRule{do}, instead, propagates the reduction of the first subterm. To take into account raised effects, we apply the function of type $\funtype{\mfun\Exp}{\mfun\Exp}$ obtained by lifting, through $\aux{map}$, the context $\Do\x{\ehole}{\e_2}$ to the monadic expression obtained from $\e_1$. 
 Analogously to above,  we identify the context $\Do\x{\ehole}{\e_2}$ with the function $\e \mapsto \Do{\x}{[\e]}{\e_2}$ of type $\Exp\to\Exp$. 
 
 The following property is needed to have an instance of the framework in \cref{sect:monadic-sem}.
\begin{proposition}[Determinism] \label{prop:det}
If $\e\red \mexpr_1$ and $\e\red\mexpr_2$ then  $\mexpr_1 = \mexpr_2$. 
\end{proposition}
We show now some examples of expressions and their monadic operational semantics.  We assume the calculus to be extended with standard constructs, such $\unit$, $\Zero$, $\aux{succ}$, $\True$ and $\False$ constructors, $\aux{pred}$ selector, $\aux{iszero}$ test, and conditional.
We write $\Seq{\e}{\e'}$ for $\Do{\x}{\e}{\e'}$ when $\x$ does not occur free in $\e'$, and sometimes, to save space, $\nat{n}$ for $\SuccSeq{n}{\Zero}$. 

\begin{example}\label{ex:exc}
Set, as underlying monad, the monad of exceptions introduced in \cref{ex:exception-mnd}, where $\ExceptFun X = X + \ExSet$.
For each $\exc\in\ExSet$, we assume an operation   $\raiseop{\exc}$, with
\begin{quoting}
$\fun{\mrun{\raiseop{\exc}}}{\mathbf{1}}{\mfun{\Val}}$\HugeSpace
$\mrun{\raiseop{\exc}}=\iota_2(\exc)$
\end{quoting}
The function $\predfun=\Fun{\x}{\If{\IsZero{\x}}{\raiseop{\PredZeroExc}}{\Ret{\Pred{\x}}}}$ raises the exception $\PredZeroExc$ when the argument is $0$.
The following are examples of small-step reduction sequences on monadic configurations\footnote{Where we omit the injections from monadic expressions and values.}:
\begin{quoting}
\begin{math}
\begin{array}{lcl}
\App{\predfun}{\Succ{\Zero}}&\Red&\If{\IsZero{\Succ{\Zero}}}{\raiseop{\PredZeroExc}}{\Ret{\Zero}}\\
&\Red&\If{\False}{\raiseop{\PredZeroExc}}{\Ret{\Zero}}\\
&\Red&\Ret{\Zero}\\
&\Red&\Zero
\\
\App{\predfun}{\Zero}&\Red&\If{\IsZero{\Zero}}{\raiseop{\PredZeroExc}}{\Ret{\Zero}}\\
&\Red&\If{\True}{\raiseop{\PredZeroExc}}{\Ret{\Zero}}\\
&\Red&\raiseop{\PredZeroExc}\\
%&\Red&\Ret{\PredZeroExc}\\
&\Red&\PredZeroExc
\end{array}
\end{math}
\end{quoting}
In the first reduction sequence,  all steps are derived by  rules \refToRule{pure} in \cref{fig:monadic-red} and \refToRule{exp} in \cref{fig:mred-conf}, except for the last one, which is derived by rule \refToRule{ret} in \cref{fig:mred-conf}.   Analogously in the second reduction sequence,  where the last step is derived  by rule \refToRule{effect} in \cref{fig:monadic-red} and \refToRule{exp} in \cref{fig:mred-conf}.
Note that, here and in the following examples, after reaching a monadic result the sequence of steps  continues with an infinite sequence of steps, in the case above $\Zero\Red\Zero$ and $\PredZeroExc\Red\PredZeroExc$ steps. 
 \end{example}

\begin{example}\label{ex:pow}
Set, as underlying monad, the monad of non-determinism of \cref{ex:pow-mnd}, in the variant  of the possibly infinite lists.
We assume a constant operation $\chooseop$,  with 
\begin{quoting}
$\fun{\mrun{\chooseop}}{\mathbf{1}}{\mfun{\Val}}$\HugeSpace
$\mrun{\chooseop}=[\True,\False]$
\end{quoting}
The expression $\e=\Do{\y}{\chooseop}{\If{\y}{\Ret{\Zero}}{\Ret{\Succ{\Zero}}}}$ reduces as follows.\footnote{ We use $\Tru$, $\Fls$, $\kw{ret}$, and $\kw{s}$, for $\True$, $\False$, $\kw{return}$, and $\kw{Succ}$, to save space.}
\begin{quoting}
\begin{math}
\begin{array}{lcl}
\List{\e}&\Red&\List{\Do{\y}{\True}{\If{\y}{\Ret{\Zero}}{\Ret{\Succ{\Zero}}}}, \\
&&\Do{\y}{\False}{\If{\y}{\Ret{\Zero}}{\Ret{\Succ{\Zero}}}}}\\
&\Red&\List{\If{\Tru}{\Rt{\Zero}}{\Rt{\nat{1}}},\If{\Fls}{\Rt{\Zero}}{\Rt{\nat{1}}}}\\
&\Red&\List{\Ret{\Zero},\Ret{\Succ{\Zero}}}\\
&\Red&\List{\Zero,\Succ{\Zero}}
\end{array}
\end{math}
\end{quoting}
Given $\chooseincr=\RecFun{\f}{\x}{\Do{\y}{\chooseop}{\If{\y}{\Ret{\x}}{\App{\f}\Succ{\x}}}}$, the expression $\App{\chooseincr}{\Zero}$ reduces as follows: 
\begin{quoting}
\begin{math}
\begin{array}{lcl}
\List{\App{\chooseincr}{\Zero}}&\Red&\List{\Do{\y}{\chooseop}{\If{\y}{\Ret{\Zero}}{\App{\chooseincr}{\Succ{\Zero}}}}}\\
&\Red&\List{\Do{\y}{\True}{\If{\y}{\Ret{\Zero}}{\Ret{\App{\chooseincr}{\Succ{\Zero}}}}}, \\
&&\Do{\y}{\False}{\If{\y}{\Ret{\Zero}}{\App{\chooseincr}\Succ{\Zero}}}}\\
&\Red&\List{\If{\Tru}{\Rt{\Zero}}{\Rt{\nat{1}}},\If{\Fls}{\Rt{\Zero}}{\App{\chooseincr}{\nat{1}}}}\\
&\Red&\List{\Ret{\Zero},\App{\chooseincr}{\Succ{\Zero}}}\\
&\Red&\List{\Zero,\App{\chooseincr}{\Succ{\Zero}}}\\
&\Red&\List{\Zero,\Do{\y}{\chooseop}{\If{\y}{\Rt{\Succ{\Zero}}}{\App{\chooseincr}{\Su{\Su{\Zero}}}}}}\\
&\Redstar&\List{\Zero,\Succ{\Zero},\App{\chooseincr}{\Succ{\Succ{\Zero}}}}\\
&\ldots
\end{array}
\end{math}
\end{quoting}
 Note that the second reduction is non-terminating,  in the sense that  a monadic result (a list of values) is never reached. Hence, with the finitary semantics, we get $\finsem{\App{\chooseincr}{\Zero}}=\infty$.  
With the infinitary semantics, instead, we get the following $\omega$-chain:
\begin{quoting}
$\List{\ }$, \ldots, $\List{\Zero}$, \ldots, $\List{\Zero, \Succ{\Zero}}$, \ldots, $\List{\Zero, \Succ{\Zero}, \ldots, \SuccSeq{n}{\Zero}}$, \ldots, 
\end{quoting}
whose supremum is, as expected, the infinite list of the (values representing the) natural numbers. On the other end, given the function 
\begin{small}
\begin{quoting}
$\choosedecr=\RecFun{\f}{\x}{\If{\IsZero{\x}}{\Rt{\x}}{\Do{\y}{\chooseop}{\If{\y}{\Rt{\x}}{\App{\f}\Pred{\x}}}}}$
\end{quoting}
\end{small}
we get
$\finsem{\App{\choosedecr}{\SuccSeq{n}{\Zero}}}=\infsem{\App{\choosedecr}{\SuccSeq{n}{\Zero}}}=\List{\SuccSeq{n}{\Zero}, \ldots, \Zero}$.
\end{example}

\begin{example}\label{ex:prob}
 Set, as underlying monad, the  monad of probabilistic non-determinism 
of \cref{ex:prob-mnd}. 
We consider a discrete uniform distribution  over a set of two elements  and use the same function $\chooseop$, now
returning the list consisting of  the   values $\True$ and $\False$  with probability $\frac{1}{2}$, that we denote by ${\List{\frac{1}{2}:\True,\frac{1}{2}:\False}}$.

Then, the expressions $\List{1:\e}$ and $\List{1:\App{\chooseincr}{\Zero}}$ reduce  analogously \mbox{to  the previous example:}
\begin{quoting}
$\List{1:\e}\Redstar\List{\frac{1}{2}:\Zero,\frac{1}{2}:\Succ{\Zero}}$\\[0.5ex]
$\List{\List{1:\App{\chooseincr}{\Zero}}}\Redstar \List{ \frac{1}{2}:\Zero,\frac{1}{4}:\Succ{\Zero},\frac{1}{8}:\SuccSeq{2}{\Zero},\frac{1}{16}:\SuccSeq{3}{\Zero}}\Red\ldots$
\end{quoting}
Again, the second reduction is non-terminating, hence, with the finitary semantics, we get $\infty$, whereas, 
with the infinitary semantics, we get an $\omega$-chain whose supremum is the infinite list where each (value representing the) number $n$ has \mbox{probability $\frac{1}{2^{n+1}}$. }
\end{example}

\begin{example}\label{ex:output}
 Set, as underlying monad, the output monad of \cref{ex:output-mnd}, in its pointed version. As a simple concrete choice, we take as elements of $\OutSet$ sequences of pairs $\Pair{\loc}{\SuccSeq{n}{\Zero}}$ where $\loc$ ranges over a fixed set $\LocSet$ of  \emph{output locations}  modeling, e.g., file names or output channels. 
We assume, for each $\loc$, an operation $\fun{\writeop{\loc}}{\Nat}{\Unit}$, with
\begin{quoting}
$\fun{\mrun{\writeop{\loc}}}{\Val}{\mfun{\Val}}$\HugeSpace
$\mrun{\writeop{\loc}}(\val)=
\begin{cases}
\Pair{\Pair{\loc}{\val}}{\unit}&\mbox{if}\ \val=\SuccSeq{n}{\Zero}\\
\mbox{undefined} & \mbox{otherwise}
\end{cases}
$
\end{quoting}
Given two distinct  output locations  $\loc$, $\loc'$, and the functions
\begin{quoting}
 $\writeincr=\RecFun{\f}{\x}{\Seq{\writeop{\loc}(\x)}{\Seq{\writeop{\loc'}(\x)}{\App{\f}{\Succ{\x}}}}}$\\
 $\writedecr=\RecFun{\f}{\x}{\Seq{\writeop{\loc}(\x)}{\Seq{\writeop{\loc'}(\x)}{\If{\IsZero{\x}}{\unit}{\App{\f}{\Pred{\x}}}}}}$
\end{quoting}
we get, as in the previous examples, the following semantics:
\begin{quoting}
$\finsem{\App{\writeincr}{\Zero}}=\infty$\\
$\infsem{\App{\writeincr}{\Zero}}=\Pair{\Pair{\loc}{\Zero}\cdot \Pair{\loc'}{\Zero}\cdot\Pair{\loc}{\Succ{\Zero}}\cdot\Pair{\loc'}{\Succ{\Zero}}\cdot\ldots}{\mbot}$\\
$\finsem{\App{\writedecr}{\SuccSeq{n}{\Zero}}}=\infsem{\App{\writedecr}{\SuccSeq{n}{\Zero}}}=\Pair{\Pair{\loc}{\SuccSeq{n}{\Zero}}\cdot\Pair{\loc'}{\SuccSeq{n}{\Zero}}\cdot\ldots\cdot\Pair{\loc}{\Zero}\cdot\Pair{\loc'}{\Zero}}{\unit}$
\end{quoting}
In the first two cases the reduction does not terminate, so no value is returned. With the finitary semantics also no effect is produced, whereas with the infinitary semantics the effect is the infinite sequence of outputs. 
\end{example}
In order to equip the calculus (\cref{fig:syntax}) with  a type-and-effect  system, we need the ingredients shown in  \cref{fig:typed}.  
\begin{figure}[th]
\begin{math}
\begin{grammatica}
%\produzione{\ve}{\x \mid \RecFunT{\f}{\x}{\T}{\e}}{value} \\ 
\produzione{\T}{\ldots\mid\FunType{\T}{\eff}{\T'}}{type}\\
%\produzione{\eff}{\eZero\mid\eTop\mid\EComp{\eff}{\eff'}\mid\opg}{effect} \\ 
\produzione{\Gamma}{\seq{\x:\T}}{context}
\end{grammatica}
\end{math}
\caption{Types and contexts}\label{fig:typed}
\end{figure}
Besides types, which are functional types and additional unspecified types, we consider  \emph{effect types} (\emph{effects} when there is no ambiguity), ranged over by $\eff$,  meant to be static approximations of the computational effects raised by an expression.
As formally detailed below, effects  are  sets of (possibly infinite) sequences of operations. In this way, they are expressive enough to approximate computational effects in many different monads, as we will describe in \cref{sect:results}, and we abstract away from details of a syntactic representation, which of course would be needed in a real language. Functional types are annotated with an effect,  approximating the computational effects of calling the function.
Finally, we assume operations to be typed; formally, for each $\opg$, we write $\fun{\opg}{\T_1\ldots\T_n}{\T}$.

Set $\Sigma^\infty$ the set of either finite or infinite sequences of operations. 
We use $\alpha,\beta$ to range over elements of $\Sigma^\infty$, denote by $\epsilon$ the empty sequence, by $\Cons{\opg}{\alpha}$ the sequence consisting of $\opg$ followed by $\alpha$, and by 
$\EComp{}{}$ sequence concatenation, coinductively defined by: 
\begin{quoting}
$\EComp{\epsilon}{\beta}=\beta$\HugeSpace
$\EComp{(\Cons{\opg}{\alpha})}{\beta}=\Cons{\opg}{(\EComp{\alpha}{\beta})}$
\end{quoting}
As customary, we write $\opg$ for the sequence $\Cons{\opg}{\epsilon}$.

%Since $\Sigma^\infty$ is a complete partial order with respect to the prefix relation, for each chain (totally ordered subset) $\chaineff$ in $\Sigma^\infty$ subset of \emph{closed} there exists the supremum, denoted $\PrefixSup{\chaineff}$.  
 An \emph{effect} is a  non-empty subset of $\Sigma^\infty$. 
 We denote by $\EComp{}{}$ composition of effects, defined by:
\begin{quoting}
$\EComp{\eff}{\eff'}=\{\EComp{\alpha}{\beta}\mid\alpha\in\eff, \beta\in\eff'\}$
\end{quoting}
 Absence of effects is modeled by the set $\{\epsilon\}$; the empty effect, if allowed, could be assigned to non-terminating computations which never call operations; however, since $\EComp{\eff}{\emptyset}=\emptyset$, effects assigned to a previous terminating computation would be lost. 

\begin{figure}[th]
\begin{math}
\begin{array}{l}
\NamedRule{sub-fun}
{\SubT{\T'_1}{\T_1}\BigSpace\SubT{\T_2}{\T_2'}}
{\SubT{\FunType{\T_1}{\eff}{\T_2}}{\FunType{\T_1'}{\eff'}{\T'_2}}}
{\SubE{\eff}{\eff'}}
\BigSpace
\NamedRule{sub-refl}{}{\SubT{\T}{\T}}{}
\\[3ex]
\NamedRule{sub-trans}{\SubT{\T}{\T'}\Space\SubT{\T'}{\T''}}{\SubT{\T}{\T''}}{}
\BigSpace
\NamedRule{sub-te}{\SubT{\T}{\T'}\Space\SubE{\eff}{\eff'}}{\SubTE{\TEff{\T}{\eff}}{\TEff{\T'}{\eff'}}}{}
\\[2ex]
\end{array}
\end{math}

%\vspace{1ex} 
\hrule
%\vspace{1ex} 

\begin{math}
\begin{array}{l}
\\[-1ex]
%\mbox{\bf Values:}\\[4ex]
\NamedRule{t-var}{}
{ \IsWFExp{\Gamma}{\x}{\T} }
{\Gamma(\x)=\T } 
\BigSpace
\NamedRule{t-abs}{
  \IsEffExp{\Gamma,\TVar{\FunType{\T}{\eff}{\T'}}{\f},\TVar{\T}{\x}}{\e}{\T''}{\eff'}
}
{ \IsWFExp{\Gamma}{\RecFun{\f}{\x}{\e}}{\FunType{\T}{\eff}{\T'}} }
{ \SubTE{\TEff{\T''}{\eff'}}{\TEff{\T'}{\eff}} }
\\[3ex]
\end{array}
\end{math}

%\vspace{1ex} 
\hrule
%\vspace{1ex} 

\begin{math}
\begin{array}{l}
\\
\NamedRule{t-app}{
  \begin{array}{l}
  \IsWFExp{\Gamma}{\ve_1}{\FunType{\T_1}{\eff}{\T}}\\
  \IsWFExp{\Gamma}{\ve_2}{\T_2}
  \end{array}}
  {  \IsEffExp{\Gamma}{\App{\ve_1}{\ve_2}}{\T}{\eff}}
{\SubT{\T_2}{\T_1}}
\BigSpace
\NamedRule{t-op}{\IsWFExp{\Gamma}{\ve_i}{\T'_i}\ \ \forall i\in 1..n}{\IsEffExp{\Gamma}{\opg(\ve_1,\ldots,\ve_n)}{\T}{\SetOf{\opg}}}
{
\fun{\opg}{\T_1\ldots\T_n}{\T}\\
\SubT{\T'_i}{\T_i} \ \forall i\in 1..n
}
\\[3ex]
\NamedRule{t-ret}{
  \IsWFExp{\Gamma}{\ve}{\T}
  }
{ \IsEffExp{\Gamma}{\Ret\ve}{\T}{\{\eZero \}} }
{ } 
\BigSpace
\NamedRule{t-do}{
  \IsEffExp{\Gamma}{{\e_1}}{\T_1}{\eff_1} 
  \BigSpace
  \IsEffExp{\Gamma,\TVar{\T_2}{\x}}{{\e_2}}{\T}{\eff_2} 
}
{ \IsEffExp{\Gamma}{\Do\x{\e_1}{\e_2}}{\T}{\EComp{\eff_1}{\eff_2}} }
{ \SubT{\T_1}{\T_2}  } 
\end{array}
\end{math}
\caption{Type-and-effect system}\label{fig:typing}
\end{figure}

The type-and-effect system is shown in \cref{fig:typing}.
The subtyping judgment has shape $\SubT{\T}{\T'}$. In \refToRule{sub-fun} inclusion of effect types is propagated to functional types. 
So a function producing less effects can be used where one producing more effects is needed. Moreover subtyping is, as expected, covariant/contravariant on the result/parameter of functions. The other rules are standard. 

The typing judgment for values has shape $\IsWFExp{\Gamma}{\ve}{\T}$, since they have no effects. The judgment for expressions, instead, has shape $\IsEffExp{\Gamma}{\e}{\T}{\eff }$. 

Rule \refToRule{t-var} is standard. In rule \refToRule{t-abs}, a (possibly recursive) function gets a functional type, consisting of parameter/result types and effect, if the body, in a context where  parameter and function are added with their types, gets 
a subtype and a subeffect of the result type and effect of the function. 
In rule \refToRule{t-app}, an application gets the result type and the effect of the applied function, provided that the argument type is subtype of the expected one. 
In rule \refToRule{t-op}, calling an operation raises the   corresponding singleton effect,  provided that the argument types are subtypes of the expected ones.
In rule \refToRule{t-ret},  an expression representing a value  has the trivial effect, and, in \refToRule{t-do}, a sequential composition of two computations has the composition of the two effects.

\begin{example}\label{ex:typing} We show some typing judgments which can be derived for the previous examples. We assume primitive types $\Nat$ and $\Bool$, an empty type $\Bot$ subtype of any type, the singleton type $\Unit$ for the constant $\unit$, and the obvious typing rule for conditional which takes the union of the effects of the two branches.  Finally, we denote by $\alpha^n$ and $\alpha^\omega$ a finite and infinite concatenation of $\alpha$s, respectively. 
%, and write $\Seq{\e}{\e'}$ for $\Do{\x}{\e}{\e'}$ meaning that $\e$ has type $\Unit$ and $\x$ does not occur free in $\e'$. 
\begin{enumerate}
\item \label{ex:typing:exc} In \cref{ex:exc}, with, for each $\exc\in\ExSet$, $\fun{\raiseop{\exc}}{\mathbf{1}}{\Bot}$,
\begin{quoting}
$\IsWFExp{\emptyset}{\predfun}{\FunType{\Nat}{\SetOf{\epsilon,\raiseop{\PredZeroExc}}}{\Nat}}$\\
$\IsEffExp{\emptyset}{\App{\predfun}{\val}}{\Nat}{\SetOf{\epsilon,\raiseop{\PredZeroExc}}}$ if $\IsWFExp{\emptyset}{\val}{\Nat}$
\end{quoting}
 Note a significant feature of our type effects:  differently from, e.g., Java checked exceptions, we can distinguish code which \emph{may} raise an exception, as expressed by the effect   $\SetOf{\epsilon,\raiseop{\PredZeroExc}}$, from code which \emph{necessarily} raises an exception, as expressed by the effect $\SetOf{\raiseop{\PredZeroExc}}$, which is assigned, e.g., to the function $\Fun{\x}{\raiseop{\PredZeroExc}}$.  More in general, our type effects can \emph{force} computational effects to be raised. 

\item  \label{ex:typing:pow} In \cref{ex:pow}, with $\fun{\chooseop}{}{\Bool}$, 
\begin{quoting}
%$\IsEffExp{\emptyset}{\myexp}{\Nat}{\SetOf{\chooseop}}$\\
$\IsWFExp{\emptyset}{\chooseincr}{\FunType{\Nat}{\SetOf{\chooseop^n\mid  n\geq 1}}{\Nat}}$\\
%$\IsEffExp{\emptyset}{\App{\chooseincr}{\val}}{\Nat}{\SetOf{\chooseop^\omega}}$ if $\IsWFExp{\emptyset}{\val}{\Nat}$\\
$\IsWFExp{\emptyset}{\choosedecr}{\FunType{\Nat}{\SetOf{\chooseop^n\mid   n\geq 0 }}{\Nat}}$
\end{quoting}
 Again, the effect of the first function forces non-determinism, differently from that of the second one. 
Apart from that, the two effects are very similar, even though calls of the first and second function always diverge and terminate, respectively.
Indeed, as usual, effect types only provide a static approximation of the computational effects. 

\item  \label{ex:typing:output} In \cref{ex:output}, with, for each  output location  $\loc$, $\fun{\writeop{\loc}}{\Nat}{\Unit}$,
\begin{quoting}
$\IsWFExp{\emptyset}{\writeincr}{\FunType{\Nat}{\SetOf{(\writeop{\loc}\cdot\writeop{\loc'})^\omega}}{\Unit}}$\\
$\IsWFExp{\emptyset}{\writedecr}{\FunType{\Nat}{\SetOf{(\writeop{\loc}\cdot\writeop{\loc'})^n\mid  n\geq0}}{\Unit}}$
\end{quoting}
 Here the difference between the effects of the two functions is even more significant: in the former, the sequence of two $\aux{write}$ calls is necessarily done infinitely many times, in the latter it can be done any arbitrary, yet finite,  positive  number of times. 
Moreover, in this case effects also provide an information on the the order among different $\aux{write}$ calls; for instance, here a  $\writeop{\loc}$ call should be always followed by a $\writeop{\loc'}$ call. 
\end{enumerate}
\end{example}
We discuss now how to express and prove type soundness.
Recall that the monadic operational semantics defined in \cref{sect:monadic-sem} constructs, on top of the one-step reduction:
\begin{itemize}
\item a finitary semantics $\fun{\finsem{-}}{\Exp}{\mfun\Res + \{\infty\}}$
\item an infinitary semantics  $\fun{\infsem{-}}{\Exp}{\mfun\Res}$
\end{itemize}
where $\Res = \Val + \Wrng$, with the latter modelling  a stuck computation. 
Hence, we expect a sound type-and-effect system to guarantee, first of all, that 
\begin{quoting}
(1) the (monadic) result of a well-typed expression is never $\wrng$
\end{quoting}
analogously to what we expect for a standard type system.
In the standard case, we also expect the result, if any, to be in agreement with the expression type. Here, since the expression has an effect as well, approximating the computational effects raised by its execution, we expect that
\begin{quoting}
(2) the monadic result, if any, is in agreement with the expression type and effect
\end{quoting}
In finitary semantics, (2) imposes nothing on diverging expressions, since they have no monadic result, whereas, in infinitary semantics, (2) is significant \mbox{for diverging expressions as well. }

To formally express (2), we need to derive, from the well-typedness predicates (one for each type and effect), analogous predicates on monadic results.  In the following section, this is achieved through a  \emph{predicate lifting} \cite{Jacobs16} $\mlift$, that is,  a way to lift, for every set $\X$, predicates over $X$ to predicates over $\mfun X$. Intuitively, $\mlift$ adds requirements on the computational effects,expressed by an effect type, that is, lifting provides an \emph{interpretation of effect types}.

% !TEX root =main.tex

\section{Monadic type-and-effect soundness}
\label{sect:monadic-ty}

The  standard technique  for proving type soundness with respect to a small-step operational semantics is  as a consequence, by a simple inductive argument, of progress and subject reduction properties  \cite{WrightF94}. 
In this section, we   introduce an analogous technique for  monadic operational semantics.  Notably, we express progress and subject reduction for the monadic one-step reduction,  and prove that they imply soundness.  
% In an  effectful setting, 
%we expect type soundness to give also guarantees on the computational effects raised during the evaluation process. 
We develop our technique for \emph{type-and-effect systems} \cite{Wadler98,NielsonN99,WadlerT03,MarinoM09,Katsumata14}, 
that is,  formal systems  providing an (over)approximation not only of the result of a computation, but also \mbox{of  its  computational effects. }

Following \cite{DagninoBZD20,Dagnino22}, a type system can be abstractly seen as a family of predicates over expressions and values indexed by types. 
%Intuitively, such predicates identify well-typed expressions and values with a given type. 
  In a type-and-effect system, predicates over expressions will be indexed not only by types but also by \emph{effect types}, describing the computational effects that  expressions can produce during their evaluation, as defined below.
\begin{definition}\label{def:mnd-type-system}
A \emph{type-and-effect system} 
$\TS = \ple{\Types,\MEff,\EWT,\VWT}$ for a language $\lang=\ple{\Exp,\Val,\valtoexp}$ consists of the following data: 
\begin{itemize}
\item a set $\Types$ of \emph{types}
\item an ordered monoid $\MEff = \ple{\Eff,\eord,\emul,\eun}$ of \emph{effect types}
\item for every $\ty\in\Types$ and $\ef\in\Eff$, predicates 
$\VWT[\ty]\subseteq\Val$ and $\EWT[\ty,\ef]\subseteq\Exp$  such that 
\begin{itemize}
\item $\ef\eord\ef'$ implies $\EWT[\ty,\ef]\subseteq\EWT[\ty,\ef']$ and 
\item $\valtoexp(\val) \in \EWT[\ty,\ef]$ iff $\val\in\VWT[\ty]$ and $\eun\eord\ef$
\end{itemize}
\end{itemize}
\end{definition}
The ordered monoid is a typical structure for effect systems \cite{NielsonN99,MarinoM09,Katsumata14}:
$\eun$ represents the absence of computational effects, 
$\ef_1\emul\ef_2$ represents the composition of computational effects described by $\ef_1$ and $\ef_2$, and 
$\ef_1\eord\ef_2$ states that the effect type $\ef_1$ is \mbox{more specific than $\ef_2$. }

The two families $\VWT$ and $\EWT$ are, for each index, predicates over values and expressions, respectively: 
$\VWT[\ty]$ is the set of values of type $\ty$, and 
$\EWT[\ty,\ef]$ is the set of expressions of type $\ty$ which may raise effects described by $\ef$. 
The first requirement, that is, monotonicity with respect to the order, states that the latter actually models if $\ef_1\eord\ef_2$, then $\ef_1$ is really more specific than $\ef_2$. 
The second requirement states that an expression which is the embedding of a value has the same type,  and an effect type which is not forcing any effect.\footnote{For instance, in \refItem{ex:typing}{exc}, we have $\SetOf{\epsilon}\eord\SetOf{\epsilon,\raiseop{\PredZeroExc}}$, whereas $\SetOf{\epsilon}\not\eord\SetOf{\raiseop{\PredZeroExc}}$.} 

Consider now an operational semantics $\Pair{\mnd}{\red}$, with $\mnd=\ple{\mfun,\mmul,\mun}$, and focus, e.g., on reduction from expressions to monadic expressions. To express type preservation, we should define, for each  $\ty$ and $\ef$, the monadic counterpart of $\EWT[\ty,\ef]$, being a predicate on $\mfun\Exp$. 
The key idea is to obtain such predicate by applying a \emph{predicate lifting} \cite{Jacobs16}, that is,  a way to lift, for every set $\X$, predicates over $\X$ to predicates over $\mfun \X$, adding requirements on the computational effects modeled by the monad. 
% In a sense, $\mlift$ can be seen as a \emph{predicate on the monad}. 
In our case, 
%the additional requirements are modeled by effect types, hence $\mlift$ is a family indexed over them. 
for each effect type $\ef$, the predicate lifting modularly models the meaning of $\ef$, that is, the computational effects approximated by $\ef$, independently from the set $X$ and the predicate $A$, \mbox{as formally detailed below. }

For a set $X$, we denote by $\PW(X)$ the poset of all subsets (a.k.a.\ predicates) on $X$, ordered by subset inclusion. 
For a function \fun{f}{X}{Y}, we have a monotone function \fun{\PW_f}{\PW(Y)}{\PW(X)}, given by the inverse image: 
for $A\subseteq Y$, $\PW_f(A) = \{ x \in X \mid f(x) \in A \}$.  That is, $\PW_f$ is a predicate transformer, giving, for each predicate $A$ on $Y$, the weakest condition elements of $X$ should satisfy to be mapped by $f$ in elements satisfying $A$. 
These data determine a functor \fun{\PW}{\Set\op}{\Pos}, where $\Pos$ denotes the category of posets and monotone functions. 

\begin{definition}[Interpretation of effect types] \label{def:mlift} 

Let $\mnd = \ple{\mfun,\mmul,\mun}$ be a monad, and $\MEff = \ple{\Eff,\eord,\emul,\eun}$ an ordered monoid of effect types.
Then, an \emph{interpretation} of $\MEff$ in $\mnd$
consists of a family $\mlift$ of monotone functions \fun{\mlift[\ef]_X}{\PW(X)}{\PW(\mfun X)}, 
\mbox{for every $\ef\in\Eff$ and set $X$, such that}
\begin{enumerate}[series=enum-mlift] 
\item\label{def:mlift:nat}  $\mlift[\ef]_X(\PW_f(A)) = \PW_{\mfun f}(\mlift[\ef]_Y(A))$, for every $A\subseteq Y$ and function \fun{f}{X}{Y}
\item\label{def:mlift:mon}  $\ef\eord\ef'$ implies $\mlift[\ef]_X(A)\subseteq\mlift[\ef']_X(A)$, for every $A\subseteq X$, 
\item\label{def:mlift:unit}  $A \subseteq \PW_{\mun_X}(\mlift[\eun]_X(A))$, for every $A\subseteq X$, 
\item\label{def:mlift:mul}  $\mlift[\ef]_{\mfun X}(\mlift[\ef']_X(A)) \subseteq \PW_{\mmul_X}(\mlift[\ef\emul\ef']_X(A))$, for every $A\subseteq X$. 
\end{enumerate} 
\end{definition}

 The family $\mlift = (\mlift[\ef])_{\ef\in\Eff}$ is a family of predicate liftings for the monad $\mnd$, indexed by effect types. 
%Each $\mlift[\ef]$, as said above, can be regarded as the semantics of the effect type $\ef$: 
For a subset $A\subseteq X$, the subset $\mlift[\ef]_X(A)\subseteq\mfun X$ contains monadic elements which agree with $A$ and whose computational effects are described by $\ef$.

\cref{def:mlift:nat} states that $\mlift[\ef]_X$ is natural in $X$, that is, for every $\ef\in\Eff$, we have a natural transformation 
\nt{\mlift[\ef]}{\PW}{\PW\circ\mfun\op}.\footnote{Here $\fun{\mfun\op}{\Set\op}{\Set\op}$ denotes the  functor defined exactly as $\mfun$ but \mbox{on the opposite category.}}
The naturality on $X$ ensures that the semantics of each effect type is independent from the specific set $X$, thus depending only on the functor $\mfun$. 

\cref{def:mlift:mon} states that $\mlift[\ef]_X$ is monotone with respect to the order on effects, that is, 
computational effects described by $\ef$ are also described by $\ef'$. 

\cref{def:mlift:unit} states that monadic elements in the image of $\mun_X$ contain computational effects described by $\eun$, that is, no computational effect. 

Finally, in \cref{def:mlift:mul} we consider elements of $\mfun^2\X$ whose computational effects are described by lifting predicates to $\mfun\X$ through $\ef'$,  and then by lifting through $\ef$.
By flattening such elements through $\fun{\mmul_\X}{\mfun^2\X}{\mfun\X}$ we obtain elements whose computational effects are described by the composition $\ef\emul\ef'$.

\begin{remark}\label{rem:graded-mnd} 
The monad $\mnd$ with an interpretation $\mlift$ determine a structure on the functor $\PW$, 
which can be described as a graded/parametric monad \cite{FujiiKM16} on $\PW$ in an appropriate 2-category (see e.g., \cite{DagninoR21}). 
Equivalently, $\mlift$ determines a graded/parametric monad above $\mnd$ \cite[Def.~2.6]{Katsumata14} 
along the fibration obtained from $\PW$ by the Grothendieck construction \cite{Grothendieck71}. 
\end{remark}
 
\begin{example}
\label{ex:exception-mlift}
Consider the exception monad $\Mnd{\ExceptFun[\ExSet]}$ of \cref{ex:exception-mnd} and the ordered monoid 
 \ple{\wp(\ExSet + \{\noEx\}),\subseteq,\cdot,\{\noEx\}} where 
$\excset_1\cdot\excset_2 = (\excset_1\setminus\{\noEx\})\cup\excset_2$, if $\noEx\in\excset_1$, and $\excset_1\cdot\excset_2 = \excset_1$, otherwise. 
For every $\excset\in\wp(\ExSet + \{\noEx\})$, set $X$, and $A\subseteq X$, 
the assigment 
\begin{quoting}
\begin{math}
\mlift[\excset]_X(A) = 
\begin{cases}
A + (\excset\setminus\{\noEx\})&\mbox{if}\ \noEx\in\excset\\
\excset&\mbox{otherwise}
\end{cases}
\end{math}
\end{quoting}
determines an interpretation of effect types into $\Mnd{\ExceptFun[\ExSet]}$. 
Intuitively,  the interpretation of $\excset$ requires exceptions possibly raised to be in $\excset$, and, if it is allowed that no exception be raised ($\noEx\in\excset$), requires the predicate $A$ to be satisfied.   
% \ple{\wp(\ExSet),\subseteq,\cup,\emptyset} of effect types. 
% For every element $\excset\in\wp(\ExSet)$, set $X$, and $A\subseteq X$, 
% the assignment $\mlift[\excset]_X(A) = A + E \subseteq \ExceptFun[\ExSet]X$ determines an interpretation of effect types into $\Mnd{\ExceptFun[\ExSet]}$. 
% Intuitively, elements in $\mlift[\excset]_X(A)$ are either elements of $A$ or exceptions in $\excset$. 
\end{example}

\begin{example}
\label{ex:pow-mlift}
Consider the powerset monad $\Mnd\PowerFun$ of \cref{ex:pow-mnd}.
\begin{enumerate}
\item\label{ex:pow-mlift:1}
Taking   
the ordered monoid \ple{\{0,1\},\le,\lor,0}, for every set $X$ and $A\subseteq X$, the following assignments determine two interpretations of effect types into $\Mnd\PowerFun$: 
\begin{quoting}
$\AllLift^1_{\!X}(A) = \{ B \in \PowerFun X \mid B \subseteq A \}$, \\
$\ExLift^1_X(A)  = {\{ B \in \PowerFun X \mid B = \emptyset \text{ or }B \cap A \ne \emptyset \}}$ and \\
$\AllLift^0_{\!X}(A) = \ExLift^0_X(A) = \{ B \in \PowerFun X  \mid B\subseteq  A \text{ and } \sharp B\leq 1\}$
\end{quoting}
  where $\sharp B$ is the cardinality of $B$.  Intuitively, in both cases, the interpretation of $0$  disallows non-determinism,  while 
the interpretation of $1$ requires the predicate $A$ to be always satisfied, according to $\AllLift$, and 
satisfied in at least one case, according to $\ExLift$. 
\item\label{ex:pow-mlift:2} Taking instead the ordered monoid \ple{\N\cup\{\infty\},\leq,\cdot,1}, \mbox{we can give a finer version of $\AllLift$:}
\begin{quoting}
$\AllLift^n_X(A) = \{ B\in\PowerFun X \mid B\subseteq A \text{ and } \sharp B \leq n \}$ \\ 
$\AllLift^\infty_X(A) = \{ B\in\PowerFun X \mid B \subseteq A \}$ 
\end{quoting}
In this way, we can quantify the level of non-determinism in terms of the maximum number of possible outcomes.  
\end{enumerate}
Similar interpretations can be defined for the list and subdistribution monads of \cref{ex:prob-mnd}. 
\end{example}

\begin{example}
\label{ex:output-mlift}
Consider the output monad $\Mnd\OutFun$ of \cref{ex:output-mnd} for the monoid \ple{A^\infty,\cdot,\elist} of possibly infinite words over $A$ and the ordered monoid 
\ple{\N\cup\{\infty\},\leq,+,0} of effect types. 
For a word $\sigma\in A^\infty$, we write $|\sigma|$ for its length, which is an element of $\N\cup\{\infty\}$. 
For every $n \in \N\cup\{\infty\}$, set $X$ and $A\subseteq X$, the assignment 
$\mlift[n]_X(A) = \{ \ple{\sigma,x}\in \OutFun X \mid x \in A, |\sigma| \leq n \}$ 
determines an interpretation of effect types into $\Mnd\OutFun$. 
 Intuitively, such interpretation imposes an upper bound (or none) to the length of the outputs.
\end{example}

\begin{example}
\label{ex:translate-mlift}
Let $\MEff$ be an ordered monoid of effect types and $\mlift$ an interpretation of $\MEff$ into  a monad $\mnd$. 
Let $\MEff'$ be another ordered monoid. 
To give an interpretation of $\MEff'$ into $\mnd$, it suffices to give a \emph{lax monoid homomorphism} \fun{f}{\MEff'}{\MEff}, that is, 
a monotone function \fun{f}{\ple{\Eff',\eord'}}{\ple{\Eff,\eord}} such that 
$\eun\eord f(\eun')$ and 
$f(\ef'_1)\emul f(\ef'_2) \eord f(\ef'_1\emul\ef'_2)$. 
Then, we can define an interpretation $\rho$ of $\MEff'$ into $\mnd$ by setting 
$\rho^{\ef'} = \mlift[f(\ef')]$ for all $\ef'\in\Eff'$. 
\end{example}

Let us fix 
a monadic operational semantics $\Pair{\mnd}{\red}$ for a language $\lang=\ple{\Exp,\Val,\valtoexp}$, a type-and-effect system $\TS = \ple{\Types,\MEff,\EWT,\VWT}$ for $\lang$, and an interpretation $\mlift$ of $\MEff$ into $\mnd$. 

Then,  we can formally state monadic progress and monadic subject reduction. 

\begin{definition}[Monadic Progress]\label{def:mnd-progress}
The type-and-effect system $\TS$ has \emph{monadic progress}  if 
$\expr\in\EWT[\ty,\ef]$ implies either $\expr = \valtoexp(\val)$ for some $\val\in\Val$, or 
$\expr\red\mexpr$ for some $\mexpr\in\mfun\Exp$.
\end{definition}

\begin{definition}[Monadic Subject Reduction] \label{def:mnd-sr}
The type-and-effect system $\TS$ has \emph{monadic subject reduction} if 
$\expr\in\EWT[\ty,\ef]$ and $\expr\red\mexpr$ imply $\mexpr \in \mlift[\ef_1]_{\Exp}(\EWT[\ty,\ef_2])$ for some $\ef_1\emul\ef_2\eord\ef$. 
\end{definition}

Monadic progress is standard: a well-typed expression either represents a value or can reduce. 
Monadic subject reduction, instead, takes into account effects: 
if an expression of type $\ty$ and effect $\ef$ reduces to a monadic expression $\mexpr$, then  $\mexpr$ ``has type $\ty$ and effect $\ef$'' as well, meaning that:  $\ef$ can be decomposed as $\ef_1\emul\ef_2$ and 
$\mexpr$ contains computational effects described by $\ef_1$ and expressions of type $\ty$ and effect $\ef_2$. 
In other words, the type $\ty$ is preserved and  the effect $\ef$ is an upper bound of the computational effects produced by the current reduction step, described by $\ef_1$, composed with those produced by future reductions, described by $\ef_2$. 

Our next step is expressing type-and-effect soundness. 
 In standard small-step semantics, soundness means that, starting from a well-typed expression, if termination, that is, an expression which cannot be reduced, is reached, then such expression should be a well-typed value.
In our monadic operational semantics, termination is conventionally represented by monadic results. Hence, an analogous statement is that, starting from a well-typed expression, if termination, that is, a monadic result, is reached, then this should be a well-typed result, meaning that is satisfies the lifting through the effect type of well-typedness of values.   

%This can be seen as an extension of monadic progress and  monadic subject reduction to the (finitary) semantics of expressions. 
% 
Again slightly abusing the notation, 
we will consider predicates $\VWT[\ty]$ on values also as predicates on results. 
Note that in particular  
$\wrng\notin\VWT[\ty]$ for all $\ty\in\Types$. 

\begin{definition}[Finitary type-and-effect soundness]\label{def:mnd-sound-fin}
The type-and-effect system $\TS$ is \emph{finitarily sound} if 
$\expr\in\EWT[\ty,\ef]$ and $\finsem{\expr} = \mres$ imply $\mres\in\mlift[\ef]_\Res(\VWT[\ty])$. 
\end{definition}
This notion of soundness is very general: 
whenever an expression of type $\ty$ and effect $\ef$ evaluates to a monadic result, 
this belongs to the interpretation of $\ef$ applied to (the image of) values of type $\ty$. 
Hence, the nature of the soundness property heavily depends on the interpretation $\mlift$ of effect types.
For instance, considering the interpretations for the powerset monad of \cref{ex:pow-mlift}, 
$\AllLift$ and $\ExLift$ induce induce a notion of must-soundness, and may-soundness, respectively: 
the former ensures that the evaluation of a well-typed expression never reaches $\wrng$, 
while the latter only that it either diverges or reaches at least a well-typed value. 

More specifically, it is not guaranteed that the monadic result is actually a monadic value. 
Formally, viewing $\mfun\Val$ as a subset of $\mfun\Res$,  the inclusion $\mlift[\ef]_\Res(\VWT[\ty])\subseteq \mfun\Val$ does not hold in general, as happens for instance with the $\ExLift$ interpretation. 
However, we can recover this property when the interpretation $\mlift$ enjoys an additional condition, as detailed below. 

Given a function \fun{f}{X}{Y}, the mapping \fun{\PW_f}{\PW(Y)}{\PW(X)}  has a left adjoint \fun{\img{f}}{\PW(X)}{\PW(Y)}, that is, 
a monotone function such that, for every $A\subseteq X$ and $B\subseteq Y$, $\img{f}(A)\subseteq B$ if and only if $A\subseteq \PW_f(B)$. 
The function $\img{f}$ is the direct image along $f$, that is, for $A\subseteq X$,  
$\img{f}(A) = \{ f(x) \mid x \in A \}$. 
Then, the following is an easy observation. 
\begin{proposition}\label{prop:must-sound}
If $\mlift$ satisfies 
\begin{enumerate}[resume=enum-mlift] 
\item\label{mlift:img}   $\mlift[\ef]_Y(\img{f}(A)) \subseteq \img{\mfun f}(\mlift[\ef]_X(A))$ for \fun{f}{X}{Y} and $A\subseteq X$
\end{enumerate}
then $\mlift[\ef]_\Res(\VWT[\ty]) \subseteq \mfun\Val$. 
\end{proposition}
\begin{proof}
Recall that we are implicitly using an inclusion \fun{\iota^\Res_\Val}{\Val}{\Res}. 
Making it explicit, the thesis becomes 
$\mlift[\ef]_\Res(\img{\iota^\Res_\Val}(\VWT[\ty])) \subseteq \img{\iota^\Res_\Val}(\mfun\Val)$. 
This follows from 
$ \mlift[\ef]_\Res(\img{\iota^\Res_\Val}(\VWT[\ty]))
    \subseteq \img{\mfun\iota^\Res_\Val}(\mlift[\ef]_\Val(\VWT[\ty]))
    \subseteq \img{\mfun\iota^\Res_\Val}(\mfun\Val)$. 
\end{proof}
Note that the inclusion in \cref{mlift:img} %
%of \cref{prop:must-sound} 
is actually an equality, since the converse always holds thanks to \cref{def:mlift:nat} of \cref{def:mlift}, 
as $\img{f}$ is the left adjoint of $\PW_f$. 
This ensures that 
a monadic result $\mres \in \mlift[\ef]_\Res(\VWT[\ty])$ contains only values of type $\ty$, hence, in particular, cannot contain $\wrng$. 
In fact, the $\ExLift$ interpretation of \cref{ex:pow-mlift} does not satisfy \cref{mlift:img} of \cref{prop:must-sound}. 
 
From now on, we assume that $\TS$ has monadic progress and monadic subject reduction, and our goal is to prove that they imply type-and-effect soundness. 

We first extend the type-and-effect system to configurations, defining $\CWT[\ty,\ef] \subseteq \Conf$  as 
\begin{quoting}
\begin{math}
\CWT[\ty,\ef] = \begin{cases}
\EWT[\ty,\ef] + \VWT[\ty] & \text{if $\eun\eord\ef$} \\ 
\EWT[\ty,\ef] & \text{otherwise} 
\end{cases}
\end{math}
\end{quoting}
Note that $\wrng$ is never a well-typed configuration, while  configurations which are 
values of type $\ty$ are well-typed with type $\ty$ and effect $\ef$ only when $\ef$ is larger than $\eun$, that is, 
the  type effect   does not force raising effects.  

Then, we should extend monadic progress and monadic subject reduction to the reduction relation $\tred$. 
However, since it is a total function, it trivially enjoys progress, 
hence, we only have to deal with subject reduction. 
In the proof, we also use monadic progress of the monadic reduction $\red$ on expressions to ensure that $\wrng$, which is ill-typed, \mbox{is not produced.  }

\begin{lemma}\label{lem:mnd-tred-sound}
If $\conf\in \CWT[\ty,\ef]$ and $\conf\tred\mconf$, then $\mconf  \in \mlift[\ef_1]_\Conf(\CWT[\ef_2])$ with $\ef_1\emul\ef_2\eord\ef$. 
\end{lemma}
\begin{proof}
We split cases on the shape of $\conf$. 

\proofcase{$\conf = \expr$} 
From $\expr \in \CWT[\ty,\ef]$ we derive $\expr \in \EWT[\ty,\ef]$. 
By monadic progress, either $\expr = \valtoexp(\val)$ or $\expr\red\mexpr$. 
In the former case, $\mconf = \mun_\Conf(\val)$ and, by \cref{def:mnd-type-system}, $\eun\eord\ef$ and $\val\in\VWT[\ty]$. 
Hence, the thesis follows by \refItem{def:mlift}{unit}, taking $\ef_1 = \eun$ and $\ef_2 = \ef$. 
In the latter case, $\mconf = \metc\mexpr = \mfun\iota^\Conf_\Exp(\mexpr)$ and, by monadic subject reduction,  
$\mexpr \in\mlift[\ef_1]_\Exp(\EWT[\ty,\ef_2])$, with $\ef_1\emul\ef_2\eord\ef$. 
From $\EWT[\ty,\ef_2] = \PW_{\iota^\Conf_\Exp}(\CWT[\ty,\ef_2])$, by \refItem{def:mlift}{nat},  
\begin{quoting}
\begin{math}
\mexpr \in \mlift[\ef_1]_\Exp(\EWT[\ty,\ef_2]) 
  = \mlift[\ef_1]_\Exp(\PW_{\iota^\Conf_\Exp}(\CWT[\ty,\ef_2]))
  = \PW_{\mfun\iota^\Conf_\Exp}(\mlift[\ef_1]_\Conf(\CWT[\ty,\ef_2]))
\end{math}
\end{quoting}
which implies that 
$\metc\mexpr = \mfun\iota^\Conf_\Exp(\mexpr) \in \mlift[\ef_1]_\Conf(\CWT[\ty,\ef_2])$, as needed. 

\proofcase{$\conf = \res$} 
Since $\res\in\CWT[\ty,\ef]$, we have $\res\ne\wrng$, hence $\res = \val$ and this implies that $\val\in\VWT[\ty]$ and $\eun\eord\ef$. 
By definition of $\tred$, we also know that $\mconf = \mun_\Conf(\val)$, hence, the thesis follows from \refItem{def:mlift}{unit} taking $\ef_1 = \eun$ and $\ef_2 = \ef$. 
\end{proof}

Then, we obtain the following result, showing a form of soundness for the multistep reduction on configurations. 
\begin{theorem}\label{thm:mnd-conf-sound} 
If $\conf\in\CWT[\ty,\ef]$ and $\conf\tredstar\mconf$, then $\mconf\in\mlift[\ef_1]_\Conf(\CWT[\ty,\ef_2])$ with $\ef_1\emul\ef_2\eord\ef$. 
\end{theorem}
\begin{proof}
By induction on rules defining $\tredstar$. 

\proofcase{\refToRule{refl}}
We have $\mconf = \mun_\Conf(\conf)$. 
By \refItem{def:mlift}{unit},  
$\conf \in \CWT[\ty,\ef] \subseteq \PW_{\mun_\Conf}(\mlift[\eun]_\Conf(\CWT[\ty,\ef]))$, 
which implies 
$\mconf = \mun_\Conf(\conf) \in \mlift[\eun]_\Conf(\CWT[\ty,\ef])$. 
This proves the thesis since $\eun\emul\ef\eord\ef$. 

\proofcase{\refToRule{step}} 
We know that $\conf\tredstar \mconf_1$ and $\mconf = \mconf_1\mbind\tredfun$. 
By induction hypothesis,  
${\mconf_1 \in \mlift[\ef_1]_{\Conf}(\CWT[\ty,\ef_2])}$ with $\ef_1\emul\ef_2\eord \ef$. 
By \cref{lem:mnd-tred-sound}, we derive 
$\CWT[\ty,\ef_2] \subseteq \PW_{\tredfun}(\mlift[\ef'_1]_\Conf(\CWT[\ty,\ef'_2]))$ with $\ef'_1\emul\ef'_2\eord\ef_2$. 
Then, using \cref{def:mlift:nat,def:mlift:mul} of \cref{def:mlift}, we have 
\vspace{-0.1cm}
\begin{align*} 
\mconf_1\in \mlift[\ef_1]_\Conf(\CWT[\ty,\ef_2])
  & \subseteq \mlift[\ef_1]_\Conf(\PW_\tredfun(\mlift[\ef'_1]_\Conf(\CWT[\ty,\ef'_2])))
    = \PW_{\mfun\tredfun}(\mlift[\ef_1]_{\mfun\Conf}(\mlift[\ef'_1]_\Conf(\CWT[\ty,\ef'_2]))) 
\\[-1ex]
  & \subseteq \PW_{\mfun\tredfun}(\PW_{\mmul_\Conf}(\mlift[\ef_1\emul\ef'_1]_\Conf(\CWT[\ty,\ef'_2])))
    = \PW_{\mmul_\Conf\circ\mfun\tredfun}(\mlift[\ef_1\emul\ef'_1]_\Conf(\CWT[\ty,\ef'_2]))
\\[-1ex]
  & = \PW_{\mkl{(\tredfun)}}(\mlift[\ef_1\emul\ef'_1]_\Conf(\CWT[\ty,\ef'_2]))
\end{align*} 
\vspace{-0.1cm}
This implies that 
$\mconf = \mkl{(\tredfun)}(\mconf_1) \in \mlift[\ef_1\emul\ef'_1]_\Conf(\CWT[\ty,\ef'_2])$, 
hence the thesis follows observing that ${(\ef_1\emul\ef'_1)\emul\ef'_2 \eord \ef_1\emul\ef_2\eord \ef}$. 
\end{proof}

\begin{corollary}[Finitary type-and-effect soundness]\label{cor:mnd-fin-sound}
If $\expr\in\EWT[\ty,\ef]$ and $\finsem\expr = \mres$ \mbox{then $\mres\in\mlift[\ef](\VWT[\ty])$. }
\end{corollary}
\begin{proof}
From $\expr\in\EWT[\ty,\ef]$,  $\expr\in\CWT[\ty,\ef]$ and, 
from $\finsem\expr = \mres$, $\expr\tredstar\mrtc\mres$. 
By \cref{thm:mnd-conf-sound}, we obtain $\mrtc\mres = \mfun\iota^\Conf_\Res(\mres)  \in \mlift[\ef_1]_\Conf(\CWT[\ty,\ef_2])$ with $\ef_1\emul\ef_2\eord\ef$. 
By \refItem{def:mlift}{nat},  
$\mres\in\PW_{\mfun\iota^\Conf_\Res}(\mlift[\ef_1]_\Conf(\CWT[\ty,1ef_2])) 
  = \mlift[\ef_1]_\Res(\PW_{\iota^\Conf_\Res}(\CWT[\ty,\ef_2]))$. 
We distinguish two cases. 

\proofcase{$\eun\eord\ef_2$}
We have $\CWT[\ty,\ef_2] = \EWT[\ty,\ef_2] + \VWT[\ty]$, hence 
$\PW_{\iota^\Conf_\Res}(\CWT[\ty,\ef_2]) = \VWT[\ty]$. 
Since $\ef_1 = \ef_1\emul\eun \eord \ef_1\emul\ef_2\eord\ef$, 
by \refItem{def:mlift}{mon}, we get 
$\mres \in \mlift[\ef_1](\VWT[\ty]) \subseteq \mlift[\ef](\VWT[\ty])$, as needed. 

\proofcase{$\eun\not\eord\ef_2$}
We have $\CWT[\ty,\ef_2] = \EWT[\ty,\ef_2]$, hence 
$\PW_{\iota^\Conf_\Res}(\CWT[\ty,\ef_2]) = \emptyset$. 
Using \cref{def:mlift:nat,def:mlift:mul,def:mlift:mon} of \cref{def:mlift} and the monad laws, we have
\vspace{-0.1cm}
\begin{align*}
\mres \in \mlift[\ef_1]_\Res(\emptyset) 
  & = \mlift[\ef_1]_\Res(\PW_{\mun_\Res}(\emptyset)) 
    = \PW_{\mfun\mun_\Res}(\mlift[\ef_1]_{\mfun\Res}(\emptyset)) 
\\[-1ex]
  & \subseteq \PW_{\mfun\mun_\Res}(\mlift[\ef_1]_{\mfun\Res}(\mlift[\ef_2]_\Res(\VWT[\ty]))) 
    \subseteq \PW_{\mfun\mun_\Res}(\PW_{\mmul_\Res}(\mlift[\ef_1\emul\ef_2]_\Res(\VWT[\ty]))) 
\\[-1ex]
  & = \PW_{\mmul_\Res\circ\mfun\mun_\Res}(\mlift[\ef_1\emul\ef_2]_\Res(\VWT[\ty])) 
    = \mlift[\ef_1\emul\ef_2]_\Res(\VWT[\ty]) 
    \subseteq \mlift[\ef]_\Res(\VWT[\ty]) 
\end{align*}
\vspace{-0.3cm}
and this proves the thesis. 
\end{proof}

\cref{cor:mnd-fin-sound} states that monadic progress and monadic subject reduction imply soundness with respect to the finitary semantics. 
To state an analogous result for infinitary semantics, the  interpretation of effect types has to take into account \mbox{the additional structure of the monad. }

\begin{definition}\label{def:mlift-ord}
Let $\mnd = \ple{\mfun,\mord,\mmul,\mun}$ be an $\omega$-CPO-ordered monad. 
An  interpretation $\mlift$ of $\MEff$ in $\mnd$  is \emph{$\omega$-CPO-ordered} if, 
for every effect type $\ef\in\Eff$, set $X$, and $A\subseteq X$, we have 
\begin{enumerate}
\item\label{def:mlift-ord:mord} $\mbot_X \in \mlift[\ef]_X(A)$ and 
\item\label{def:mlift-ord:msup} for every $\omega$-chain $(\alpha_n)_{n\in\N}$ in $\mfun X$, $\alpha_n \in \mlift[\ef]_X(A)$ for all $n\in\N$ implies $\msup_{n\in\N} \alpha_n \in \mlift[\ef]_X(A)$. 
\end{enumerate}
\end{definition}
 For example, the interpretations in \cref{ex:pow-mlift} are $\omega$-CPO-ordered and also that in \cref{ex:output-mlift} can be turned into an $\omega$-CPO-ordered one if applied to the pointed output monad. Finally, the construction of \cref{ex:translate-mlift} applies to $\omega$-CPO-ordered interpretations as well.  

From now on, we assume that the monad $\mnd$ has an $\omega$-CPO-ordered structure and the  interpretation $\mlift$ of effect types  is $\omega$-CPO-ordered as well. 
\mbox{We define infinitary soundness as follows. }

\begin{definition}[Infinitary type-and-effect soundness]\label{def:mnd-sound-inf}
The type-and-effect system $\TS$ is \emph{infinitarily sound} if 
$\expr\in\EWT[\ty,\ef]$  implies $\infsem\expr \in\mlift[\ef]_\Res(\VWT[\ty])$. 
\end{definition}
Infinitary soundness states that the limit behaviour of an expression of type $\ty$ and effect type $\ef$ is a monadic result  belonging to the interpretation of $\ef$ applied to values of type $\ty$. 
Observations in \cref{prop:must-sound} applies to infinitary soundness as well. 

In order to prove that monadic progress and monadic subject reduction imply infinitary soundness, 
we first need a simple property of the function $\mctr = \mklind{\mctr}{0}$, introduced at page \pageref{res},  which is at the basis of the definition of the infinitary semantics. 

\begin{lemma}\label{lem:mnd-mctr0}
If $\mconf \in \CWT[\ty,\ef]$ then $\mctr_0(\conf) \in \mlift[\ef]_\Res(\VWT[\ty])$. 
\end{lemma}
\begin{proof}
We split cases on the shape of $\conf$. 

\proofcase{$\conf = \expr$}
We have $\mctr_0(\conf) = \mbot_\Res$ that belongs to $\mlift[\ef]_\Res(\VWT[\ty])$ by \refItem{def:mlift-ord}{mord}. 

\proofcase{$\conf = \res$}
Since $\conf\in\CWT[\ty,\ef]$, we have $\eun\eord\ef$ and $\conf = \res = \val \in \VWT[\ty]$. 
We also know that $\mctr_0(\conf) = \mun_\Res(\val)$. 
By \cref{def:mlift:unit,def:mlift:mon} of  \cref{def:mlift}, we have 
$\res \in \VWT[\ty] 
  \subseteq \PW_{\mun_\Res}(\mlift[1]_\Res(\VWT[\ty]))
  \subseteq \PW_{\mun_\Res}(\mlift[\ef]_\Res(\VWT[\ty]))$, 
thus proving that $\mctr_0(\conf) = \mun_\Res(\val) \in \mlift[\ef]_\Res(\VWT[\ty])$, as needed. 
\end{proof}

\begin{theorem}[Infinitary type-and-effect soundness] \label{thm:mnd-sound-inf}
If $\expr \in \EWT[\ty,\ef]$ then $\infsem\expr \in \mlift[\ef]_\Res(\VWT[\ty])$. 
\end{theorem}
\begin{proof}
By \refItem{def:mlift-ord}{msup} it suffices to show that $\infsem[n]{\expr} \in \mlift[\ef]_\Res(\VWT[\ty])$ for all $n \in \N$. 
%Consider  $n\in\N$. 
We know that $\infsem[n]{\expr} = \mctr(\mconf)$ where $\mconf = \tredfunind{n}(\expr)$. 
By \cref{prop:mnd-n-step}, we also know that $\expr \tredstar \mconf$. 
Since $\expr\in\EWT[\ty,\ef]$, we also have $\expr\in\CWT[\ty,\ef]$, hence, 
by \cref{thm:mnd-conf-sound}, we have
$\mconf \in \mlift[\ef_1]_\Conf(\CWT[\ty,\ef_2])$ with $\ef_1\emul\ef_2\eord\ef$. 
Using \cref{lem:mnd-mctr0} and \cref{def:mlift:mon,def:mlift:nat,def:mlift:mul} of \cref{def:mlift}, we get 
\vspace{-0.1cm}
\begin{align*}
\mconf \in \mlift[\ef_1]_\Conf(\CWT[\ty,\ef_2]) 
  & \subseteq \mlift[\ef_1]_\Conf(\PW_{\mctr_0}(\mlift[\ef_2]_\Res(\VWT[\ty])))
    = \PW_{\mfun\mctr_0}(\mlift[\ef_1]_{\mfun\Res}(\mlift[\ef_2]_\Res(\VWT[\ty]))) 
\\[-0.5ex]
  & \subseteq \PW_{\mfun\mctr_0}(\PW_{\mmul_\Res}(\mlift[\ef_1\emul\ef_2]_\Res(\VWT[\ty]))) 
    = \PW_{\mmul_\Res\circ\mfun\mctr_0}(\mlift[\ef_1\emul\ef_2]_\Res(\VWT[\ty])) 
\\[-0.5ex]
  & = \PW_{\mctr}(\mlift[\ef_1\emul\ef_2]_\Res(\VWT[\ty])) 
    \subseteq \PW_{\mctr}(\mlift[\ef]_\Res(\VWT[\ty]))
\end{align*}
\vspace{-0.1cm}
because $\mctr = \mklind{\mctr}{0} = \mmul_\Res\circ\mfun\mctr_0$. 
This proves that 
$\infsem[n]{\expr} = \mctr(\mconf) \in \mlift[\ef]_\Res(\VWT[\ty])$\mbox{, as needed. }
\end{proof}

% !TEX root =main.tex
\section{Example of soundness proof}\label{sect:results}
We show an instance of the technique introduced in the previous section, by proving monadic progress (\cref{theo:mnd-red-progress}) and 
monadic subject reduction (\cref{theo:mnd-subject-reduction}), hence, type-and-effect soundness, for our example.
Recall that  monadic reduction in \cref{sect:lang}  is parametric on a \mbox{monad $\mnd$, } and, for each operation $\opg$ with arity $k$, a partial function 
${\pfun{\mrun{\opg}}{\Val^k}{\mfun\Val}}$.

The type-and-effect system defined in \cref{sect:lang} is an example of \cref{def:mnd-type-system}, where, omitting 
empty environments and environments in judgments for simplicity:
\begin{itemize}
\item  $\Types$ is the set of types $\T$ as in \cref{fig:typed}
\item  $\MEff = \ple{\Eff,\sube,\cdot,\{\eZero\}}$ where   $\Eff$ is the set of non-empty subsets of $\Sigma^\infty$ 
%, $\eZero$ is the singleton of the empty sequence and $\cdot$ is defined in \cref{sect:lang}  
\item $\EWT[\T,\eff](\e)$ iff $\IsEffGroundExp{\e}{\T'}{\eff'}$ for some $\T',\eff'$ such that $\SubTE{\TEff{\T'}{\eff'}}{\TEff\T\eff}$.
\item $\VWT[\T](\ve)$ iff $\IsWFGroundExp{\ve}{\T'}$ for some $\T'$ such that $\SubT{\T'}{\T}$
% \item $\EWT=\bigcup_{\T\in\Types,\eff\in\Eff}\EWT[\T,\eff]$ and $\VWT=\bigcup_{\T\in\Types}\VWT[\T]$.
\end{itemize}

In order to prove progress and subject reduction properties, we need a last parameter, that is, an interpretation $\mlift$ of effect types. Since the proof is parametric on the computational effects raised by operations, these two parameters should agree, as described below. 
\begin{quoting}
for each $\fun{\opg}{\T_1\ldots\T_n}{\T}$, and $\seq\ve$ such that $\IsWFGroundExp{\seq\ve}{\seq\T'}$  and  $\SubT{\seq\T'}{\seq\T}$\\
\refToRule{run}\BigSpace $\mrun{\opg}(\seq\ve)\in \mlift[\{\opg\}]_{\Val}(\VWT[\T])$
\end{quoting}

\begin{example}\label{ex:lifting}
We describe  interpretations of the effect types suitable
%, which are  (non-empty)  sets of possibly infinite sequences of operations, 
for the examples in \cref{sect:lang}. Such interpretations are defined by first mapping\footnote{In all the examples it is easy to see that the mapping is a lax monoid homomorphism.} the effect types into one of the ordered monoids in \cref{sect:monadic-ty}, and then taking the interpretation of the latter into the monad; in this way, as described in \cref{ex:translate-mlift}, we get an interpretation of the original effect types. In other words, for an instantiation of the calculus on specific monad and operations, sets of possibly infinite sequences could be reduced to  simpler effect types, as exemplified below. 
\begin{enumerate}
\item  In \cref{ex:exc}, we  reduce effect types to sets whose elements are \mbox{either exceptions or $\noEx$:} 
\begin{quoting}
$\effsem{\eZero}=\{\noEx\}$\\
$\effsem{\Cons{\raiseop{\exc}}{\alpha}} = \{\exc\}$\\
 $\effsem{\eff}=\bigcup_{\alpha\in\eff}\effsem{\alpha}$
\end{quoting}
 That is, effect types are mapped into those of \cref{ex:exception-mlift}, so that,  
if $\effsem{\eff}=\excset$, then
\begin{quoting}
\begin{math}
\mlift[\eff]_\Val(\VWT[\T]) = 
\begin{cases}
\VWT[\T] + \excset\ \mbox{if}\ \eZero \in\eff\\
\excset\ \mbox{otherwise}
\end{cases}
\end{math}
\end{quoting}
In this way, monadic values\footnote{We explain how the lifting works on values; of course the same applies to expressions and configurations.} (either values or exceptions)  are well-typed if they are either exceptions in $\excset$, or, if it is allowed that no exception be raised ($\noEx\in\excset$), well-typed values.   
Note that an expression such as, e.g., $\Seq{\raiseop{\exc}}{\raiseop{\exc'}}$, gets the effect (reducing to) $\{\exc\}$, highlighting the fact that $\raiseop{\exc'}$ cannot be reached. 
%
%if $\excset = \effsem{\eff}\setminus\{\noEx\}$, then
%$\mlift[\eff]_\Val(\VWT[\T]) = \VWT[\T] + \excset$, if $\eZero \in\eff$, and 
%$\mlift[\eff]_\Val(\VWT[\T]) = \excset$, otherwise. 
%In this way, monadic values\footnote{We explain how the lifting of well-typedness works on values; of course the same applies to expressions and configurations as well.} (either values or exceptions)  are well-typed 
%if they are either well-typed values, or exceptions in the set $\excset$, when $\eZero\in\eff$, i.e., there is the possibility of raising no effect, 
%otherwise, only exceptions in $\excset$ are considered well-typed. 
%Note that an expression such as, e.g., $\Seq{\raiseop{\exc}}{\raiseop{\exc'}}$, gets the effect (reducing to) $\{\exc\}$, highlighting the fact that $\raiseop{\exc'}$ cannot be reached. 
\item In \cref{ex:pow}, the simplest interpretation  is to reduce  effect types to either 0 or 1: 
\begin{quoting}
$\effsem{\eZero}=0$\\
$\effsem{\Cons{\chooseop}{\alpha}}=1$\\
 $\effsem{\eff}=1$ \mbox{if} $\effsem{\alpha}=1$ for some $\alpha\in\eff$, $0$ otherwise
\end{quoting}
That is, effect types are mapped into those of  \refItem{ex:pow-mlift}{1},  so that, if $\effsem{\eff}=0$, then 
\begin{quoting}
$\mlift[\eff]_\Val(\VWT[\T])=\mlift[0]_\Val(\VWT[\T]) = \{ \aux{V} \in \PowerFun \Val  \mid \aux{V}\subseteq  \VWT[\T] \text{ and } \sharp \aux{V}\leq 1\}$
\end{quoting}
If, instead, $\effsem{\eff}=1$, then we can choose
\begin{quoting}
either $\mlift[\eff]_\Val(\VWT[\T])=\AllLift^1_\Val(\VWT[\T])   = \{ \aux{V} \in \PowerFun \Val \mid \aux{V} \subseteq \VWT[\T] \}$\\
or $\mlift[\eff]_\Val(\VWT[\T])=\ExLift^1_\Val(\VWT[\T])  = {\{ \aux{V} \in \PowerFun \Val \mid \aux{V} = \emptyset \text{ or }\aux{V} \cap \VWT[\T] \ne \emptyset \}}$
\end{quoting}
In this way, monadic values (sets of values, representing possibile results of a computation) are well-typed with a type effect (reducing to) $0$ if they have at most one element, and this element, if any, is well-typed; in other words, the computation is deterministic.  On the other hand, they are well-typed with a type effect (reducing to) $1$ if all the values in the set are well-typed, or there is at least one well-typed value, respectively.

\item  A finer interpretation for \cref{ex:pow} is to reduce effect types to the monoid \ple{\N\cup\{\infty\},\leq,\cdot,1} of \refItem{ex:pow-mlift}{2}, 
thus controlling the level of non-determinism:. We set 
\begin{quoting}
$\effsem{\chooseop^n}=2^n$ for $n\in\N$, $\effsem{\chooseop^{ \omega }}=\infty$\\
$\effsem{\eff}=\sup\{ \effsem{\alpha}\mid\alpha \in\eff\}$
\end{quoting}
Indeed, each call can be seen as a node in a binary tree of choices. In this way, if $\effsem{\eff}=k\in\N\cup\{\infty\}$, then monadic values (sets of values, representing possibile results of a computation) are well-typed with $\eff$ if there are at most $2^k$ values, hence possible results, in the set and these are all well typed. 
% Note that, in this case, we need the multiplicative variant of the ordered monoid in \cref{ex:output-mlift},
% since the maximum number of choices in an expression $\Do{\x}{\e}{\e'}$ is the product of the maximum numbers in the two subexpressions.

\item In \cref{ex:output},  a possible interpretation of a sequence of $\writeop{\loc}$ is its length:
\begin{quoting}
$\effsem{\writeop{\loc}^n}=n$ for $n\in\N\cup\{\infty\}$\\
$\effsem{\eff}=\sup\{ \effsem{\alpha}\mid\alpha \in\eff\}$
\end{quoting}
That is, effect types are mapped into $\N\cup\{\infty\}$ as done in \cref{ex:output-mlift}, so that,  
 if $\effsem{\eff}=n$, then $\mlift[\eff]_\Val(\VWT[\T])=\mlift[n]_\Val(\VWT[\T]) = \{ \ple{\sigma,\val} \in \OutFun \Val  \mid \val\in\VWT[\T] \text{ and } \mid\sigma\mid\leq n\}$.
In this way, an upper bound (or none) is imposed on the length of the produced outputs.

\item A finer interpretation for \cref{ex:output} is obtained by taking effect types as they are,
 % in is that effect types are taken as they are,    FRASE PRECEDENTE
 and $\mlift[\eff]_\Val(\VWT[\T]) = \{ \ple{\sigma,\val}\in \OutFun \Val \mid \val \in \VWT[\T], \extract{\sigma}\in\eff \}$ where,
 % $\extract{\sigma}$ is the sequence of $\writeop{\loc}$ operations corresponding to $\sigma$, that is, 
  if $\sigma = \Pair{\loc_1}{n_1}\ldots\Pair{\loc_k}{n_k}$, then $\extract{\sigma}=\writeop{\loc_1}\ldots\writeop{\loc_k}$.  In this way, effect types can express properties about the order, or the fairness, in which $\aux{write}$ operations to different  output locations  can be performed. Similar sophisticated properties can be expressed in cases where different operations can be performed, e.g., reading and updating in the global state monad.
\end{enumerate}
\end{example}
 We state now monadic progress and monadic subject reduction for  the type-and-effect system in \cref{sect:lang}; as shown in \cref{sect:monadic-ty}, they imply  monadic soundness. 
We report only the proof of monadic subject reduction; other proofs and lemmas they depend on are given in  \cref{app:lambda-results}. 

\begin{theorem}[Monadic Progress]\label{theo:mnd-red-progress}
If  $\e\in \EWT[\T,\eff]$  then either $\e= \Ret\ve$
% for some $\ve\in\Val$, 
or 
$\e\red\mexpr$.
% for some $\mexpr\in\mfun\Exp$.
\end{theorem}

 The proof of monadic subject reduction uses the standard substitution lemma, and subject reduction for the pure relation $\purered$ defined in \cref{fig:pure-red}. 
Both properties do not involve any monadic ingredient, and are proved by standard techniques. 

\begin{lemma}[Substitution]\label{lem:substitution}
If $\IsEffExp{\Gamma,\TVar{\seq\T}{\seq\x}}{\e}{\T}{\eff}$ and  $\SubT{\seq{\T'}}{\seq{\T}}$ , then 
$\IsWFGroundExp{\seq\val}{\seq{\T'}}$ implies $\IsEffExp{\Gamma}{\Subst\e{\seq\val}{\seq\x}}{\T'}{\eff'}$ with $\SubTE{\TEff{\T'}{\eff'}}{\TEff{\T}{\eff}}$.
\end{lemma}

\begin{lemma}[Subject Reduction]\label{lem:subject-reduction}
\mbox{If $\IsEffGroundExp{\e}{\T}{\eff}$ and $\e\purered \e'$ then $\IsEffGroundExp{\e'}{\T'}{\eff'}$ with $\SubTE{\TEff{\T'}{\eff'}}{\TEff\T\eff}$. }
\end{lemma}

\begin{theorem}[Monadic Subject Reduction]\label{theo:mnd-subject-reduction}
If $\e\in \EWT[\T,\eff]$ and $\e\red\mexpr$ then $\mexpr \in \mlift[\eff_1]_{\Exp}(\EWT[\T,\eff_2])$ 
for some $\eff_1$ and $\eff_2$ such that 
$\SubE{\EComp{\eff_1}{\eff_2}}{\eff}$. 
\end{theorem}
\begin{proof}
From $\e\in \EWT[\T,\eff]$ we get $\IsEffGroundExp{\e}{\T'}{\eff' }$ 
%for some $\T'$, $\eff'$ such that  
and $\SubTE{\TEff{\T'}{\eff'}}{\TEff\T\eff}$.
By induction on the reduction rules of \cref{fig:monadic-red}.
\begin{description}
\item [\refToRule{pure}]
In this case  $\e\purered \e'$ and $\mexpr= \mun(\e')$. From $\IsEffGroundExp{\e}{\T'}{\eff' }$ and  \cref{lem:subject-reduction} we get 
$\IsEffGroundExp{\e'}{\T_1}{\eff_1}$ and $\SubTE{\TEff{\T_1}{\eff_1}}{\TEff{\T'}{\eff'}}$ and, by transitivity of $\subt$, $\e'\in\EWT[\T,\eff]$.
From \refItem{def:mlift}{unit} we derive   
$\mun(\e')\in \mlift[\{\eZero\}]_{\Exp}(\EWT[\T,\eff])$ with $\SubE{\EComp{\{\eZero\}}\eff}\eff$.
\item[\refToRule{effect}] 
In this case  $\e=\opg(\seq\ve)$ and $\mexpr=\Mmap{(\Ret{\ehole})}{\mval}$, with $\mval=\mrun{\opg}(\seq\ve)$.
\mbox{From rule  \refToRule{t-op},}
$\fun{\opg}{\T_1\ldots\T_n}{\T'}$ and $\IsWFGroundExp{\seq\ve}{\seq\T'}$ and $\SubT{\seq\T'}{\seq\T}$ and $\eff'=\{\opg\}$. 
 Hence, by rule \refToRule{run},  $\mval\in \mlift[\{\opg\}]_{\Val}(\VWT[\T'])$. 
Let \fun{f}{\Val}{\Exp} be defined by $f(\val)=\Ret\val$, then 
$\mexpr=  \Mmap{(\Ret{\ehole})}{\mval} = \mfun f (\mval)$.   
From \refItem{def:mlift}{nat}  and $\VWT[\T']= \PW_f(\EWT[T',\{\eZero\}])$
\begin{quoting}
$\mval \in  \mlift[\{\opg\}]_{\Val}(\VWT[\T'])= \mlift[\{\opg\}]_{\Val}(\PW_f(\EWT[T',\{\eZero\}])=
\PW_{\mfun f}(\mlift[\{\opg\}]_\Exp(\EWT[\T',\{\eZero\}]))
$
\end{quoting}
This implies $\mexpr= \mfun f (\val)\in\mlift[\{\opg\}]_\Exp(\EWT[\T',\{\eZero\}])$. 
 Since $\SubT{\T'}{\T}$, $\EWT[\T',\{\eZero\}]\subseteq\EWT[\T,\{\eZero\}]$, 
hence by monotonicity of $\mlift[\{\opg\}]_\Exp$, we get 
$\mexpr\in\mlift[\{\opg\}]_\Exp(\EWT[\T,\{\eZero\}])$, with 
$\EComp{\{\opg\}}{\{\eZero\}}=\SubE{\{\opg\}}{\eff}$.   
\item [\refToRule{ret}] 
In this case $\e=\Do{\x}{\Ret\ve}{\e'}$ and $\mexpr= \mun(\Subst{\e'}\ve\x)$. 
From rules \refToRule{t-do} and \refToRule{t-ret},
$\IsWFGroundExp{\ve}{\T_1}$ and  $\IsEffExp{\TVar{\T'_1}{\x}}{\e'}{\T'}{\eff'}$, 
with  $\SubT{\T_1}{\T'_1}$.  
% $\eff'=\EComp{\{\eZero\}}{\eff_1}=\eff_1$.
% From  the (standard) substitution lemma (in \cref{app:lambda-results})  
 By \cref{lem:substitution}, 
we get $\IsEffGroundExp{\Subst{\e'}{\ve}{\x}}{\T''}{\eff''}$ with $\SubTE{\TEff{\T''}{\eff''}}{\TEff{\T'}{\eff'}}$.
Hence, $\SubTE{\TEff{\T''}{\eff''}}{\TEff{\T}{\eff}}$ and so $\Subst{\e'}{\ve}{\x}\in \EWT[\T,\eff]$. 
Finally, from \refItem{def:mlift}{unit},   
$\mun(\Subst{\e'}\ve\x)\in \mlift[\{\eZero\}]_{\Exp}(\EWT[\T,\eff])$ with $\EComp{\{\eZero\}}\eff=\eff$.
\item [\refToRule{do}] 
In this case $\e=\Do\x{\e_1}{\e_2}$ and $\mexpr=\Mmap{(\Do\x{\ehole}{\e_2})}{\mexpr_1} $ and $\e_1\red\mexpr_1$. 
 From rule \refToRule{t-do},
$\IsEffGroundExp{\e_1}{\T_1}{\eff_1}$  and $\IsEffExp{\TVar{\T'_1}{\x}}{\e_2}{\T'}{\eff_2}$ 
with $\eff'=\EComp{\eff_1}{\eff_2}$ and  $\SubT{\T_1}{\T'_1}$.  
Hence, from $\e_1\in\EWT[\T_1,\eff_1]$, by induction hypothesis we get that 
$\mexpr_1\in \mlift[\eff'_1]_{\Exp}(\EWT[\T_1,\eff'_2])$ 
%for some $\eff'_1$ and $\eff'_2$ such that 
with $\SubE{\EComp{\eff'_1}{\eff'_2}}{\eff_1}$. 
Let \fun{f}{\Exp}{\Exp} be defined by $f(\hat\e)=\Do\x{\hat\e}{\e_2}$, hence  
$\mexpr = \mfun f (\mexpr_1)$. 
By rule \refToRule{t-do}, we know that $\hat\e \in \EWT[\T_1,\hat\eff]$ implies $f(\hat\e) \in \EWT[\T',\EComp{\hat\eff}{\eff_2}] \subseteq \EWT[\T,\EComp{\hat\eff}{\eff_2}]$, that is, 
$\EWT[\T_1,\hat\eff] \subseteq \PW_f(\EWT[\T,\EComp{\hat\eff}{\eff_2}])$. 
From \refItem{def:mlift}{nat} 
and monotonicity of $ \mlift[\eff'_1]_{\Exp}$ we get
\begin{quoting}
$\mexpr_1 \in \mlift[\eff'_1]_{\Exp}(\EWT[\T_1,\eff'_2])
  \subseteq \mlift[\eff'_1]_{\Exp}(\PW_{f}(\EWT[\T,\EComp{\eff'_2}{\eff_2}])) 
  = \PW_{\mfun f}(\mlift[\eff'_1]_\Exp(\EWT[\T,\EComp{\eff'_2}{\eff_2}]))
$
\end{quoting}
that is, $\mexpr = \mfun f(\mexpr_1) \in \mlift[\eff'_1]_\Exp(\EWT[\T,\EComp{\eff'_2}{\eff_2}])$, and we get the thesis since
${\SubE{\EComp{\eff'_1}{\EComp{\eff'_2}{\eff_2}}}{\SubE{\EComp{\eff_1}{\eff_2}}{\eff'}}}$. 
 
\end{description}
\end{proof}

The results hold for the core calculus in \cref{fig:syntax}, for an arbitrary family $\Sigma$ of operations. 
The calculus, the type system and the proofs can be modularly extended by just considering cases for additional constructs, as we will do in \cref{sect:handlers} for handlers. 
Extending the subtyping relation, instead,  requires some care to preserve the needed properties.\footnote{For instance, adding $\SubT{\Bot}{\T}$ for all $\T$ as in \cref{ex:exc} is sound since $\Bot$ is an empty type. }

% !TEX root =main.tex
\section{Handlers}\label{sect:handlers}

We extend $\lambdaEff$ with \emph{handlers}, showing how our framework can deal with more sophisticated language features and, at the same time, how proofs can be modularly extended. 
In particular, it is important to illustrate that monadic semantics can incorporate handlers.
Constructs and terminology are inspired by those for algebraic effects, see, e.g., \cite{Pretnar15}; however,  the approach is different since our calculus,  being based on generic effects,  \mbox{has no explicit continuations. }

The syntax is reported in \cref{fig:syntax-handlers}.  
\begin{figure}[th]
\begin{math}
\begin{grammatica}
\produzione{\e}{ \ldots\mid\With{\handler}{\e}}{expression with handler} \\ 
\produzione{\handler}{\Handler{\seq\hc}{\x}{\e}}{handler}\\
\produzione{\hc}{\HC{\opg}{\seq\x}{\e}{\mode}}{clause}\\
\produzione{\mode}{\Continue\mid\Stop}{mode}
\end{grammatica}
\end{math}
\caption{Syntax of  handlers}\label{fig:syntax-handlers} 
\end{figure}
A handler specifies a \emph{final expression}, and a sequence of \emph{clauses}, assumed to be a map, that is, there can be at most one clause for an operation. 
Such a clause, if any, handles a call of the operation by executing the clause expression. After that, the final expression  is either executed or not depending on the \emph{mode}, 
either $\Continue$ or $\Stop$, for ``continue'' and ``stop'', respectively. As illustrated in the following examples, a  $\Continue$-clause  replaces an effect with an alternative behaviour in a continuous manner, whereas in  $\Stop$-clauses  handling the computational effect interrupts the \mbox{normal flow of execution.  }

The pure reduction extended with handlers is shown in \cref{fig:pure-red-handlers}.
\begin{figure}[th]
%\begin{small}
\begin{math}
\begin{array}{c}
\handler=\Handler{\seq\hc}{\x}{\e'}
\\[2ex]
\NamedRule{with-do}{ }
{ \With{\handler}{\Do{\y}{\e_1}{\e_2}}\purered \WithLong{\seq\hc}{\y}{(\With{\handler}{\e_2})}{\e_1}}
{
%\e=\Do{\y}{\e_1}{\e_2} 
}
\\[3ex]
\NamedRule{with-ret}{ }
{ \With{\handler}{\Ret\ve}\purered \Do\x{\Ret\ve}{\e'} }
{
%\e=\Ret\ve
 } 
\\[3ex] 
\NamedRule{with-continue}
{ }
{ 
\With{\handler}{\opg(\seq\ve)} \purered \Do{\x}{\Subst{\e}{\seq\ve}{\seq\x}}{\e'}
}
{ 
%\e=\opg(\seq\ve)\\
\HC{\opg}{\seq\x}{\e}{\Continue}\in\seq\hc
}
\\[3ex]
\NamedRule{with-stop}
{ }
{ 
\With{\handler}{\opg(\seq\ve)}\purered \Subst{\e}{\seq\ve}{\seq\x}
}
{ 
% \e=\opg(\seq\ve)\\
\HC{\opg}{\seq\x}{\e}{\Stop}\in\seq\hc
}
\\[3ex]
\NamedRule{with-fwd}{ }
{\With{\handler}{\opg(\seq\ve)} \purered \Do{\x}{\opg(\seq\ve)}{\e'} }
{  
% \e=\opg(\seq\ve)\\
{\opg}\not\in\seq\hc
}
\\[3ex]
\NamedRule{with-ctx}{
  \e \purered \e'
}{ \With{\handler}{\e}  \purered \With{\handler}{\e'}}
{ }
\end{array}
\end{math}
%\end{small}
\caption{Pure reduction with handlers}\label{fig:pure-red-handlers} 
\end{figure}
The behaviour of an expression with handler depends on the shape of the handled expression.

In case of a  \lstinline{do} composition of two subexpressions,  
the \lstinline{do} is eliminated by reducing to the first subexpression with as final expression the 
second one; clauses are propagated to both the subexpressions. In case of a \lstinline{return}, the handler is eliminated by reducing to the \lstinline{do} composition 
of the handled expression and the final expression.

In case of an operation call, the behaviour depends on whether a matching clause is found or not. If it is found, then the clause expression is executed, after 
replacing parameters by arguments, as shown in rules \refToRule{with-continue} and \refToRule{with-stop}. In a $\Continue$-clause, the final 
expression is executed as well. If there is no matching clause, instead, the handler is eliminated, by reducing to the \lstinline{do} composition of the operation call 
and the final expression. The outcome is that the operation call is     forwarded to be possibly handled by an outer level.  gFinally, the contextual rule is as expected. 

To extend the type-and-effect system, we rely on \emph{filter functions} associated to handlers, which describe how they
 transform effects, by essentially replacing operations matching some clause with the effect of the clause expression. 
 
To this end, first we define a \mbox{\emph{(handler) filter} $\hfilter$} to be the information about transforming effects which can be extracted from a handler, as shown in the top section of \cref{fig:filters}. Then, given a filter $\hfilter$, we define the associated function $\fun{\FilterFEx{\hfilter}}{\EffSet}{\EffSet}$.
%,  where $\EffSet$ is the set of effects defined in \cref{sect:lang}.  
This function, as shown in the bottom section of \cref{fig:filters},  is obtained on top of the function $\fun{\FilterF{\hfilter}}{\Sigma^\infty}{\EffSet^\infty}$ which transforms a single possibly infinite sequence of operations into a possibly infinite sequence of effects. The latter is transformed into a unique effect by taking the possibly infinite concatenation of its elements, denoted $\fun{\InfConc{}}{\EffSet^\infty}{\EffSet}$. Finally, the function is extended to effects (sets of sequences) in the obvious way.

\begin{figure}[th]
\begin{math}
\begin{grammatica}
\produzione{\hfilter}{\HFilter{\seq\cfilter}{\eff}}{filter}\\
\produzione{\cfilter}{\CFilter{\opg}{\mode}{\eff}}{clause filter}\\[1ex]
\end{grammatica}
\end{math}
\hrule
\begin{math}
\begin{array}{ll}
\\[-1.5ex]
\hfilter=\HFilter{\CFilter{\opg_1}{\mode_1}{\eff_1} \dots \CFilter{\opg_n}{\mode_n}{\eff_n}}{\eff}\\[1ex]
\fun{\FilterF{\hfilter}}{\Sigma^\infty}{\EffSet^\infty}\  \mbox{coinductively defined by:}\\
\FilterFun{\epsilon}{\hfilter}=\eff\\
\FilterFun{\Cons{\opg_i}{\alpha}}{\hfilter}=\Cons{\eff_i}{\FilterFun{\alpha}{\hfilter}}&i\in 1..n, \mode_i=\Continue\\
\FilterFun{\Cons{\opg_i}{\alpha}}{\hfilter}=\eff_i&i\in 1..n, \mode_i=\Stop\\
\FilterFun{\Cons{\opg}{\alpha}}{\hfilter}=\Cons{\{\opg\}}{\FilterFun{\alpha}{\hfilter}}&\opg\neq\opg_i\ \mbox{for all}\ i\in 1..n\\[2ex]
\fun{\FilterFEx{\hfilter}}{\EffSet}{\EffSet}\\
\FilterFunEx{\eff}{\hfilter}=\bigcup_{\alpha\in\eff}\{\InfConc{\FilterFun{\alpha}{\hfilter}}\}
\end{array}
\end{math}
\caption{Filters}\label{fig:filters}
\end{figure}

In \cref{fig:typing-handlers} we show the typing rules for expressions with handlers.
\begin{figure}[th]
\begin{math}
\begin{array}{l}
\NamedRule{t-with}{
\begin{array}{l}
\IsEffExp{\Gamma}{\e}{\T}{\eff}\\
\IsWFHandler{\Gamma}{\T'}{\handler}{\T''}{\hfilter}
\end{array}
}{\IsEffExp{\Gamma}{\With{\handler}{\e}}{\T''}{\FilterFun{\eff}{\hfilter}}}
{
\SubT{\T}{\T'} 
} 
\\[3ex]
\NamedRule{t-handler}
{\begin{array}{l}
\IsEffExp{\Gamma,\TVar{\T}{\x}}{\e'}{\T'}{\eff'}\\
\IsWFClause{\Gamma}{\T''}{\hc_i}{\cfilter_i}
\end{array}
}{\IsWFHandler{\Gamma}{\T}{\Handler{\hc_1\ldots\hc_n}{\x}{\e'}}{\T''}{\HFilter{\cfilter_1\ldots\cfilter_n}{\eff'}}
}{
\SubT{\T'}{\T''} 
}
\\[3ex]
\NamedRule{t-continue}{
  \IsEffExp{\Gamma,\TVars{\T}{\x}} {\e} {\T''} {\eff'}
}{\IsWFClause{\Gamma}{\T'}{\HC{\opg}{\seq\x}{\e}{\Continue}}{\CFilter{\opg}{\Continue}{\eff'}}}
{
  \fun{\opg}{\seq\T}{\T} \\ 
  \SubT{\T''}{\T}  
}
\\[3ex]
\NamedRule{t-stop}{
  \IsEffExp{\Gamma,\TVars{\T}{\x}} {\e} {\T''} {\eff'}
}{\IsWFClause{\Gamma}{\T'}{\HC{\opg}{\seq\x}{\e}{\Stop}}{\CFilter{\opg}{\Stop}{\eff'}}}
{
  \fun{\opg}{\seq\T}{\T} \\ 
  \SubT{\T''}{\T'} 
}
\end{array}
\end{math}
\caption{Typing rules for handlers}\label{fig:typing-handlers}
\end{figure}
In rule \refToRule{t-with}, in order to typecheck an expression with handler, first we get the type and effect of the handled expression. The type is used to typecheck the handler, as  (subtype of the)  type of the parameter of the final expression, see rule \refToRule{t-handler}. Typechecking the handler we get a type, being that of the final expression, which will be the type of the whole expression. Moreover, we extract from the handler a filter, which is used to transform the effect $\eff$ of the handled expression, getting the resulting effect of the whole expression. In detail,  as formally described in \cref{fig:filters}, the filter transforms any sequence of operations in $\eff$ by replacing the first operation matching some clause, if any, with the effect of the clause expression; then, the remaining sequence is disregarded if the clause is $\Stop$, otherwise filtered in turn. If the sequence to be filtered is finite, and no matching $\Stop$-clause is found, then the final effect is appended in the end.

In rule \refToRule{t-handler}, as said above, the type on the left of the judgment is used as type of the parameter of the final expression, whose type will be returned by the handler. This type is also needed to typecheck  $\Stop$-clauses, see below. The filter extracted from the handler consists in a clause filter for each clause, and the effect of the final expression.

For each clause, the extracted filter consists of the operation name, mode, and effect of the expression, as shown in rules \refToRule{t-continue} and \refToRule{t-stop}.  A $\Continue$-clause is meant to provide alternative code to be executed before the final expression, hence the type of the clause expression should be (a subtype of) the return type of the operation. In a $\Stop$-clause, instead, the result of the clause expression becomes that of the whole expression with handler, hence the type of the former should be (a subtype of) the latter. 

\begin{example}
We show handlers for some of the previous examples. A handler of shape  $\Handler{\seq\hc}{\x}{\Ret{\x}}$ is abbreviated  by  $\seq\hc$. 
\begin{enumerate}\label{ex:handlers}
\item Set $\handler =\HC{\raiseop{\PredZeroExc}}{}{\Ret{\Zero}}{\Stop}$. Then
\begin{quoting}
$
\begin{array}{lcl}
\With{\handler}{\App{\predfun}{\Zero}}&\Redstar&
%\With{\handler}{\If{\IsZero{\Zero}}{\raiseop{\PredZeroExc}}{\Ret{\Zero}}}\\
%&\Red&\With{\handler}{\If{\True}{\raiseop{\PredZeroExc}}{\Ret{\Zero}}}\\
%&\Red&
\With{\handler}{\raiseop{\PredZeroExc}}\\
&\Red&\Ret{\Zero}\\
&\Red&\Zero
\end{array}
$
\end{quoting}
As shown in \refItem{ex:typing}{exc}, we get the judgment $\IsEffExp{\emptyset}{\App{\predfun}{\Zero}}{\Nat}{\SetOf{\epsilon,\raiseop{\PredZeroExc}}}$. On the other hand, with the handler we get
\begin{quoting}
$\IsEffExp{\emptyset}{\With{\handler}{\App{\predfun}{\Zero}}}{\Nat}{\SetOf{\epsilon}}$
\end{quoting}
since $\FilterFEx{\hfilter}(\SetOf{\epsilon,\raiseop{\PredZeroExc}})=\SetOf{\epsilon}$ where $\hfilter=\HFilter{\CFilter{\raiseop{\PredZeroExc}}{\Stop}{\SetOf{\epsilon}}}{\SetOf{\epsilon}}$ is the filter extracted from $\handler$. 
As the reader could expect, an $\Stop$-clause is appropriate in this case. 
With a $\Continue$-clause,  see rule \refToRule{t-continue}, the type of the clause expression should be (a subtype of) the return type of the operation, which is $\Bot$. Since no  value has  type $\Bot$, no value could be returned, as already noted in \cite{PlotkinP03}.\footnote{Hence, the clause expression could only be another $\aux{raise}$ or a diverging expression.}

\item 

Assuming the function $\fun{\aux{even}}{\,\Nat}{\Bool}$ checking the parity of a number, set
\begin{quoting}
$\handler_1 =\HC{\writeop{\loc'}}{\x}{\writeop{\loc}(\x)}{\Continue}$\\
$\handler_2 =\HC{\writeop{\loc'}}{\x}{\If{\Even{\x}}{\Ret{\x}}{\writeop{\loc}(\x)}}{\Continue}$
\end{quoting}
Then
\begin{small}
\begin{quoting}
$\infsem{\With{\handler_1}{\App{\writeincr}{\Zero}}}=\Pair{\Pair{\loc}{\Zero}\cdot \Pair{\loc}{\Zero}\cdot\Pair{\loc}{\nat{1}}\cdot\Pair{\loc}{\nat{1}}\cdot\ldots\cdot\Pair{\loc}{\nat{n}}\cdot\Pair{\loc}{\nat{n}}\cdot \ldots}{\mbot}$\\
$\infsem{\With{\handler_2}{\App{\writeincr}{\Zero}}}=\Pair{\Pair{\loc}{\Zero}\cdot\Pair{\loc}{\nat{1}}\cdot\Pair{\loc}{\nat{1}}\cdot\ldots\Pair{\loc}{\nat{2k}}\cdot\Pair{\loc}{\nat{2k+1}}\cdot\Pair{\loc}{\nat{2k+1}}\cdot\ldots}{\mbot}$\\[1ex]
$\finsem{\With{\handler_1}{\App{\writedecr}{\nat{n}}}}=
%\infsem{\With{\handler_1}{\App{\writedecr}{\nat{n}}}}=$\\$
\Pair{\Pair{\loc}{\nat{n}}\cdot\Pair{\loc}{\nat{n}}\cdot\ldots\cdot\Pair{\loc}{\Zero}\cdot\Pair{\loc}{\Zero}}{\unit}$\\
$\finsem{\With{\handler_2}{\App{\writedecr}{\nat{2k}}}}=
%\infsem{\With{\handler_1}{\App{\writedecr}{\nat{2k}}}}=$\\$
\Pair{\Pair{\loc}{\nat{2k}}\cdot\Pair{\loc}{\nat{2k-1}}\cdot\Pair{\loc}{\nat{2k-1}}\ldots\cdot\Pair{\loc}{\nat{1}}\cdot\Pair{\loc}{\nat{1}}\cdot\Pair{\loc}{\Zero}}{\unit}$
%\\
%\\
%$\finsem{\With{\handler_2}{\App{\writedecr}{\nat{2k+1}}}}=
%%\finsem{\With{\handler_1}{\App{\writedecr}{\nat{2k+1}}}}$
%\Pair{\Pair{\loc}{\nat{2k+1}}\cdot\Pair{\loc}{\nat{2k+1}}\cdot\Pair{\loc}{\nat{2k}}\cdot\Pair{\loc}{\nat{2k-1}}\cdot\Pair{\loc}{\nat{2k-1}}\ldots\cdot\Pair{\loc}{\nat{1}}\cdot\Pair{\loc}{\nat{1}}\cdot\Pair{\loc}{\Zero}}{\unit}$\\
%\infsem{\With{\handler_1}{\App{\writedecr}{\nat{n}}}}=$\\
%$\Pair{\Pair{\loc}{\nat{2k}}\cdot\Pair{\loc}{\nat{2k-1}}\cdot\Pair{\loc}{\nat{2k-1}}\ldots\cdot\Pair{\loc}{\nat{1}}\cdot\Pair{\loc}{\nat{1}}\cdot\Pair{\loc}{\Zero}}{\unit}$\\
\end{quoting}
\end{small}
In this case, a $\Continue$-clause is appropriate, since the aim is to continuously handle the $\writeop{\loc'}$ operation.
By the typing judgments shown in \refItem{ex:typing}{output}, we get 
\begin{quoting}
$\IsEffExp{\emptyset}{\App{\writeincr}{\Zero}}{\Unit}{\SetOf{(\writeop{\loc}\cdot\writeop{\loc'})^\omega}}$\\
$\IsEffExp{\emptyset}{\writedecr}{\Unit}{\SetOf{(\writeop{\loc}\cdot\writeop{\loc'})^n\mid n\geq1}}$
\end{quoting}
On the other hand, with the handler we get, with  $\produzioneinline{\alpha}{\epsilon\mid\writeop{\loc}}$ 
\begin{quoting}
$\IsEffExp{\emptyset}{\With{\handler_1}{\App{\writeincr}{\Zero}}}{\Unit}{\SetOf{(\writeop{\loc}\cdot\writeop{\loc})^\omega}}$\\
 $\IsEffExp{\emptyset}{\With{\handler_2}{\App{\writeincr}{\Zero}}}{\Unit}{\SetOf{\alpha\cdot\writeop{\loc})^\omega}}$\\
 $\IsEffExp{\emptyset}{\App{\writedecr}{\nat{n}}}{\Unit}{\SetOf{\alpha\cdot\writeop{\loc'})^n\mid n\geq1}}$
\end{quoting}
As already noted,  effect types only provide a static approximation of the computational effects; notably, in the last two judgments,  the effect type contains other sequences  besides the two which can be actually performed, depending on the argument.
\end{enumerate}
\end{example}

% !TEX root =main.tex
\smallskip
\noindent\textbf{Soundness for handlers }
The results of \cref{sect:results} can be extended to handlers. For monadic subject reduction we only need 
to show subject reduction for the newly introduced rules, since they are pure. The proofs of the results are in \cref{app:handlers-results}. 

\begin{lemma}[Monadic Progress for handlers]\label{theo:h-mnd-red-progress} Set $\e$ of shape $\With{\_}{\_}$. \\
If $\IsEffGroundExp{\e}{\T}{\eff}$ then
 $\e\red\mexpr$ for some $\mexpr\in\mfun\Exp$.
\end{lemma}

Subject reduction relies on the properties of the  functions associated to filters  defined
in \cref{fig:filters}. In order to state these properties, define {\em $\hfilter$ to be a subhandler of $\hfilter'$}, dubbed $\SubH{\hfilter}{\hfilter'}$, if  
%$\hfilter={\HFilter{\CFilter{\opg_1}{\mode_1}{\eff_1}\dots\CFilter{\opg_n}{\mode_n}{\eff_n}}{\eff}}$ and 
%$\hfilter'={\HFilter{\CFilter{\opg_1}{\mode_1}{\eff'_1}\dots\CFilter{\opg_n}{\mode_n}{\eff'_n}}{\eff'}}$

\begin{quoting}
$\hfilter={\HFilter{\CFilter{\opg_1}{\mode_1}{\eff_1}\dots\CFilter{\opg_n}{\mode_n}{\eff_n}}{\eff}}$\\
$\hfilter'={\HFilter{\CFilter{\opg_1}{\mode_1}{\eff'_1}\dots\CFilter{\opg_n}{\mode_n}{\eff'_n}}{\eff'}}$\\
$\SubE{\eff}{\eff'}$ and $\SubE{\eff_i}{\eff'_i}$ for all $i\in i..n$
\end{quoting}

\begin{lemma}[Properties of $\FilterFEx\hfilter$]\label{lem:hfilter}\
\begin{enumerate}
\item \label{lem:hfilter:one} If $\SubE{\eff}{\eff'}$ and  $\SubH{\hfilter}{\hfilter'}$, then $\SubE{\FilterFunEx{\eff}{\hfilter}}{\FilterFunEx{\eff'}{\hfilter'}}$.
\item \label{lem:hfilter:two} $\FilterFunEx{\EComp{\eff_1}{\eff_2}}{\HFilter{\seq\cfilter}{\eff}}\supseteq\FilterFunEx{{\eff_1}}{\HFilter{\seq\cfilter}{\eff'}}$ where $\eff'=\FilterFunEx{{\eff_2}}{\HFilter{\seq\cfilter}{\eff}}$
\end{enumerate}
\end{lemma}

\begin{lemma}[Subject Reduction for handlers]\label{lem:h-subject-reduction} Set $\e$ of shape $\With{\_}{\_}$. \\
If $\IsEffGroundExp{\e}{\T}{\eff}$ and
 $\e\purered \e'$,
 then $\IsEffGroundExp{\e'}{\T'}{\eff'}$
such that  $\SubTE{\TEff{\T'}{\eff'}}{\TEff\T\eff}$. 
\end{lemma}

\cref{lem:h-subject-reduction} is proved, as customary, by induction on the reduction rules of \cref{fig:pure-red-handlers}. \cref{lem:hfilter:one} of \cref{lem:hfilter} is used
for the case of rule \refToRule{with-ctx} and \cref{lem:hfilter:two} for the one of rule \refToRule{with-do}, where the effects of the
second subexpression of the \lstinline{do} construct must be accounted for, after the reduction,  in the effects of  the final expression 
of the handler of its first subexpression.

% !TEX root =main.tex

\section{Related work and conclusion}
\label{sect:related}

\smallskip
\noindent\textbf{Monadic semantics} 
The idea that monads can model computational effects in programming languages goes back to the pioneering Moggi's work \cite{Moggi89,Moggi91}. 
He showed that one can use (strong) monads to organise the denotational semantics of effectful languages, interpreting impure expressions as functions (actually  arrows of an arbitrary category) returning monadic values, which can be sequenced by Kleisli composition. 
However, the structure of a monad does not include any operation for actually raising computational effects, which thus need to be defined ad-hoc in specific instances.  
Moreover, monads are difficult to combine, requiring non trivial notions like monad transformers \cite{LiangHJ95,JaskelioffM10}. 

To overcome these difficulties and make the model closer to the syntax, 
 Plotkin and Power  \cite{PlotkinP01,PlotkinP02,PlotkinP03} introduced algebraic effects which, instead, 
explicitly consider operations to raise computational effects. 
These can be interpreted by additional structure on the monad and, moreover, when equipped with an equational theory, 
they actually determine a monad, which provides a syntactic model for the language. 
Thus, one reduces the problem of combining monads to the much easier problem of combining theories \cite{HylandPP06}, 
 greatly increasing modularity. 

An alternative, essentially equivalent, way of interpreting algebraic operations is by means of runners, a.k.a. comodels 
\cite{PowerS04,PlotkinP08,Uustalu15,AhmanB20}. Roughly, runners describe how operations are executed by the system, that is, 
how they transform the environment where  they are  run. This essentially amounts to giving an interpretation of operations in the state monad. 
More general runners, where the system is modelled in a more expressive way, are considered by \cite{AhmanB20}, where the state monad is combined with errors and system primitive operations. 

On the operational side, algebraic effects are typically treated as uninterpreted operations, that is, 
the evaluation process just builds a tree  of operation calls \cite{JohannSV10,SimpsonV20,XiaZHHMPZ20,Pretnar15}. 
Monadic operational semantics for $\lambda$-calculi with algebraic effects are also considered, mainly in the form of a monadic definitional interpreter (see, e.g., \cite{LiangHJ95,DaraisLNH17,Gavazzo18,DalLagoG22,DagninoG24}). 
That is, they directly define a function from expressions to monadic values, which essentially corresponds to our infinitary semantics. 
Small-step approaches are also considered by \cite{GavazzoF21,GavazzoTV24}. 
The former %shares with us the aim of developing a general theory independent from specific languages. However, it 
tackles a different problem, that is, studying monadic rewriting systems, which require the use of  sophisticated relational techniques, and thus restricts the class of available monads.
%(e.g., the distribution monad must be replaced by a multidistribution monad). 
We can avoid these difficulties
%working with arbitrary monads, 
since we focus on deterministic rewriting, which can be addressed using just sets and functions. 
The latter, instead, studies a specific calculus where, as already noticed, the way sequences of steps are constructed is very close to ours; however, they do not need to introduce  $\wrng$  in configurations, as type errors are prevented syntactically.

%\smallskip
\noindent\textbf{Type-and-effect systems} 
Type-and-effect systems, or simply effect systems \cite{Wadler98,NielsonN99,WadlerT03,MarinoM09,Katsumata14}, 
are the most popular way of statically controlling computational effects. Many have been designed for specific notions of computational effect and implemented in mainstream programming languages, the most  well-known  being the  mechanism of Java checked exceptions. 
Katsumata  \cite{Katsumata14} recognized that effect systems share a common algebraic structure, notably they form an ordered monoid, and gave them denotational semantics through parametric monads, using a structure equivalent to our notion of interpretation \mbox{(see \cref{rem:graded-mnd}). }
%
%The effect system of our calculus is inspired by  that for algebraic effects \cite{BauerP14,BauerP15,Pretnar15}, also implemented in the Eff programming language \cite{EffLang}. 
%However, we go beyond both of them: 
%with respect to the former, we track all calls to magic methods and not just exceptions, 
%with respect to the latter, we enrich the effects with type information, not using just the name of magic methods,  and 
%with respect to both, we allow effects to be generic and to depend on the effect of other methods, which is crucial in the object-oriented context where the implementation of methods can be  deferred  to subtypes. 

%\smallskip
\noindent\textbf{Effect handlers} 
Plotkin and Pretnar \cite{PlotkinPretnar09,PlotkinPretnar13} introduced effect handlers as a generalisation of exception handling mechanisms. 
They are an extremely powerful programming abstraction, allowing to describe the semantics of algebraic operations in the language itself, thus enabling the simulation of several effectful programs, such as stream redirection or cooperative concurrency \cite{PlotkinPretnar13,KammarLO13,BauerP15,Pretnar15}. 
When a call to an algebraic operation is caught, the alternative code can resume the original computation using a form of continuation-passing style. 
%Another difference with exception handling is that all occurrences of an operation within the scope of an handler are actually handled. 
Other forms of handlers have  been considered, notably, shallow handlers  \cite{KammarLO13,HillerstromL18}, where only the first call to an operation is handled. 
Our handlers are inspired by those for algebraic effects, see, e.g., \cite{Pretnar15}; however,  the approach is different since our calculus  \mbox{has no explicit continuations. }

% !TEX root =main.tex

\smallskip
\noindent\textbf{Summary}
In  the research on   foundations of programming languages,  it is a routine task to describe execution through a small-step reduction, and prove progress and subject reduction for the type system. Can this be smoothly combined with the long-established approach where computational effects are modularly modeled by a monad, so to enjoy all the advantages of separation of concerns? The answer provided in this paper is yes. Notably, we provide a meta-theory defining abstract notions of monadic small-step semantics and type-and-effect system, and prove that type-and-effect soundness is implied by progress and subject reduction properties, with an inductive argument 
 similar  to the standard one. 
 
This overall achievement relies on two key specific contributions. On one hand, we provide a canonical way to construct, on top of a monadic reduction, a small-step operational semantics where computations, even though always represented by infinite sequences, can be distinguished as either non-terminating, or successfully terminating, or stuck.
 On the other hand, we provide a formal model of the ``meaning'' of effect types, independent in principle from the underlying language and  type system.

\smallskip
\noindent\textbf{Discussion and future work}

  The way we define the ``transitive closure'' of a monadic reduction, which is a relation from a set to a different one,  is similar, as said, to that proposed by \cite{GavazzoTV24}. Notably, such reduction is assumed to be deterministic,  since starting from an arbitrary relation would require a relational extension of the monad \cite{ Barr70}. 
Confluence as well would require strong assumptions on the monad, notably some form of commutativity, ruling out most of the relevant examples.
Moreover, the aim here is to prove soundness for programming languages, which typically adopt a deterministic evaluation strategy. Differently from \cite{GavazzoTV24}, we provide a language independent definition; moreover, whereas they consider an intrinsically total reduction, in this paper, as mentioned above, we address the additional problem to characterize stuck computations, as needed to express soundness. 

In our framework, non-termination is always possibile, rather than be considered as an effect. This is essentially a choice we made, possibly influenced by the fact that in standard soundness we have three possible outcomes: non-termination, termination with a value, and stuck. The coinductive Delay monad \cite{Capretta05} could be an alternative approach to define the infinitary semantics, assuming a way to be combined with the monad modeling computational effects, that is, a distributive law. The relationship between these two approaches, as far as we know, is not clear, and is an interesting direction to be investigated. 

Our definition of $\omega$-CPO-ordered monad is given for monads on $\Set$. A challenging and relevant problem is to consider a category different from $\Set$; our feeling is that the notion could be generalized by considering a monad $\mnd$ on a category $\ct{C}$ such that the Kleisli category $\ct{C}^\mnd$ is CPO-enriched.

In this paper, where the focus is different, we did not study decidability of the type-and-effect system; we did not even provide a syntactic representation of effects, which are considered semantic entities, notably possibly infinite sets of possibly infinite sequences. Of course decidability is a very important issue to be investigated; the first step should then be to choose a finite representation, e.g., by means of a system of guarded equations.

We  illustrated our approach by a lambda-calculus with generic effects. Clearly, it would be important to investigate how other calculi can be  formalized as to take advantage of  the meta-theory.  Notably, we plan to apply the approach \mbox{to an object-oriented calculus. } 
Moreover, here we considered non-standard handlers, as our calculus is based on generic effects and so it does not use explicit continuations. 
Hence, it would be nice to investigate the precise relationship between them and handlers for algebraic effects used in the literature.  
It would also be  interesting to allow the interpretation of operations to return monadic expressions, rather than monadic values. 
This would enable a more interactive behaviour with the system; for instance, the semantics of an operation, instead of returning an unrecoverable error,  could  return a call to  the operation $\raiseop{\exc}$, which then could be handled by the program.

On the side of the meta-theory, 
 one soon realizes that the proofs of (monadic) progress and subject reduction all have a similar structure: 
they are carried out by inductive arguments relying on inversion and substitution properties of the operational semantics and the type system. 
A natural question is thus whether this common structure can be abstracted in our meta-theory. 
This is indeed the case,  and we are currently working on such proof technique, which  requires  considering a more structured notion of language.   
Another property of  a type system one could be interested  in   is its  completeness, whose proof typically relies on the subject-expansion property. 
The latter could be formulated in  our monadic setting and we conjecture that, together with some additional conditions on predicate liftings interpreting effect types, it would imply completeness.

\bibliographystyle{plainurl}% the mandatory bibstyle
\bibliography{biblio}

\appendix 
% !TEX root =main.tex

\section{Proofs of \cref{sect:results}}\label{app:lambda-results}
% We prove  the monadic progress and subject reduction properties for our example, which, as shown in \cref{sect:monadic-ty}, imply monadic soundness. 
%
% The proofs need the following standard inversion and canonical forms lemmas on the type-and-effect system. 

\begin{lemma}[Inversion]\label[lemma]{lem:inversion}\ 
\begin{enumerate}
%\item \label{lem:inversion:var} If $ \IsWFExp{\Gamma}{\x}{\T}$,
%\item \label{lem:inversion:fun} If $ \IsWFExp{\Gamma}{\RecFun{\f}{\x}{\e}}{\T}$,
\item \label{lem:inversion:app} If $ \IsEffGroundExp{\App{\ve_1}{\ve_2}}{\T}{\eff}$, then
$ \IsWFGroundExp{\ve_1}{\FunType{\T_1}{\eff}{\T}}$ and $\IsWFGroundExp{\ve_2}{\T_2}$ and $\SubT{\T_2}{\T_1}$.
\item \label{lem:inversion:op} If $\IsEffGroundExp{\opg(\ve_1,\ldots,\ve_n)}{\T}{\eff}$, then $\eff=\{\opg\}$ and 
$\fun{\opg}{\T_1\ldots\T_n}{\T}$ and $\IsWFGroundExp{\ve_i}{\T'_i}$ with $\SubT{\T'_i}{\T_i}$ for all $i\in 1..n$.
\item \label{lem:inversion:ret} If $ \IsEffGroundExp{\Ret\ve}{\T}{\eff}$, then 
$\eff=\{\eZero\}$ and 
$\IsWFGroundExp{\ve}{\T}$.
\item \label{lem:inversion:do} If $\IsEffGroundExp{\Do\x{\e_1}{\e_2}}{\T}{\eff} $, then 
 $\IsEffGroundExp{{\e_1}}{\T_1}{\eff_1}$ and  $  \IsEffExp{\TVar{\T'_1}{\x}}{{\e_2}}{\T}{\eff_2}$ and  $\SubT{\T_1}{\T'_1}$ and 
$\eff=\EComp{\eff_1}{\eff_2}$.
%\item \label{lem:inversion:try} If $\IsEffGroundExp{\Try\e{\seq\cc}{\x}{\e'}}{\T}{\eff}$, then 
%\begin{itemize}
%\item
%$\IsEffGroundExp{\e}{\T'}{\eff'}$ and $\IsEffExp{\emptyset}{\TVar{\T'}{\x}}{\e'}{\T''}{\eff''}$ and  $\SubType{}{\T''}{\T}$ and
%\item for all $i \in 1..n$, $\IsEffGroundCC{\T}{\eff''}{\cc_i}{\T_i}{\eff_i}$ and $\SubType{}{\T_i}{\T}$ and
%\item $\eff=\EComp{\eff''}{\EComp{\CatchEff{\eff'}{\seq\cc}}{\bigvee_{i\in 1..n}\eff_i}}$
%\end{itemize}
%\item \label{it:cc} If $\IsEffGroundCC{\T'}{\eff'}{\CC{\T_\x}{\m}{\seq\X}{\x\,\seq\x}{\e}}{\T}{\eff}$, then 
%\begin{itemize}
%\item $\typeof{}{\T_\x}{\sig}$  and
%\item $\mkind(\sig,\m)=\mgc$ and  $\mtypeff(\sig,\m)=\TMeth{\Ext{\seq\X}{\seq{\UT}}}{\T_1\ldots\T_n}{\TEff{\T_\m}{\eff_\m}}$ and
%\item $\IsEffExp{\Ext{\seq\X}{\seq{\UT}}}{\TVar{\T_\x}{\x}, \TVars{\T}{\x},\TVar{\FunType{{\T_\m}}{\TEff{\T'}{\eff'}}}{\cont}} {\e} {\T} {\eff}$ and
%\item $ \fv(\T)\cup\fv(\eff)\subseteq\dom(\TEnv)$
%\end{itemize}
 \end{enumerate}
\end{lemma}

\begin{lemma}[Canonical Forms]\label[lemma]{lem:canonical}\
\begin{enumerate}
\item \label{lem:canonical:fun} 
If $\IsWFGroundExp{\ve}{\FunType{\T'}{\eff}{\T}}$, then $\ve=\RecFun{\f}{\x}{\e}$.
%\begin{itemize}
%\item
%$\IsEffGroundExp{\e}{\T'}{\eff'}$ and $\IsEffExp{\emptyset}{\TVar{\T'}{\x}}{\e'}{\T''}{\eff''}$ and  $\SubType{}{\T''}{\T}$ and
%\item for all $i \in 1..n$, $\IsEffGroundCC{\T}{\eff''}{\cc_i}{\T_i}{\eff_i}$ and $\SubType{}{\T_i}{\T}$ and
%\item $\eff=\EComp{\eff''}{\EComp{\CatchEff{\eff'}{\seq\cc}}{\bigvee_{i\in 1..n}\eff_i}}$
%\end{itemize}
%\item \label{it:cc} If $\IsEffGroundCC{\T'}{\eff'}{\CC{\T_\x}{\m}{\seq\X}{\x\,\seq\x}{\e}}{\T}{\eff}$, then 
%\begin{itemize}
%\item $\typeof{}{\T_\x}{\sig}$  and
%\item $\mkind(\sig,\m)=\mgc$ and  $\mtypeff(\sig,\m)=\TMeth{\Ext{\seq\X}{\seq{\UT}}}{\T_1\ldots\T_n}{\TEff{\T_\m}{\eff_\m}}$ and
%\item $\IsEffExp{\Ext{\seq\X}{\seq{\UT}}}{\TVar{\T_\x}{\x}, \TVars{\T}{\x},\TVar{\FunType{{\T_\m}}{\TEff{\T'}{\eff'}}}{\cont}} {\e} {\T} {\eff}$ and
%\item $ \fv(\T)\cup\fv(\eff)\subseteq\dom(\TEnv)$
%\end{itemize}
 \end{enumerate}
\end{lemma}

\begin{lemma}[Progress for application]\label{theo:red-progress}
If $\IsEffGroundExp{\App{\ve_1}{\ve_2}}{\T}{\eff}$, then 
$\App{\ve_1}{\ve_2}\red\e$ for some $\e$.
\end{lemma}
\begin{proof}
From \cref{lem:inversion:app} of \cref{lem:inversion} we have
$ \IsWFGroundExp{\ve_1}{\FunType{\T_1}{\eff}{\T}}$ for some $\T_1$. From \cref{lem:canonical:fun} of \cref{lem:canonical} $\ve_1$ is $\RecFun{\f}{\x}{\e'}$. By rule  \refToRule{app} $\App{\ve_1}{\ve_2}\purered\Subst{\Subst{\e'}{\ve_1}{\f}}{\ve_2}{\x}$.
\end{proof}

\begin{proofOf}{theo:mnd-red-progress}
By induction on the typing rules of \cref{fig:typing}.
\begin{description}
\item [\refToRule{t-app}] In this case $\e$ is $ \App{\ve_1}{\ve_2}$. From \cref{theo:red-progress} $\e\purered\e'$ for some $\e'$ and then by rule \refToRule{pure} we have $\e\red\mun(\e')$.
\item [\refToRule{t-op}] In this case $\e$ is $\opg(\ve_1,\ldots,\ve_n)$ and by \cref{lem:inversion:op} of \cref{lem:inversion} we have 
$\fun{\opg}{\T_1\ldots\T_n}{\T}$ and $\IsWFGroundExp{\ve_i}{\T'_i}$ with $\SubT{\T'_i}{\T_i}$ for all $i\in 1..n$. Therefore ${\pfun{\mrun{\opg}}{\Val^n}{\mfun\Val}}$ is defined and 
rule \refToRule{op} is applicable.
\item [\refToRule{t-ret}]  In this case $\e$ is $\Ret\ve$. 
\item [\refToRule{t-do}] In this case $\e$ is $\IsEffGroundExp{\Do\x{\e_1}{\e_2}}{\T}{\eff}$. From \cref{lem:inversion:do} of \cref{lem:inversion} we have $\IsEffGroundExp{{\e_1}}{\T_1}{\eff_1}$ and 
$\IsEffExp{\TVar{\T'_1}{\x}}{{\e_2}}{\T'}{\eff_2}$ and  $\SubT{T'_1}{\T_1} $ and $\eff'=\EComp{\eff_1}{\eff_2}$. By induction hypothesis we get that either $\e_1 = \Ret\ve$ for some $\ve\in\Val$, or 
$\e_1\red\mexpr$ for some $\mexpr$. In the first case we can apply rule \refToRule{ret}  and in the second rule \refToRule{do} of \cref{fig:monadic-red}. In both cases the reduction 
produces a monadic expression.
\end{description}
\end{proofOf}

\begin{lemma}[Arrow Subtyping]\label{lem:arrowSub}
If $\SubT{\T}{\FunType{\T_1}{\eff}{\T_2}}$, then $\T=\FunType{\T'_1}{\eff'}{\T'_2}$ and $\SubT{\T_1}{\T'_1}$ and $\SubT{\T'_2}{\T_2}$ and $\SubE{\eff'}{\eff}$.
\end{lemma}
\begin{proof}
By induction on the derivation of $\SubT{\T}{\FunType{\T_1}{\eff}{\T_2}}$.
\end{proof}

%\begin{lemma}[Substitution]\label{lem:substitution}
%If $\IsEffExp{\Gamma,\TVar{\seq\T}{\seq\x}}{\e}{\T}{\eff}$ and  $\SubT{\seq{\T'}}{\seq{\T}}$ , then 
%%\begin{enumerate}
%%\item \label{lem:substitution:one} 
%%$\IsEffExp{\Gamma,\TVar{\seq\T'}{\seq\x}}{\e}{\T'}{\eff'}$  with $\SubTE{\TEff{\T'}{\eff'}}{\TEff{\T}{\eff}}$;
%%\item \label{lem:substitution:two} 
%$\IsWFGroundExp{\seq\val}{\seq{\T'}}$ implies $\IsEffExp{\Gamma}{\Subst\e{\seq\val}{\seq\x}}{\T'}{\eff'}$ with $\SubTE{\TEff{\T'}{\eff'}}{\TEff{\T}{\eff}}$.
%%\end{enumerate}
%\end{lemma}
\begin{proofOf}{lem:substitution}
We prove the result by induction on the derivation of $\IsEffExp{\Gamma,\TVar{\seq\T}{\seq\x}}{\e}{\T}{\eff}$ by proving simultaneously the following statement: 
\begin{quoting}
If $\IsWFExp{\Gamma,\TVar{\seq\T}{\seq\x}}{\val}{\T}$ and  $\SubT{\seq{\T'}}{\seq{\T}}$ , then 
$\IsWFGroundExp{\seq\val}{\seq{\T'}}$ implies $\IsWFExp{\Gamma}{\Subst\val{\seq\val}{\seq\x}}{\T'}$ with $\SubT{\T'}{\T}$.
\end{quoting}
By cases on the last rule applied in the derivation.
\begin{description}
\item [\refToRule{t-var}] In this case $\IsWFExp{\Gamma,\TVar{\seq\T}{\seq\x}}{\x}{\T}$ and $(\Gamma,\TVar{\seq\T}{\seq\x})(\x)=\T$. The are two cases:
% and  $\SubT{\seq{\T'}}{\seq{\T}}$ and $\IsWFGroundExp{\seq\val}{\seq{\T'}}$. 
 either $\x\in\seq\x$, say $\x=\x_i$, or $\x\not\in\seq\x$. In the first case $\T=\T_i$ and ${\Subst\x{\seq\val}{\seq\x}}=\val_i$ 
and, by weakening, $\IsWFExp{\Gamma}{\val_i}{\T'_i}$ with $\SubT{\T'_i}{\T_i}$. In the second case ${\Subst\x{\seq\val}{\seq\x}}=\x$ and $\IsWFExp{\Gamma}{\x}{\T}$ since $\Gamma(\x)=\T$.
\item [\refToRule{t-abs}]  In this case $ \IsWFExp{\Gamma,\TVar{\seq\T}{\seq\x}}{\RecFun{\f}{\x}{\e}}{\FunType{\T}{\eff}{\T'}} $
%and  $\SubT{\seq{\T'}}{\seq{\T}}$ and  $\IsWFGroundExp{\seq\val}{\seq{\T'}}$
and $\IsEffExp{\Gamma,\TVar{\seq\T}{\seq\x},\TVar{\FunType{\T}{\eff}{\T'}}{\f},\TVar{\T}{\x}}{\e}{\T''}{\eff'}$ with $\SubTE{\TEff{\T''}{\eff'}}{\TEff{\T'}{\eff}}$.
By induction hypothesis, $\IsEffExp{\Gamma,\TVar{\FunType{\T}{\eff}{\T'}}{\f},\TVar{\T}{\x}}{\Subst\e{\seq\val}{\seq\x}}{\T_1}{\eff_1}$ with
$\SubTE{\TEff{\T_1}{\eff_1}}{\TEff{\T''}{\eff'}}$. From $\SubTE{\TEff{\T_1}{\eff_1}}{\TEff{\T'}{\eff}}$ and rule \refToRule{t-abs} we get 
 $ \IsWFExp{\Gamma}{ \Subst{(\RecFun{\f}{\x}{\e})}{\seq\val}{\seq\x}}{\FunType{\T}{\eff}{\T'}} $.
 \item [\refToRule{t-app}] 
 In this case $ \IsEffExp{\Gamma,\TVar{\seq\T}{\seq\x}}{\App{\ve}{\ve'}}{\T}{\eff}$ and
 $\IsWFExp{\Gamma,\TVar{\seq\T}{\seq\x}}{\ve}{\FunType{\T'}{\eff}{\T}}$ and $\IsWFExp{\Gamma,\TVar{\seq\T}{\seq\x}}{\ve'}{\T''}$ and $\SubT{\T''}{\T'}$.
 By induction hypotheses,
  $\IsWFExp{\Gamma}{\Subst\ve{\seq\val}{\seq\x}}{\T_1}$ and $\IsWFExp{\Gamma}{\Subst{\ve'}{\seq\val}{\seq\x}}{\T_2}$ with
 $\SubT{T_1}{\FunType{\T'}{\eff}{\T}}$ and $\SubT{T_2}{\T''}$. From \cref{lem:arrowSub} $\T_1=\FunType{\T'_1}{\eff'}{\T'_2}$ with 
 $\SubT{\T'}{\T'_1}$ and $\SubT{\T'_2}{\T}$ and $\SubE{\eff'}{\eff}$. Since $\SubT{\T_2}{\SubT{\T''}{\SubT{\T'}{\T'_1}}}$, applying
 \refToRule{t-app} we get $ \IsEffExp{\Gamma}{\Subst{(\App{\ve}{\ve'})}{\seq\val}{\seq\x}}{\T'_2}{\eff'}$ with 
 $\SubTE{\TEff{\T'_2}{\eff'}}{\TEff{\T}{\eff}}$.
\item [\refToRule{t-op} and \refToRule{t-ret}] In both cases the proof follows easily 
from induction hypothesis.
\item [\refToRule{t-do}]  In this case $\IsEffExp{\Gamma,\TVar{\seq\T}{\seq\x}}{\Do\x{\e_1}{\e_2}}{\T}{\eff}$ with
$\eff=\EComp{\eff_1}{\eff_2}$ and 
$\IsEffExp{\Gamma,\TVar{\seq\T}{\seq\x}}{{\e_1}}{\T_1}{\eff_1}$ and $\IsEffExp{\Gamma,\TVar{\seq\T}{\seq\x},\TVar{\T_2}{\x}}{{\e_2}}{\T}{\eff_2} $ and 
$\SubT{\T_1}{\T_2} $. By induction hypotheses $\IsEffExp{\Gamma}{{\Subst{\e_1}{\seq\val}{\seq\x}}}{\T'_1}{\eff'_1}$ and 
$\IsEffExp{\Gamma,\TVar{\T_2}{\x}}{{\Subst{\e_2}{\seq\val}{\seq\x}}}{\T'}{\eff'_2} $ and 
 $\SubTE{\TEff{\T'_1}{\eff'_1}}{\TEff{\T_1}{\eff_1}}$ and  $\SubTE{\TEff{\T'}{\eff'_2}}{\TEff{\T}{\eff_2}}$. From $\SubT{\T'_1}{\SubT{\T_1}{\T_2}}$
 and rule \refToRule{t-do} we get $\IsEffExp{\Gamma}{\Subst{(\Do\x{\e_1}{\e_2})}{\seq\val}{\seq\x}}{\T'}{\EComp{\eff'_1}{\eff'_2}}$ with 
 $\SubTE{\TEff{\T}{\eff}}{\TEff{\T'}{\EComp{\eff'_1}{\eff'_2}}}$.
\end{description} 
\end{proofOf}

% \IsEffExp{\Gamma}{{\e_1}}{\T_1}{\eff_1} 
%  \BigSpace
%  \IsEffExp{\Gamma,\TVar{\T}{\x}}{{\e_2'}}{\T_2}{\eff_2} 
%}
%{ \IsEffExp{\Gamma}{\Do\x{\e_1}{\e_2}}{\T_2}{\EComp{\eff_1}{\eff_2}} }
%{ \SubT{\T_1}{\T}  } 

%
%\begin{lemmaa}{\ref{lem:subject-reduction} {\sf (Subject Reduction)}}
%{If $\IsEffGroundExp{\e}{\T}{\eff}$ and $\e\purered \e'$, then $\IsEffGroundExp{\e'}{\T'}{\eff'}$
%such that  $\SubTE{\TEff{\T'}{\eff'}}{\TEff\T\eff}$. }
%\end{lemmaa}
\begin{proofOf}{lem:subject-reduction}
Let $\IsEffGroundExp{\e}{\T}{\eff}$ and $\e\purered \e'$. Then $\e=\App{\ve_1}{\ve_2}$ and $\ve_1=\RecFun{\f}{\x}{\e_1}$ and $\e'=\Subst{\Subst{\e_1}{\ve_1}{\f}}{\ve_2}{\x}$. From \cref{lem:inversion:app} of \cref{lem:inversion} we get $ \IsWFGroundExp{\ve_1}{\FunType{\T'}{\eff}{\T}}$ and $\IsWFGroundExp{\ve_2}{\T''}$ and $\SubT{\T''}{\T'}$. From rule \refToRule{t-abs}  
$\IsEffExp{\TVar{\FunType{\T'}{\eff}{\T}}{\f},\TVar{\T'}{\x}}{\e_1}{\T_1}{\eff'}$ with  $\SubTE{\TEff{\T_1}{\eff'}}{\TEff{\T}{\eff}}$. 
%and from  \bpa\cref{lem:substitution:one}\epa of  \cref{lem:substitution} we get $\IsEffExp{\TVar{\FunType{\T'}{\eff}{\T}}{\f},\TVar{\T''}{\x}}{\e_1}{\T_1}{\eff_1}$  and $\SubTE{\TEff{\T_1}{\eff_1}}{\TEff{\T}{\eff}}$. Finally, 
Therefore, from  \cref{lem:substitution} we get
$\IsEffGroundExp{\Subst{\Subst{\e_1}{\ve_1}{\f}}{\ve_2}{\x}}{\T_2}{\eff_2}$  with $\SubTE{\TEff{\T_2}{\eff_2}}{\TEff{\T_1}{\eff'}}$ and so $\SubTE{\TEff{\T_2}{\eff'}}{\TEff{\T}{\eff}}$. 
\end{proofOf}

\section{Proofs of \cref{sect:handlers}} \label{app:handlers-results}
We extend the results of \cref{sect:results} to handlers. Since the reductions introduced are pure for monadic subject reduction we need 
to prove only subject reduction for the newly introduced rules.

\begin{lemma}[Inversion for handlers]\label[lemma]{lem:h-inversion}
If $\IsEffGroundExp{\With{\handler}{\e}}{\T}{\eff}$, where $\handler$ is\\
\centerline{
$\Handler{\HC{\opg_1}{\seq\x^1}{\e_1}{\mode_1}, \ldots, \HC{\opg_n}{\seq\x^n}{\e_n}{\mode_n}}  {\x}{\e_0}$}
and $ \fun{\opg_i}{\seq\T^i}{\T_i}$ for $i\in 1..n$, then
\begin{enumerate}
\item \label{lem:h-inversion:one}  $\IsEffGroundExp{\e}{\T_0}{\eff_0}$ 
\item \label{lem:h-inversion:four}  $\IsWFHandler{\emptyset}{\T_0}{\handler}{\T}{\hfilter}$ where $\hfilter=\HFilter{\cfilter_1\dots \cfilter_n}{\eff'_0}$ with $\cfilter_i=\CFilter{\opg_i}{\mode_i}{\eff_i}$
\item  \label{lem:h-inversion:three} $\FilterFunEx{\eff_0}{\hfilter}=\eff$ 
\item \label{lem:h-inversion:two}  $\IsEffExp{\TVar{\T_0}{\x}}{\e_0}{\T}{\eff'_0}$ and $\IsWFClause{\emptyset}{\T}{ \HC{\opg_i}{\seq\x^i}{\e_i}{\mode_i}}{\cfilter_i}$ where
\begin{enumerate}
\item  \label{lem:h-inversion:five}$\IsEffExp{\TVar{\seq\T^i}{\seq\x^i}}{\e_i}{\T_i}{\eff_i}$ if $\mode_i=\Continue$
\item \label{lem:h-inversion:six}$\IsEffExp{\TVar{\seq\T^i}{\seq\x^i}}{\e_i}{\T}{\eff_i}$ if $\mode_i=\Stop$
\end{enumerate}
\end{enumerate}
\end{lemma}
%
%\begin{lemmaa}{\ref{theo:h-mnd-red-progress} {\sf (Monadic Progress for handlers)}}
%{If $\IsEffGroundExp{\With{\handler}{\e}}{\T}{\eff}$, then
% $\With{\handler}{\e}\red\mexpr$ for some $\mexpr\in\mfun\Exp$.}
%\end{lemmaa}
\begin{proofOf}{theo:h-mnd-red-progress}
 By cases on $\e$ and induction on $\kw{with}$ expressions we prove that $\With{\handler}{\e}\purered\e'$ for some $\e'$.  
 \begin{itemize}
\item  If $\e$ is  $\Ret\ve$, then rule \refToRule{with-ret} of \cref{fig:pure-red-handlers} is applicable. 
\item   If $\e$ is  $\App{\ve}{\ve'}$, then by  \cref{lem:h-inversion:one} of \cref{lem:h-inversion} we have $\IsEffGroundExp{\App{\ve}{\ve'}}{\T_0}{\eff_0}$ for some $\T_0$ and $\eff_0$.
 By \cref{theo:red-progress} $\e\purered\e_1$ for some $\e_1$. By rule \refToRule{with-ctx} we get that  $\With{\handler}{\e}\purered\With{\handler}{\e_1}$. 
\item 
  If $\e$ is   $\opg(\seq\ve)$, let $\handler$ be $\Handler{\seq\hc}{\x}{\e_1}$. If ${\opg}\not\in\seq\hc$ then rule \refToRule{with-fwd} is applicable, otherwise either rule \refToRule{with-stop} or rule  \refToRule{with-continue}  is applicable. In any case $\With{\handler}{\e}\purered\e'$ for some $\e'$. 
  \item 
 If $\e$ is  $\Do\x{\e_1}{\e_2}$, then rule \refToRule{with-do} of \cref{fig:pure-red-handlers} is applicable. 
 \item 
 If $\e$ is  $\With{\handler_1}{\e_1}$, then by  \cref{lem:h-inversion:one} of \cref{lem:h-inversion} we have $\IsEffGroundExp{\With{\handler_1}{\e_1}}{\T_0}{\eff_0}$ for some $\T_0$ and $\eff_0$. By induction hypothesis $\e\purered\e'_1$ for some $\e'_1$. By rule \refToRule{with-ctx} we get that  $\With{\handler}{\e}\purered\With{\handler}{\e'_1}$.
 \end{itemize}
 Finally since $\With{\handler}{\e}\purered\e'$ by rule \refToRule{pure} we have $\With{\handler}{\e}\red\mun(\e')$.
%By \cref{lem:h-inversion:app} we have 
%$\IsEffGroundExp{\e}{\T'}{\eff'}$ and
%$\IsWFHandler{\emptyset}{\T'}{\handler}{\T}{\hfilter}$ and $\FilterFun{\eff'}{\hfilter}=\eff$.
\end{proofOf}
%
%
%Let  $\hfilter={\HFilter{\CFilter{\opg_1}{\mode_1}{\eff_1}\dots\CFilter{\opg_n}{\mode_n}{\eff_n}}{\eff}}$ and
%$\hfilter'={\HFilter{\CFilter{\opg_1}{\mode_1}{\eff'_1}\dots\CFilter{\opg_n}{\mode_n}{\eff'_n}}{\eff'}}$, {\em $\hfilter$ is a subhandler of $\hfilter'$},  
%\begin{quoting}
%$\SubH{\hfilter}{\hfilter'}$ if $\SubE{\eff}{\eff'}$ and $\SubE{\eff_i}{\eff'_i}$ for all $i\in i..n$
%\end{quoting}
%
%\begin{lemma}[Properties of $\FilterFEx\hfilter$]\label{lem:hfilter}\
%\begin{enumerate}
%\item \label{lem:hfilter:one} If $\SubE{\eff}{\eff'}$ and  $\SubH{\hfilter}{\hfilter'}$, then $\SubE{\FilterFunEx{\eff}{\hfilter}}{\FilterFunEx{\eff'}{\hfilter'}}$.
%\item \label{lem:hfilter:two} $\FilterFunEx{\EComp{\eff_1}{\eff_2}}{\HFilter{\seq\cfilter}{\eff}}\supseteq\FilterFunEx{{\eff_1}}{\HFilter{\seq\cfilter}{\eff'}}$ where $\eff'=\FilterFunEx{{\eff_2}}{\HFilter{\seq\cfilter}{\eff}}$
%\end{enumerate}
%\end{lemma}
\begin{proofOf}{lem:hfilter}
 From the definition of  $\FilterFEx{\hfilter}$ of \cref{fig:filters}, $\beta\in\FilterFunEx{\eff}{\hfilter}$ iff $\beta\in\InfConc{\FilterFun{\alpha}{\hfilter}}$ for some $\alpha\in\eff$.
 \begin{enumerate}
\item Since $\SubE{\eff}{\eff'}$ we have $\alpha\in\eff$ implies $\alpha\in\eff'$. Therefore $\beta\in\FilterFunEx{\eff}{\hfilter}$ implies $\beta\in\InfConc{\FilterFun{\alpha}{\hfilter}}$ for some $\alpha\in\eff'$ and so $\beta\in\FilterFunEx{\eff'}{\hfilter}$.
\item Let $\beta\in\FilterFunEx{{\eff_1}}{\HFilter{\seq\cfilter}{\eff'}}$. Then  $\beta\in\InfConc{\FilterFun{\alpha_1}{\HFilter{\seq\cfilter}{\eff'}}}$ for some $\alpha_1\in\eff_1$. 
Let $\seq\cfilter=\CFilter{\opg_1}{\mode_1}{\eff'_1}, \ldots, \CFilter{\opg_n}{\mode_n}{\eff'_n}$. \\
If $\alpha_1$ contains an $\opg_i$ such that $\mode_i=\Stop$, consider the first occurrence of such an operation, i.e., $\alpha_1=\EComp{\alpha}{(\Cons{\opg_i}{\alpha'})}$ where there is no occurrence of $\opg$  in $\alpha$ such that $\opg=\opg_j$ for some $j\in 1..n$ with $\mode_j=\Stop$. From the definition of  $\FilterF{\HFilter{\seq\cfilter}{\eff'}}$ of \cref{fig:filters}, we get
$\FilterFun{\alpha_1}{\HFilter{\seq\cfilter}{\eff'}}=\EComp{\FilterFun{\alpha}{\HFilter{\seq\cfilter}{\eff'}}}{\eff'_i}$ and also that 
$\FilterFun{\alpha_1}{\HFilter{\seq\cfilter}{\eff'}}=\FilterFun{\alpha_1}{\HFilter{\seq\cfilter}{\eff}}$, since we do not reach the end of $\alpha_1$. Moreover, taking any $\alpha_2\in\eff_2$ we have
$\FilterFun{\alpha_1}{\HFilter{\seq\cfilter}{\eff}}=\FilterFun{\EComp{\alpha_1}{\alpha_2}}{\HFilter{\seq\cfilter}{\eff}}$ and $\EComp{\alpha_1}{\alpha_2}\in\EComp{\eff_1}{\eff_2}$. Therefore
$\beta\in\FilterFunEx{\EComp{\eff_1}{\eff_2}}{\HFilter{\seq\cfilter}{\eff}}$.\\
Let $\alpha_1$ be such that it does not contains an $\opg_i$ with $\mode_i=\Stop$ for $i\in 1..n$. \\
If $\alpha_1$ is finite, then from the definition of  $\FilterF{\HFilter{\seq\cfilter}{\eff'}}$, we get
$\FilterFun{\alpha_1}{\HFilter{\seq\cfilter}{\eff'}}=\EComp{\infEff}{\eff'}$ for some $\infEff\in\EffSet^\infty$ such that the $j$th element of  $\infEff$ is either an $\eff'_i$ for $i\in 1..n$, if the
$j$th element of  $\alpha_1$ is an operation in $\seq\cfilter$, or $\{\opg\}$ if it is not. From $\beta\in\InfConc{(\EComp{\infEff}{\eff'})}$ we get that $\beta=\EComp{\beta_1}{\beta_2}$
with $\beta_1\in\InfConc\infEff $ and $\beta_2\in\InfConc\eff'$. Since
$\eff'=\FilterFunEx{{\eff_2}}{\HFilter{\seq\cfilter}{\eff}}$, $\beta_2\in\InfConc{\FilterFun{\alpha_2}{\HFilter{\seq\cfilter}{\eff}}}$ for some $\alpha_2\in\eff_2$. 
Therefore  $\beta\in\InfConc{(\EComp\infEff{\FilterFun{\alpha_2}{\HFilter{\seq\cfilter}{\eff}}})}=\InfConc{\FilterFun{\EComp{\alpha_1}{\alpha_2}}{\HFilter{\seq\cfilter}{\eff}}}$ and
from $\EComp{\alpha_1}{\alpha_2}\in\EComp{\eff_1}{\eff_2}$ we derive $\beta\in\FilterFunEx{\EComp{\eff_1}{\eff_2}}{\HFilter{\seq\cfilter}{\eff}}$. \\
If $\alpha_1$ is infinite, then $\FilterFun{\alpha_1}{\HFilter{\seq\cfilter}{\eff'}}=\EComp{\infEff}{\eff'}$ where $\infEff$ is an infinite string.
Since for all infinite strings $\str$ for all strings $\str'$ we have that $\str=\EComp{\str}{\str'}$, 
 taking any $\alpha_2\in\eff_2$ we have
$\FilterFun{\EComp{\alpha_1}{\alpha_2}}{\HFilter{\seq\cfilter}{\eff'}}=\FilterFun{\alpha_1}{\HFilter{\seq\cfilter}{\eff'}}$ and 
$\FilterFun{\alpha_1}{\HFilter{\seq\cfilter}{\eff'}}=\FilterFun{\alpha_1}{\HFilter{\seq\cfilter}{\eff}}$. Therefore
$\beta\in\FilterFunEx{\EComp{\eff_1}{\eff_2}}{\HFilter{\seq\cfilter}{\eff}}$.
\end{enumerate}
\end{proofOf}

%\begin{lemma}[Substitution for handlers]\label[lemma]{lem:h-substitution}
%If $\IsWFHandler{\Gamma}{\T}{\handler}{\T_1}{\hfilter_1}$ and
%$\SubT{\T'}{\T}$, then $\IsWFHandler{\Gamma}{\T'}{\handler}{\T_2}{\hfilter_2}$ such that $\SubT{T_1}{\T_2}$ and $\SubH{\hfilter_1}{\hfilter_2}$.
%\end{lemma}
%\begin{proof}
%By induction on the typing rules using \bpa\cref{lem:substitution:one}\epa of \cref{lem:substitution}.
%\end{proof}

%\begin{lemma}[Subject Reduction for handlers]\label{lem:h-subject-reduction}\
%If $\IsEffGroundExp{\With{\handler}{\e}}{\T}{\eff}$ and
%%\\
%%\centerline{
% $\With{\handler}{\e}\purered \e'$,
% %}
% then $\IsEffGroundExp{\e'}{\T'}{\eff'}$
%such that  $\SubTE{\TEff{\T'}{\eff'}}{\TEff\T\eff}$. 
%\end{lemma}
%\begin{lemmaa}{\ref{lem:h-subject-reduction}  {\sf (Subject Reduction for handlers)}}
%{If $\IsEffGroundExp{\With{\handler}{\e}}{\T}{\eff}$ and
% $\With{\handler}{\e}\purered \e'$,
% then $\IsEffGroundExp{\e'}{\T'}{\eff'}$
%such that  $\SubTE{\TEff{\T'}{\eff'}}{\TEff\T\eff}$. }
%\end{lemmaa}
\begin{proofOf}{lem:h-subject-reduction}
Let $\e$ and $\e'$ be such that $\With{\handler}{\e}\purered\e'$, where $\handler=\Handler{\hc_1, \ldots,\hc_n}  {\x}{\e_0}$ and $\hc_i=\HC{\opg_i}{\seq\x^i}{\e_i}{\mode_i}$ for $i\in 1..n$. 
By cases and induction on the rules of \cref{fig:pure-red-handlers}. 
\begin{description}
\item [\refToRule{with-do}] In this case $\e$ is $\Do{\y}{\e'_1}{\e'_2}$ and $\e'$ is $\WithLong{\seq\hc}{\y}{(\With{\handler}{\e'_2})}{\e'_1}$. 
From \cref{lem:h-inversion:one}  of \cref{lem:h-inversion} we get $\IsEffGroundExp{\Do{\y}{\e'_1}{\e'_2}}{\T_0}{\eff_0}$. From \cref{lem:inversion:do} of \cref{lem:inversion} we get 
 $\IsEffGroundExp{{\e'_1}}{\T'_0}{\eff'_1}$ and  $\IsEffExp{\TVar{\T''_0}{\y}}{{\e'_2}}{\T_0}{\eff'_2}$ with $\SubT{\T'_0}{\T''_0}$  and  
$\eff_0=\EComp{\eff'_1}{\eff'_2}$.
By weakening and  \cref{lem:h-inversion:one}  of \cref{lem:h-inversion} we get $\IsWFHandler{\TVar{\T''_0}{\y}}{\T_0}{\handler}{\T}{\hfilter}$ and
from $\IsEffExp{\TVar{\T''_0}{\y}}{{\e'_2}}{\T_0}{\eff'_2}$ applying rule \refToRule{t-with}, we derive ${\IsEffExp{\TVar{\T''_0}{\y}}{\With{\handler}{\e'_2}}{\T}{\FilterFunEx{{\eff'_2}}{\hfilter}}}$. 
Let $\handler'$ be  $\Handler{\hc_1, \ldots,\hc_n}  {\y}{(\With{\handler}{\e'_2})}$. From \cref{lem:h-inversion:two}  of \cref{lem:h-inversion}  and rule \refToRule{t-handler}
we derive $\IsWFHandler{\emptyset}{\T''_0}{\handler'}{\T}{\hfilter'}$  where $\hfilter'=\HFilter{\seq\cfilter}{\FilterFunEx{{\eff'_2}}{\hfilter}}$.
From  $\IsEffGroundExp{{\e'_1}}{\T''_0}{\eff'_1}$ and rule \refToRule{t-with} we get ${\IsEffGroundExp{\With{\handler'}{\e'_1}}{\T}{{\FilterFunEx{{\eff'_1}}{\hfilter'}}}}$.
Finally, from \cref{lem:hfilter:two} of \cref{lem:hfilter} and  \cref{lem:h-inversion:three}  of \cref{lem:h-inversion}   we have that $\SubE{\FilterFunEx{{\eff'_1}}{\hfilter'}}{\FilterFunEx{\EComp{\eff'_1}{\eff'_2}}{\hfilter}}=\eff$.

\item [\refToRule{with-ret}] In this case $\e$ is $\Ret\ve$ and $\e'$ is $\Do\x{\Ret\ve}{\e_0} $. From \cref{lem:h-inversion:one}  of \cref{lem:h-inversion} we get $\IsEffGroundExp{\Ret\ve}{\T_0}{\eff_0}$ and from \cref{lem:inversion:ret} of \cref{lem:inversion}, $\eff_0=\{\eZero\}$. By \cref{lem:h-inversion:two,lem:h-inversion:three}  
of \cref{lem:h-inversion} we get  $\IsEffExp{\TVar{\T_0}{\x}}{\e_0}{\T}{\eff'_0}$  and $\FilterFunEx{\{\eZero\}}{\hfilter}=\eff$. From the definition of $\FilterFEx\hfilter$,
see \cref{fig:filters}, $\eff=\eff'_0$.\\ 
Consider now $\Do\x{\Ret\ve}{\e_0}$. From $\IsEffGroundExp{\Ret\ve}{\T_0}{\eZero}$ and $\IsEffExp{\TVar{\T_0}{\x}}{\e_0}{\T}{\eff}$ with rule \refToRule{t-do} we get 
$\IsEffGroundExp{\Do\x{\Ret\ve}{\e_0}}{\T}{\eff}$ since ${\EComp{\eZero}{\eff}}=\eff$.

\item [\refToRule{with-continue}] In this case $\e$ is $\opg_i(\seq\ve)$ and  $\mode_i=\Continue$ (for some $i\in 1..n$) and  $\e'$ is $ \Do{\x}{\Subst{\e_i}{\seq\ve}{\seq\x^i}}{\e_0}$. 
From \cref{lem:h-inversion:one}  of \cref{lem:h-inversion} we get $\IsEffGroundExp{\opg_i(\seq\ve)}{\T_0}{\eff_0}$ and from \cref{lem:inversion:op} of \cref{lem:inversion}, 
$\eff_0=\{\opg_i\}$ and $\fun{\opg_i}{\T^i_1\ldots\T^i_m}{\T_0}$ and $\IsWFGroundExp{\ve_j}{\T'^i_j}$  
with $\SubT{\T'^i_j}{\T^i_j}$ for all $j\in 1..m$.
By \cref{lem:h-inversion:two,lem:h-inversion:three,lem:h-inversion:five}  
of \cref{lem:h-inversion} we get  $\IsEffExp{\TVar{\T_0}{\x}}{\e_0}{\T}{\eff'_0}$  and $\IsEffExp{\TVar{\seq\T^i}{\seq\x^i}}{\e_i}{\T_0}{\eff_i}$ and $\eff=\FilterFunEx{\{\Cons{\opg_i}{\eZero}\}}{\hfilter}=\EComp{\eff_i}{\eff'_0}$
by the definition $\FilterF{\hfilter}$ of   \cref{fig:filters}.   \\
Consider now $\Do{\x}{\Subst{\e_i}{\seq\ve}{\seq\x^i}}{\e_0}$. From  $\IsEffExp{\TVar{\seq\T^i}{\seq\x^i}}{\e_i}{\T_0}{\eff_i}$  and $\IsWFGroundExp{\ve_j}{\T'^i_j}$  
with $\SubT{\T'^i_j}{\T^i_j}$ for all $j\in 1..m$ and \cref{lem:substitution} we get $\IsEffGroundExp{{\Subst{\e_i}{\seq\ve}{\seq\x^i}}}{\T'_0}{\eff_i}$ with $\SubT{\T'_0}{\T_0}$. 
 Therefore
from $\IsEffExp{\TVar{\T_0}{\x}}{\e_0}{\T}{\eff'_0}$ and rule \refToRule{t-do} we get 
$\IsEffGroundExp{\Do{\x}{\Subst{\e_i}{\seq\ve}{\seq\x^i}}{\e_0}}{\T}{\EComp{\eff_i}{\eff'_0}}$, which proves the result.
%
%\bigskip
%From 
% \bpa\cref{lem:substitution:one}\epa of  \cref{lem:substitution} we get  $\IsEffExp{\TVar{\T'_0}{\x}}{\e_0}{\T'}{\eff''_0}$ for some $\T'$ and $\eff''_0$ such that $\SubTE{\TEff{\T'}{\eff''_0}}{\TEff{\T}{\eff'_0}}$. 
% 
% From rule \refToRule{t-do}
% we derive $\IsEffGroundExp{\Do{\x}{\Subst{\e_i}{\seq\ve}{\seq\x^i}}{\e_0}}{\T'}{\EComp{\eff_i}{\eff''_0}}$ and since $\SubE{\eff''_0}{\eff'_0}$ and $\cdot$ is monotone we have  
%$\SubTE{\TEff{\T'}{\eff''_0}}{\TEff{\T}{\eff}}$.

\item [\refToRule{with-stop}] In this case $\e$ is $\opg_i(\seq\ve)$ and  $\mode_i=\Stop$ (for some $i\in 1..n$) and  $\e'$ is ${\Subst{\e_i}{\seq\ve}{\seq\x^i}}$. 
From \cref{lem:h-inversion:one}  of \cref{lem:h-inversion} we get $\IsEffGroundExp{\opg_i(\seq\ve)}{\T_0}{\eff_0}$ and from \cref{lem:inversion:op} of \cref{lem:inversion}, 
$\eff_0=\{\opg_i\}$ and $\fun{\opg_i}{\T^i_1\ldots\T^i_m}{\T_0}$ and $\IsWFGroundExp{\ve_j}{\T'^i_j}$  
with $\SubT{\T'^i_j}{\T^i_j}$ for all $j\in 1..m$. 
By \cref{lem:h-inversion:two,lem:h-inversion:three,lem:h-inversion:six}  
of \cref{lem:h-inversion} we get  $\IsEffExp{\TVar{\T_0}{\x}}{\e_0}{\T}{\eff'_0}$  and $\IsEffExp{\TVar{\seq\T^i}{\seq\x^i}}{\e_i}{\T}{\eff_i}$ and $\eff=\FilterFunEx{\{\Cons{\opg_i}{\eZero}\}}{\hfilter}={\eff_i}$
by the definition $\FilterF{\hfilter}$ of   \cref{fig:filters}. From \cref{lem:substitution} we get that  $\IsEffGroundExp{\Subst{\e_i}{\seq\ve}{\seq\x^i}}{\T'}{\eff_i}$ 
with $\SubTE{\TEff{\T'}{\eff'}}{\TEff{\T}{\eff}}$.

\item [\refToRule{with-fwd}] In this case $\e$ is $\opg(\seq\ve)$ with $\opg\not=\opg_i$ for all $i\in 1..n$ and  $\e'$ is $ \Do{\x}{\opg(\seq\ve)}{\e_0} $. The proof is similar to the case of rule
\refToRule{with-continue}.

\item [\refToRule{with-ctx}] In this case $\e'$ is $\With{\handler}{\e''}$ where $\e$ is such that $\e\purered \e''$. By \cref{lem:h-inversion:one} of \cref{lem:h-inversion} we
have that $\IsEffGroundExp{\e}{\T_0}{\eff_0}$. From $\e\purered \e''$ we derive that that $\e$ is either $\With{\handler_1}{\e_1}$  or $\App{\ve_1}{\ve_2}$. 
In the first case by induction hypothesis and in the second by \cref{lem:subject-reduction} we 
get   $\IsEffGroundExp{\e''}{\T'_0}{\eff'_0}$  and $\SubTE{\TEff{\T'_0}{\eff'_0}}{\TEff{\T_0}{\eff_0}}$. By \cref{lem:h-inversion:four} of \cref{lem:h-inversion} we have
$\IsWFHandler{\emptyset}{\T_0}{\handler}{\T}{\hfilter}$ and  so from \cref{lem:h-subject-reduction} we derive $\IsWFHandler{\emptyset}{\T'_0}{\handler}{\T'}{\hfilter'}$ for some
$\T'$ such that $\SubT{\T'}{\T}$. Therefore, applying rule \refToRule{t-with} we get 
$\IsEffGroundExp{\With{\handler}{\e''}}{\T'}{\eff'}$ where $\eff'=\FilterFun{\eff'_0}{\hfilter'}$. 
 By \cref{lem:h-inversion:three} of \cref{lem:h-inversion} we have  $\FilterFunEx{\eff_0}{\hfilter}=\eff$  and from $\SubT{{\eff_0}}{{\eff'_0}}$ and $\SubH{\hfilter}\hfilter'$ 
and \cref{lem:hfilter:one} of \cref{lem:hfilter} we get  $\eff'=\SubE{\FilterFun{\eff'_0}{\hfilter'}}{\eff}$.Therefore $\SubTE{\TEff{\T'}{\eff'}}{\TEff{\T}{\eff}}$.
\end{description}
\end{proofOf}

\end{document}